\documentclass{CSML}

\newcommand{\dOi}{13(1:2)2017}
\lmcsheading{}{1--84}{}{}{Mar.~23, 2015}{Jan.~24, 2017}{}


\newcounter{customInlineItem}
\newenvironment{custominlinelist}{%
\setcounter{customInlineItem}{0}
\def\item{\stepcounter{customInlineItem}(\arabic{customInlineItem})~}%
}{}

\usepackage[pdftex,bookmarks,breaklinks,hidelinks]{hyperref}
\usepackage{latexsym}
\usepackage{latexsym}
\usepackage{amsmath}
\usepackage{amssymb}
\usepackage{stmaryrd}
\usepackage[PostScript=Rokicki,hug]{diagrams}
\usepackage{color}

\renewcommand{\epsilon}{\varepsilon}
\renewcommand{\phi}{\varphi}

\newcommand{\cO}{\mathcal O}
\newcommand{\cS}{\mathcal S}
\newcommand{\cP}{\mathcal P}
\newcommand{\cH}{\mathcal H}
\newcommand{\cL}{\mathcal L}
\newcommand{\cV}{\mathcal V}

\newcommand{\type}{\!:\!}

\newcommand{\R}{\mathbb{R}}

\newcommand{\I}{\mathbb{I}}

\newcommand{\dbdownarrow}{\rlap{\raise.25ex\hbox{$\shortdownarrow$}}\raise-.25ex\hbox{$\shortdownarrow$}}
\newcommand{\dda}{\rlap{\raise.25ex\hbox{$\shortdownarrow$}}\raise-.25ex\hbox{$\shortdownarrow$}}
\newcommand{\dua}{\rlap{\raise-.25ex\hbox{$\shortuparrow$}}\raise.25ex\hbox{$\shortuparrow$}}
\newcommand{\dbuparrow}{\rlap{\raise-.25ex\hbox{$\shortuparrow$}}\raise.25ex\hbox{$\shortuparrow$}}

\newcommand{\dsup}{\mathop{\bigvee{}^{^{\,\makebox[0pt]{$\scriptstyle\uparrow$}}}}}

\newcommand{\myconv}{\mathop{\rm conv}}
\newcommand{\id}{\mathop{\rm id}}

\newcommand{\cchull}[1]{\overline{{\rm conv} #1}}

\newcommand{\diamondplus}{\mathop{\Diamond\mkern-13.9mu\raise.22ex\hbox{$+$}}}
\newcommand{\diamonddot}{\mathop{\Diamond\mkern-9.5mu\raise.2ex\hbox{$\cdot$}}\,}

\renewcommand{\diamonddot}{\cdot_{_{_{\mbox{}\!\!\!\!S}}}\kern-.1em}
\renewcommand{\diamondplus}{+_{_{_{\mbox{}\!\!\!S}}}\kern-.2em}
\renewcommand{\odot}{\cdot_{_{_{\mbox{}\!\!\!\!P}}}\kern-.2em}
\renewcommand{\oplus}{+_{_{_{\mbox{}\!\!\!P}}}\kern-.3em}
\renewcommand{\boxdot}{\cdot_{_{_{\mbox{}\!\!\!\!H}}}\kern-.2em}
\renewcommand{\boxplus}{+_{_{_{\mbox{}\!\!\!H}}}\kern-.3em}

\DeclareMathOperator{\ua}{\uparrow\!}
\DeclareMathOperator{\da}{\downarrow\!}

\newcommand{\cDw}{\mathcal{D}_{\omega}}

\newcommand{\cMwp}{\mathcal{M}^+_{\omega}}
\newcommand{\cPwp}{\mathcal{P}^+_{\omega}}

\newcommand{\oRp}{\overline{\mathbb{R}}_{\mbox{\tiny +}}}
\newcommand{\RP}{\mathbb{R}_{\mbox{\tiny +}}}
\newcommand{\Rpp}{\mathbb{R}_{\mbox{\tiny $>\!\!0$}}}
\newcommand{\Rp}{\RP}
\newcommand{\eqdef}{=_{\mbox{\tiny $\mathit{def}$}}}

\newcommand{\un}[1]{\underline{#1}}
\newcommand{\ov}[1]{\overline{#1}}

\newcommand{\half}{1/2}

\newcommand{\norm}[1]{|\!| #1 |\!|}

\newcommand{\display}{:}

\newcommand{\iso}{\cong}

\newcommand{\KS}{Kegelspitze}
\newcommand{\KSs}{Kegelspitzen}

\newcommand{\Cone}{{\sf{Cone}}}
\newcommand{\dCone}{{\sf{d\mbox{-}Cone}}}
\newcommand{\Lin}{{\mathcal L}_{\rm lin}}
\newcommand{\SubL}{{\mathcal{L}}_{\rm sub}}
\newcommand{\SuperL}{{\mathcal{L}}_{\rm sup}}

\newcommand{\llem}{ \ll_{\mathrm{EM}}}
\newcommand{\leqem}{ \leq_{\mathrm{EM}}}

\newcommand{\sub}{L}

\newcommand{\sideP}{\sideset{}{_P}}

\newcommand{\nonor}{\cup}
\newcommand{\nonand}{\cap}
\newcommand{\bignonor}{\bigcup}

\newcommand{\ev}{\mathrm{ev}}
\newcommand{\PT}{\mathrm{PT}}

\newcommand{\EXT}{\mathrm{EXT}}
\newcommand{\RES}{\mathrm{RES}}

\newcommand{\CCSA}{\rm{CCSA}}

\newcommand{\mycut}[1]{}

\newcommand{\makered}[1]{#1}
\newcommand{\makeblue}[1]{#1}

 \theoremstyle{plain}\newtheorem{standardconstr}[thm]{Standard Construction}

\begin{document}

\title{Mixed powerdomains for probability and nondeterminism}

\author[K.~Keimel]{Klaus Keimel\rsuper a}
\address{{\lsuper{a}}Fachbereich Mathematik, Technische Universit\"{a}t Darmstadt} 

\author[G.~D.~Plotkin]{Gordon D. Plotkin\rsuper b}
\address{{\lsuper{b}}Laboratory for Foundations of Computer Science, School of Informatics, University of Edinburgh}




\begin{abstract}
We consider mixed powerdomains combining ordinary nondeterminism and probabilistic nondeterminism. We characterise them as free algebras for suitable 
(in)equation-al theories; we establish functional representation theorems; and we show equivalencies between state transformers and appropriately healthy predicate 
transformers. The extended nonnegative reals serve as `truth-values'. As usual with powerdomains, everything comes in three flavours: lower, upper, and order-convex. 
The powerdomains are suitable convex sets of subprobability valuations, corresponding to resolving nondeterministic choice before probabilistic choice. Algebraically this 
corresponds to the probabilistic choice operator distributing over the nondeterministic choice operator. (An alternative approach to combining the two forms of 
nondeterminism would be to resolve probabilistic choice first, arriving at a  domain-theoretic version of random sets. However, as we also show, the algebraic approach 
then runs into difficulties.)

Rather than working directly with valuations, we take a domain-theoretic functional-analytic approach, employing domain-theoretic abstract convex sets called 
Kegelspitzen; these are equivalent to the abstract probabilistic algebras of Graham and Jones, but are more convenient to work with.  So we define power Kegelspitzen, 
and consider free algebras, functional representations, and predicate transformers. To do so we make use of previous work on domain-theoretic  cones (d-cones), with the 
bridge between the two of them being  provided by a  free d-cone construction on  Kegelspitzen. 
\end{abstract}

\maketitle



\section{Introduction}

In this paper we investigate mixed powerdomains combining ordinary and  probabilistic nondeterminism. These can be defined generally as free algebras over dcpos (directed complete posets). The algebraic laws we consider in this regard are  for the binary choice operators  $\nonor$  and $+_r$ of ordinary and probabilistic nondeterminism as well as   a constant for nontermination (where $x +_r y$ expresses a choice of $x$ with probability $r$ versus one of $y$ with probability $1-r$) together with an axiom
  to the effect that probabilistic choice distributes over ordinary nondeterministic choice, viz.:
 \[x +_r (y \nonor z) = (x +_r y) \nonor (x +_r z)\]
 (see below for axioms for  ordinary and probabilistic nondeterminism, or, for example,~\cite{HPP}).
We characterise the free algebras as suitable convex sets of subprobability valuations in the case of domains (continuous dcpos). We do this for all three domain-theoretic notions of ordinary nondeterminism, viz.\ lower (or Hoare), upper (or Smyth), and convex (or Plotkin), though in the last case we need an additional assumption, that the domains are coherent.

We further give suitable notions of predicates and predicate transformers, obtaining a dual correspondence between predicate transformers and (mixed) nondeterministic functions (i.e., Kleisli category morphisms). The relevant notions of predicate use $\oRp$, the domain of the  non-negative reals \makered{extended with a point at infinity}, or its convex powerdomain, as `truth-values.' 
Our results on predicate transformers are obtained via functional characterisations of the mixed powerdomains, and are again  obtained for all three notions of ordinary nondeterminism. As before these results obtain generally for domains except in the convex case where, additionally, coherence is again required.

In previous joint work with Regina Tix~\cite{TKP09,KP} 
  based on Tix's Ph.D.\ thesis~\cite{TIX}, we carried out a similar
programme for ordinary and so-called `extended' probabilistic
computation where valuations  take  values in $\oRp$ rather than
$[0,1]$. The method we used was a kind of domain-theoretic functional
analysis where, instead of working directly with domains of
valuations, one works with an abstract domain-theoretic notion of
cone, called a d-cone. We investigated powercones, which embody
notions of non-determinism at the cone-theoretic level, and then
applied our results to cones of valuations. The powercones were
  shown to be the free cones with a semilattice operation over which the cone operations distribute, with a further requirement that the semilattice be a join-semilattice in the lower case, and a meet-semilattice in the upper case;  it immediately follows, although not remarked in~\cite{TKP09}, that the  valuation powercones provide corresponding free constructions on  domains.
  Predicate transformers and functional representations were also first investigated at the cone-theoretic level, with the results being again applied to cones of valuations.

We proceed analogously here, but replacing cones with a suitable domain-theoretic notion of  abstract convex space, termed a \KS
\footnote{The German word \KS\  means `tip of a cone', suggesting a convex set obtained by cutting off the top of a cone.}.
By embedding \KSs\ in d-cones we are able to make use of our previous results, thereby avoiding a good deal of work. This approach works particularly well when characterising the mixed powerdomains, but less well when considering functional representations. We do obtain strong enough abstract results for their intended application. \makered{However the general results require assumptions (taken from previous work) on the cones in which the \KSs\ are embedded, rather than  natural assumptions directly concerning the \KSs\ themselves.  One wonders to what extent one can succeed with more natural assumptions, and, indeed, to what extent one can proceed without making use of d-cones. }

Several other authors have previously considered the combination of ordinary and
probabilistic nondeterminism, making use of sets of distributions. In a domain-theoretic context,
the pioneers  were  the Oxford Programming Research
Group~\cite{MMS,MM01a,MM01b,MM05}. 
 That work was restricted to the
case of countable discrete domains, as was that of Ying~\cite{Yin03}. Later, 
Mislove~\cite{Mis00} defined
mixed powerdomains 
for all three notions of ordinary nondeterminism  over
  continuous domains required to be coherent in the convex
  case. Mislove also considered nondeterministic powerdomains over
  `abstract probabilistic algebras'. With the addition of a bottom element, these are the same as the identically named algebras of Graham and Jones~\cite{graham,jones} and equivalent to our \KSs; however \KSs\ seem more in the spirit of domain theory, as we discuss below.


Goubault-Larrecq \cite{GL07, GL08, GL10, GL12} worked at a
topological level, considering all three notions of ordinary
nondeterminism combined with various classes of valuations:  all
valuations, subprobability valuations, and
probability valuations, but without explicit consideration of the
algebraic structures involved. He established functional representation
results under quite weak assumptions on the underlying spaces.
When
specialised to domains and subprobability valuations, his results
correspond to our Corollaries \ref{th:lowerdomain}, \ref{th:upperdomain}, and
\ref{th:convexdomain}. He worked directly with the valuation spaces
rather than, as we do, making use of abstract structures such as cones
and barycentric algebras.
 In~\cite{Beau1,Beau2}  Beaulieu worked algebraically; his results
 include free constructions of algebras satisfying the above laws over
 sets and partial orders, but not domains.

 There has been some discussion of other ways to combine
  nondeterminism with probability. Categorical distributive laws
     provide a standard means  
  of showing the composition of two monads form a third 
  (see~\cite{Mac}). However there is no  such law enabling one
  to compose the monad of ordinary nondeterminism with that of
  probabilistic nondeterminism --- see 
 the Appendix of \cite{VarWin}. For this reason, Varacca and
  Winskel~\cite{Var02,Var03,VarWin} reject, or weaken, 
  one of the axioms of (extended)
  probabilistic nondeterminism, viz.\ the distributivity law $(r + s)x
  = rx + sx$, 
  and consider certain `indexed valuations' in place of the more usual ones. 
  As shown by  Varacca in~\cite[Chapter 4]{Var03} this approach applies to domains, where indexed valuations come in three flavours: Hoare, Smyth, and Plotkin. Categorical distributivity laws are obtained in several cases, and a freeness result (with the above equational distributive law) is given in the case of the combination of Hoare indexed valuations and the Hoare powerdomain.

  In~\cite{GL07a,GL10}, and see too~\cite{GK11}, 
  Goubault-Larrecq also combined the two types
  of monads in a different order, considering the probabilistic
  powerdomains over the nondeterministic powerdomains; however no
  algebraic aspects were discussed. 
 This  approach  is in the spirit of
  what is known under the name of \emph{random sets} in
  probability theory. The approach can be thought of as resolving
  probabilistic choice before nondeterministic choice; in contrast,
  approaches employing sets of distributions can rather be thought of
  as resolving nondeterministic  choice first.

  Algebraically, the above distributive law corresponds to resolving
  nondeterministic choice first. The other distributive law 
 \[x \nonor (y +_r z) = (x \nonor y) +_r (x \nonor z)\]
       corresponds to first resolving probabilistic choice. As pointed out in~\cite{Mis03} this law leads to some odd consequences when combined with the other laws, particularly the idempotence of nondeterministic choice. In Appendix~\ref{appbad} we show that the equational theory consisting of the laws for nondeterministic and probabilistic choice together with this distributive law is equivalent to that of join-distributive bisemilattices~\cite{Rom80}, i.e., of  algebras with two semilattice operations called join and meet, with join distributing over meet. It therefore has no quantitative content.

An algebraic treatment of this combination of the two forms of    nondeterminism would therefore have to weaken some law. One natural possibility is to drop the idempotence of nondeterministic choice (
this was done in a process   calculus-oriented context in~\cite{YL92, DG08}). 
There is a natural `finite random sets' functor supporting models of this weaker theory, namely $\cDw\circ\cPwp$ where $\cDw$ is the finite probability distributions monad, and $\cPwp$ is the finite non-empty sets monad. The barycentric structure is evident, and the nondeterministic choice operations $\cup_X\colon \cDw\cPwp(X)^2 \to \cDw\cPwp(X)$ are given by\display
       \[(\sum_{i \,=\, 1,\ldots,m} \alpha_i x_i) \,\cup_X\, (\sum_{i \,=\, 1,\dots,n} \beta_j y_j) = \sum_{{\mbox{$\substack{\scriptsize i \,=\, 1,\dots,m\\ j \,=\, 1,\dots,n}$}}}\alpha_i\beta_j (x_i \cup  y_j )\]
  Unsurprisingly, these finite random set algebras do not provide the free algebras for the weaker theory. More surprisingly, perhaps, they do not provide the free algebras for \emph{any} equational theory over the relevant signature; this too is shown in Appendix~\ref{appbad}. As another point along these lines, we recall a result of Varacca~\cite[Proposition 3.1.3]{Var03} that there is no distributive law of $\cPwp$ over $\cDw$. Overall, it seems hard to see how there can be any satisfactory algebraic treatment of the combination of probability and nondeterminism in which probabilistic choice is resolved first.

In Section~\ref{Keg} we develop the theory of \KSs, beginning with a
notion of (ordered) barycentric algebra. This notion is based on the
equational theory of the barycentric operations $rx + (1-r)y$ (for real
numbers $r$ between $0$ and $1$) on convex sets in vector spaces. There
is an extensive relevant literature, which we survey. We need this
notion augmented with a compatible partial order and, in addition,
with a distinguished element $0$. Finally, specialising to directed
complete partial orders and Scott-continuous operations, we introduce the
central notion of \KSs, axiomatising subprobabilistic
powerdomains. In order to relate these structures to our previous
work on cones, we prove embedding theorems at the various levels, and
then establish the preservation of crucial properties.
In the case of \KSs, the embedding   theorem is Theorem~\ref{prop:KSembed} and the property-preserving results are those of 
Propositions~\ref{prop:continuous},~\ref{prop:KSembedding1} and ~\ref{prop:KSembedding2}, and Corollary~\ref{cor:lawsoncompact}; the properties preserved include continuity and coherence.

In Section~\ref{Powerkeg} we define, and give universal algebraic characterisations of, the various mixed powerdomains, first doing the same for suitable notions of power  \KS. The universal characterisations of the free power \KSs\ are given in Theorems~\ref{Huni},~\ref{Suni}, and~\ref{Puni} (one for each notion of nondeterminism). The universal characterisations of the mixed powerdomains then follow, and are given in Corollaries~\ref{HVuni},~\ref{SVuni}, and~\ref{PVuni}; the first two hold for any domain, and the third holds for any coherent domain.

In Section~\ref{funcrep}, we consider functional
representations. The three functional representations of power \KSs\
are given by Theorems~\ref{th:lowerKS},~\ref{th:upperKS},
and~\ref{th:convexKS}; they largely follow straightforwardly from the
corresponding results for cones in~\cite{KP}. \makered{The
  corresponding three functional representations of 
  the mixed powerdomains over domains
are given by Corollaries~\ref{th:lowerdomain}, ~\ref{th:upperdomain},
 and~\ref{th:convexdomain} and are derived from the corresponding results for \KSs.}

In
Section~\ref{pred-tran} we consider predicate transformers \makered{for domains}, showing
the equivalence of `state transformers', i.e., Kleisli maps, and
`healthy' predicate transformers, viz.\  maps on predicates obeying
suitable conditions. The conditions and equivalences follow from the functional representation theorems \makered{for domains}, and
are given by Corollaries~\ref{th:lowerdomainPred}, ~\ref{th:upperdomainPred},
and~\ref{th:convexdomainPred}. 
\makered{There are related general results for \KSs; these follow from the functional representation theorems for \KSs, and are discussed briefly. 
}


%

\makered{The results given in Sections~\ref{funcrep} and~\ref{pred-tran} all make use of $\oRp$, the d-cone of the non-negative reals augmented by a point at infinity; this is unsurprising given their derivation from the corresponding results for cones. However, in the context of \KSs, it is natural to further seek results replacing that cone by $\I$, the unit interval \KS. This is done for the mixed powerdomains in Section~\ref{unit-interval}, where both functional and predicate transformer results are obtained in all three cases. 
The functional representation results are Corollaries~\ref{th:lowerdomain-unit},~\ref{th:upperdomain-unit}, and~\ref{th:convexdomain-unit}; the predicate transformer results are Corollaries~\ref{th:lowerdomainPred-unit},~\ref{th:upperdomainPred-unit}, and~\ref{th:convexdomainPred-unit}.
With additional assumptions there are  related general results for \KSs, which we discuss briefly. }



{\bf Terminology}. Throughout the paper, we assume familiarity with
 standard terminology and notation  of 
domain theory as covered in, say,~\cite{GHK}. In particular, \makered{for
partially ordered set we shortly say poset;} 
the abbreviation \emph{dcpo} stands for `directed complete partially ordered
set'; \makered{`bounded directed complete' means that every upper-bounded
directed set has a least upper bound;} the way below relation in a
dcpo $C$ is written as $\ll_C$, or 
simply $\ll$; and a \emph{domain} is a continuous dcpo, that is, a
dcpo in which, 
for every element $a$, the elements way-below $a$ form a directed
set with supremum $a$. For any subset $X$ of a poset $C$ we write
$\ua_C X$, or simply $\ua X$, for the set of elements  
above some element of $X$ (and
$\da X$ is understood analogously); similarly, we write $\dua_C X$ for
the set of 
elements way-above some element of $X$.

\makered{A \emph{sub-dcpo} is a subset $C$ of a dcpo $D$ that is closed for
suprema of directed sets, that is, 
the supremum of any directed subset of $C$ belongs to
$C$.  Sub-dcpos should not be mixed up with Scott-closed sets which 
are sub-dcpos and, in addition, lower sets. The intersection of an
arbitrary family of sub-dcpos is a 
sub-dcpo; thus, for any subset $P$ of a dcpo there is a least sub-dcpo
$P^d$ containing $P$. The elements of $P^d$ can also be obtained by
closing under directed suprema repeatedly. We will say that $P$ is
\emph{dense} in $C$, if $P^d= C$. Again, this notion of density is
different from {dense for the Scott topology}.}

We write $\Rp$ for the set of nonnegative real numbers with their usual
order, and $\oRp \eqdef \Rp\cup\{+\infty\}$ for the nonnegative real numbers
augmented by a top element $+\infty$. 
Finally, $\makered{\I \eqdef} [0,1]$ is the closed and $]0,1[$ the
open unit interval.

\tableofcontents

   \section{Kegelspitzen} \label{Keg}

In this section we introduce the notion of a \KS, which provides the foundation for 
our later developments. We begin with its algebraic
structure. This is that of an abstract convex set, 
which we call a barycentric algebra. We then enrich the
structure, first with a compatible partial order, and then with a directed
complete partial order. In order to make use of our previous work \cite{KP} on d-cones,
we prove embedding theorems of barycentric algebras in abstract cones  at each stage. We conclude by recalling the properties of the subprobabilistic powerdomain and showing how it fits within our framework.

\subsection{\makered{Ordered cones and} ordered barycentric algebras}\label{sec:ba}

For our work there are two basic notions abstracted from substructures in real vector spaces, an abstract notion of a cone and an abstract notion of a convex set.

In a real vector space $V$, a subset $C$ is understood to be a
cone, if $x+y\in C$ and $r\cdot x\in C$ for all $x,y\in C$ and 
every nonnegative real number $r$. Generalising, we obtain an abstract
notion of a cone\display

\begin{defi}\label{def:cone}
An \emph{(abstract) cone} is a set $C$  together with a commutative associative addition
  $(x,y)\mapsto x+y\colon 
C\times C\to C$ 
that  admits a
neutral element $0$, and a scalar multiplication
$x\mapsto r\cdot x\colon C\to C$ by real numbers $r>0$
satisfying the usual laws for scalar multiplication in vector spaces,
that is,  the following  equational laws hold for 
all $x,y,z\in C$ and all real numbers $r>0,s >0$\display
\[\begin{array}{rcllrcl}
x+(y+z)&=&(x+y)+z&\ \ &(rs)\cdot x&=&r\cdot (s\cdot x)\\
x+y&=&y+x&\ \ &(r+s)\cdot x&=&r\cdot x+ s\cdot x \\
x+0&=&x&\ \ & r\cdot (x+y)&=&r\cdot x+ r\cdot y\\
&&&&1\cdot x&=&x
\end{array}\]
Preserving all the above laws we may extend (and we will tacitly
always  do so) the scalar multiplication on a cone to real numbers
$r\geq 0$ by defining  
\[0\cdot x \eqdef 0 \mbox{ for all } x\in C\]
A map $f\colon C\to D$ between cones is said to be\display 

\begin{tabular}{lll}
\emph{homogeneous} &if $f(r\cdot x)=r\cdot f(x)$& for all $r\in \Rp$
and all $x\in C$,\\  
\emph{additive} &if $f(x + y) = f(x) + f(y)$& for all $x,y\in C$,\\
\emph{linear} &if  $f$ is homogeneous and additive.&
\end{tabular}
\end{defi}

In a cone all the equational laws for addition and scalar multiplication
that hold in vector spaces also hold, 
except that we restrict scalar multiplication to nonnegative
real numbers and elements $x$ need not have negatives $-x$. Thus,
 we may calculate in cones just as we do in vector spaces, except that
we have to avoid negatives. \makered{As usual, we generally write scalar multiplication $r\!\cdot \!x$ as $rx$.
}

The notion of a barycentric algebra captures the equational properties
of convex sets. A subset $A$ of a real vector space is
\emph{convex} if
\[ra +(1-r)b \in A \mbox{ for all } a,b \in A, r\in[0,1]
\eqno{\rm (Conv)}\] 
We may use the same property for defining convexity of subsets in an
(abstract) cone. 
On every convex set $A$ we may define for every real number $r\in [0,1]$ a
binary operation $+_r$, the convex combination $(a,b)\mapsto
a+_rb\eqdef r\cdot a +(1-r)\cdot b$. Straightforward calculations show
that these  operations satisfy the following equational laws\display
\begin{align} 
a+_1b &= a \label{B1} \tag{B1}\\
a+_ra& = a \label{B2} \tag{B2} \\
a+_rb& = b+_{1-r}a\label{SC} \tag{SC}\\  
(a+_pb)+_r c &= a+_{pr}(b+_{\frac{r-pr}{1-pr}}c)\ \ \  \mbox{ provided }
r<1, p<1\label{SA} \tag{SA}
\end{align}
SC stands for \emph{skew commutativity} and SA for \emph{skew associativity}.

\begin{defi}\label{def:ba}
An \emph{abstract convex set} or \emph{barycentric algebra} is a
  set $A$ endowed with a binary 
  operation $a+_r b$ for every real number $r$ in the unit interval $[0,1]$
  such that the above equational laws 
(\ref{B1}), (\ref{B2}), (\ref{SC}), (\ref{SA}) hold. 
A map $f\colon A\to B$ between barycentric algebras is
\emph{affine} if $f(a+_r b)=f(a)+_r f(b)$
for all $a,b\in A$ and all $r\in [0,1]$.
\end{defi}

Barycentric algebras $A$ are \emph{entropic} (or  \emph{commutative})
in the sense that all the 
operations $+_r$ are affine maps from $A\times A\to A$, that is, for
all $r$ and $s$ in the unit interval we have the entropic identity 
\[(a +_r b)+_s(c+_r d) = (a+_s c)+_r(b+_s d). \eqno{\rm (E)}\]
If  $c=d$ this reduces to the distributivity law
\[(a+_r b)+_s c = (a+_s c) +_r (b+_s c). \eqno{\rm (D)} \]
The entropic identity (E) can be verified by direct calculation. However,
calculations in barycentric algebras are quite tedious  as the  skew
associativity law (\ref{SA}) is awkward to apply. A simple proof is indicated below after Lemma \ref{lem:freecone}.

 Since barycentric algebras (resp., cones) are equationally defined classes
of algebras, there are free barycentric algebras (resp., free cones), and
every barycentric algebra (resp., cone) is the image of a free one
 under an affine (resp., linear) map. 
 
 These free objects have a simple
 description. 
 For any set $I$ we write $\R^{(I)}$ for the direct sum  of $I$ copies
 of $\R$, that is the vector space of all $I$-tuples
 $x=(x_i)_{i\in I}$ of real numbers $x_i$ such that $x_i\neq 0$ for
 finitely many indices $i$.
 For $i\in I$ we write 
 $\delta(i)$ for $(\delta_{ij})_{j\in I}$,  the canonical basis vector, where
 $\delta_{ij}$ is the Kronecker symbol. Thus, $\delta$ maps $I$ to a
 basis of 
 $\R^{(I)}$. Analogously to the the fact that $\R^{(I)}$ is the free vector space  over the set $I$, we have:

 \begin{lem}\label{lem:freecone} \hfill
 \begin{enumerate}
 \item 
The positive cone $\Rp^{(I)}$ of all $x\in \R^{(I)}$ with nonnegative
entries is the free cone over $I$ with unit $\delta$.
 
 \item The simplex $P_I$ of all $x\in \Rp^{(I)}$ such that $\sum_i x_i =1$ is
 the free barycentric algebra over $I$ with unit $\delta$. \qed
 \end{enumerate} 
  \end{lem}   

The first claim means that, for every map $f$ from $I$ to a cone $C$,
 there is a unique 
 linear map $\overline f\colon \Rp^{(I)}\to C$ such that $f =
 \overline f\circ \delta$, namely $\overline f(x) = \sum_ix_if(i)$.
 Similarly, the second claim tells us that, for every map $f$ from $I$ to
 a barycentric algebra  $A$, there is a unique 
 affine map $\overline f\colon P_I\to A$ such that $f =
 \overline f\circ \delta$. 

In an equationally definable class of algebras, an equational law holds  if and only if it holds in the free algebras. Thus, for the entropic identity (E) to hold in all barycentric algebras, it suffices to verify this property for the free barycentric algebras; and an easy calculation shows that (E) holds in any convex subset of a real vector space. The same calculation can be done after  embedding a barycentric algebra in a cone as a convex subset:

\begin{standardconstr}\label{sc1}  
{\rm Let $A$ be a barycentric algebra. On
\[ \Rpp\times A=  \{(r,a)\mid 0<r\in\R,\ a\in A\}\]
define addition and multiplication by scalars $t>0$ by\display 
\[(r,a)+(s,b)\eqdef (r+s,a+_{\frac{r}{r+s}}b),\quad\quad t(r,a)\eqdef
(tr,a)\]
Adjoin a new element $0$\display 
\[C_A\eqdef \{0\}\cup \{(r,a)\mid r>0,
a\in A\} =  \{0\}\cup (\Rpp\times A)\]
so that $0$ is a neutral element for addition
and  $t\cdot 0 = 0$. 
Simple calculations show that
$C_A$ thereby becomes a  
 cone, and the map $e=(a\mapsto (1,a))\colon A\to C_A$ is an embedding of a barycentric algebra in a cone (i.e., it is affine and injective).
%
%

We  identify
 elements $a\in A$ with the corresponding elements $e(a) = (1,a)\in C_A$, thus identifying
$A$ with the convex subset $\{1\}\times A$ of $C_A$. In this way $A$
becomes a \emph{base} 
of the cone $C_A$ in the sense that $A$ is convex and that every element
$x=(r,a)\neq 0$ in $C_A$ can be written in the form $x=ra$, where $r$
and $a$ are uniquely determined by $x$.}
\end{standardconstr}

We want to add a partial order to our algebraic structure\display

\begin{defi}
An \emph{ordered (abstract) cone} is a cone equipped with a partial order $\leq$
such that addition and scalar multiplication are monotone, that is,
\[
a\leq a'\implies a+b\leq a'+b \ , 
\ \ ra \leq ra'\ .
\]
A map $f\colon C\to D$ between ordered cones is said to be\display

\begin{tabular}{ll}
\emph{subadditive} &if $f(a+b)\leq f(a)+f(b)$ for all $a,b\in C$,\\ 
\emph{superadditive} &if $f(a + b) \geq f(a) + f(b)$ for all $a,b\in C$,\\
\emph{sublinear} &if  $f$ is homogeneous and subadditive,\\
\emph{superlinear}& if $f$ is homogeneous and superadditive.
\end{tabular}
\end{defi}
The linear maps are those that are sublinear and superlinear.

\begin{defi}\label{def:ocone} 
An \emph{ordered barycentric algebra} is a  
barycentric algebra with a partial order $\leq$ such that the
barycentric operations $+_r$ are monotone, that is,
\[a\leq a'  \implies \ a+_r b\leq a'+_rb\]  
for $0\leq r\leq 1$. A map $f\colon A\to B$ between ordered barycentric algebras
is said to be\display

\begin{tabular}{lll}
\emph{convex} &if $f(a+_r b)\leq f(a)+_r f(b)$& for  $0\leq r\leq 1$ and $a,b\in C$,\\ 
\emph{concave} &if $f(a +_r b) \geq f(a) +_r f(b)$& for  $0\leq r\leq 1$ and $a,b\in C$.
\end{tabular}
\end{defi}
 The affine maps are those that are both convex and concave. 

Every barycentric algebra can be understood to be ordered by the
discrete order $a\leq b$ iff $a=b$ and, if $A$ and $B$ both are
discretely ordered, the convex maps between them are the affine ones. 

 Convex subsets of 
ordered vector spaces and ordered cones are ordered barycentric algebras with respect to
the induced order. 


%
In complete analogy to Standard Construction \ref{sc1}, we can construct  embeddings of ordered barycentric algebras in ordered cones, by which we mean  monotone affine maps which are also order embeddings (a monotone map between partial orders is an order embedding if it reflects the partial order).

\begin{standardconstr}\label{sc1'} {\rm
For an ordered barycentric algebra $A$ we use the embedding of
 $A$ in the abstract cone $C_A$ as in 
  Standard Construction \ref{sc1} and we extend the order on $A$ by
 defining an order $\leq$ on $C_A$ by $0\leq 0$ and\display
\[(r,a)\leq (s,b) \iff r=s \mbox{ and } a\leq b \mbox{ in } A\]
With this order, $C_A$ becomes an ordered (abstract) cone and the
affine map $e=a\mapsto (1,a)$ 
from $A$ to  $C_A$ is an affine order embedding.
%
}
\end{standardconstr}

The following surprising lemma 
will be used in later sections\display

\begin{lem}\label{neumann}
Let $a,b,c$ be elements of an ordered barycentric algebra $A$. If
$a+_p c\leq b+_p c$ holds for some $p$ with $0 < p < 1$, then this holds
for all such $p$. 
\end{lem}
\begin{proof} 
We may view $A$ as a convex subset of an ordered cone $C$ according to
\ref{sc1'}. Suppose now that $a,b,c$ are
elements of $A$ and that the inequality $a +_p c\leq b+_p c$ holds for
some $0<p<1$. 

We first claim that  $a +_q c\leq b+_q c$ for $q=\frac{2p}{1+p}$:
By the above hypothesis we have the inequality $pa + (1-p)c\leq pb+(1-p)c$. We
use this inequality twice for establishing this first claim: 
             $qa+(1-q)c = \frac{1}{1+p}(2pa+(1-p)c)=
\frac{1}{1+p}(pa+pa+(1-p)c)\leq\frac{1}{1+p}(pa+pb+(1-p)c)\leq
\frac{1}{1+p}(pb+pb+(1-p)c) = 
\frac{1}{1+p}(2pb+(1-p)c) = qb+(1-q)c$. 

Secondly, we define recursively $p_0 = p$ and $p_{n+1} =
\frac{2p_n}{1+p_n}$. Our first claim allows to conclude $a+_{p_n}b\leq
b+_{p_n} c$ for all $n$. As the $p_n$ form an increasing sequence converging to
$1$, for every $q<1$, there is an $n$ such that $q<p_n$. Thus the
following third claim finishes the proof of our lemma.   

Claim: If $a+_p c\leq b+_p c$ holds for some $p$, then it also holds for
all $q\leq p$. Indeed, 
$a+_q c = qa+(1-q)c = \frac{q}{p}(pa +(1-p)c) + \frac{p-q}{p}c
\leq  \frac{q}{p}(pb +(1-p)c) + \frac{p-q}{p}c = qb+(1-q)c =b+_q c$. 
\end{proof}

\begin{rems}\label{rem:hist1} {\rm ({\bf Historical Notes and References})
\begin{custominlinelist}
\item The axiomatization of convex sets arising in vector spaces over the reals
 has a long history. The first
axiomatization seems to be due to M.~H.~Stone \cite{S}. Independently,
H.\ Kneser \cite{Kne} gave a similar axiomatization motivated by von
Neumann's and Morgenstern's work on game theory \cite{NM}. Stone's and Kneser's results are not restricted to the reals; they axiomatise convex sets embeddable in vector spaces over linearly ordered skew fields. Such an axiomatization cannot be equational. For a
barycentric algebra to be embeddable in a  vector
space one has to add a cancellation axiom\display
\[a +_r c = b+_r c \implies a=b  \quad (\mbox{for } 0 < r < 1)\eqno{\rm (C1)}\]
Similarly, one can show that an ordered barycentric algebra is
embeddable in an ordered vector space if, and only if, it satisfies the
order cancellation axiom\display 
\[a +_r c \leq b+_r c \implies a \leq b  \quad (\mbox{for } 0 < r < 1)\eqno{\rm (OC1)}\]

\item Several authors  independently developed  the
equational theory of convex sets and the corresponding notion of an
abstract convex set (while ignoring
each other). Our extension to ordered structures seems to be
new. We now try to give  as
complete as possible an account of these developments. The reader
should be aware that we always stay 
within the equational theory of convex sets in real vector spaces, and
we do not go into generalities, such as abstract convexities in the 
sense of, for example, van de Vel \cite{V}.       
 
\item W.~Neumann \cite{Neu} seems to be the first to have looked at the  
equational theory of convex sets. He remarked that barycentric
algebras may be very different from convex sets in vector
spaces. Indeed $\vee$-semilattices 
become examples of barycentric algebras if we define
$a+_r b\eqdef a\vee b$ for $0<r<1$. Neumann \cite{Neu} noticed
that the semilattices form the only proper nontrivial
equationally definable subclass of the class of all barycentric
algebras. Every barycentric algebra has a greatest homomorphic
image which is a semilattice. This semilattice is significant; indeed,
for a convex subset of a real vector space this greatest homomorphic
semilattice image is the (semi-)lattice of its faces. \makered{W.~Neumann also
characterised the free barycentric algebras, a characterisation that we reproduced in Lemma \ref{lem:freecone}.} 
   
The equational axioms (B1), (B2), (SC), (SA) that we use in our definition of 
barycentric algebras are due to \v{S}wirszcz \cite{Sw} and 
have been reproduced by Romanowska and Smith \cite[Section
  5.8]{RS}; the same axioms have also been used by
  Graham \cite{graham}  
and by Jones and Plotkin \cite{JP,jones} when they introduced the
notion of an \emph{abstract probabilistic powerdomain}. Romanowska and
Smith introduced the term \emph{barycentric algebra} for an abstract
convex set; \makered{their monograph, cited above, is an
  exhaustive source on  barycentric algebras and related structures.}

\item The notion of an abstract cone has emerged under the name of
\emph{quasilinear space} in interval mathematics in works of
O.\ Mayer \cite{May}; one may also consult papers by W.\ Schirotzek, in
particular \cite{Sch}, where ordered quasilinear spaces appear, the
closed intervals in the reals with the Egli-Milner order forming a prime
example. The embedding of a barycentric algebra 
as a convex subset in the abstract cone $C_A$ is due to J.\ Flood
\cite{Fl}. 
The surprising lemma \ref{neumann} is due to W.\ Neumann
\cite{Neu} in the  unordered case. 
The possibility of embedding a barycentric algebra in
a cone makes calculations much easier as
already  remarked by J.\ Flood \cite{Fl}. The proof of Lemma
\ref{neumann} illustrates this advantage when compared with Neumann's
original proof in the unordered case.  

Let us note in passing that every $\vee$-semilattice with a zero
becomes
an abstract cone by defining $a+b= a\vee b$ and $ra= a$ for $r> 0$,
but $ra = 0$ for $r=0$. 
The class of $\vee$-semilattices with $0$ 
is the only proper nontrivial equationally definable subclass of
the class of all abstract cones.


  


\item Another early axiomatisation of abstract convex sets is due to Gudder
\cite{gudder}. He uses as operations convex combinations of two and
three elements and axiomatises these operations. Without introducing
the notion of an abstract cone he gives the construction of the
embedding of an abstract convex set into a cone viewed as a convex set.

\item Equivalent to the equational one, there is another approach to abstract convex sets
 initiated by T.\ \v{S}wirszcz \cite{Sw} who characterises them as the
Eilenberg-Moore algebras of the monad $P$ of probability distributions
with finite support over the category of sets. In functional analysis
this approach has been rediscovered by G.\ Rod\'e \cite{Ro} and  
developed further by H.\  K\"onig \cite{Koe} without any background in
category theory. The approach has been
pursued further in a series of papers by Pumpl\"un, R\"ohrl, Kemper
and others (see e.g. \cite{Pu, Pu1, PR, Ke, W}).
We summarise it as follows:

For all natural numbers $m>0$, the set $P_m$ of all probability
measures $\mathbf q = (q_1,\dots,q_m)$ 
on an $m$-element set is a compact convex
set and, for every finite set $\mathbf q_1, \dots,\mathbf q_n\in P_m$,
the convex combination $\sum_{i=1}^n p_i\mathbf q_i$ is again an
element of $P_m$ for every $\mathbf p = (p_1,\dots,p_n)\in P_n$. The
sum $\sum_{i=1}^n p_i\mathbf q_i$ is the \emph{barycenter} of
masses $p_1,\dots,p_n$ placed at points $\mathbf
q_1,\dots,\mathbf q_n$. The extreme points of $P_n$ are the Dirac
measures $\delta_i$, $i=1,\dots,n$, given by the Kronecker symbol
$\delta_{ki}$. 
 
A \emph{convex space} is a nonempty set $X$ together with a family of mappings 
$\mathbf p^X :X^n\to X$ for $\mathbf p \in P_n$, $n=1,2,\dots$, satisfying for every
$x=(x_1,\dots,x_m)\in X^m$ the identities
\[\delta^X_i(x) = x_i,\ \ \ \ \mathbf p^X(\mathbf q^X_1(x),\dots,\mathbf
q^X_n(x)) = \big(\sum_{i=1}^np_i\mathbf q_i\big)^X(x)\]
for all $\mathbf p \in P_n$ and $\mathbf q_1,\dots,\mathbf q_n\in
P_m$.
Using the formal notation $\sum_{i=1}^n p_ix_i = \mathbf p^X(x)$
the two equations take the form
\[\sum_{i=1}^n\delta_{ik}x_i = x_k \eqno{\rm (A1)}\]
where $\delta_{ik}$ is the Kronecker symbol,  and
\[ \sum_{i=1}^n p_i\big(\sum_{k=1}^m q_{ik}x_k\big) =
      \sum_{k=1}^m\big(\sum_{i=1}^n p_i q_{ik}\big)x_k \eqno{\rm (A2)}\]
for $\mathbf p=(p_1, \dots,p_n)\in P_n$, and  $\mathbf
q_i=(q_{i1},\dots,q_{im})\in P_m$, $i=1,\dots,n$.

It should be added that the main interest of the authors using the latter
approach was not in convex but in \emph{superconvex} spaces. For
those  one replaces the sets $P_n$ of probability measures on
finite sets by the set $P_I$ of discrete probability measures on an
infinite countable set $I$, using the same defining identities as
above replacing finite by infinite sums.
\end{custominlinelist} }
\end{rems}

\subsection{Ordered pointed barycentric algebras}\label{sec:oba}

Convex sets containing $0$ in vector spaces and in (abstract) cones
are pointed barycentric algebras in the following sense:

\begin{defi}\label{def:pba}
 A  \emph{pointed barycentric algebra} is a
barycentric algebra $A$ with a distinguished element usually written $0$.
A map $f\colon A\to B$ between pointed barycentric algebras is
\emph{0-affine} or \emph{linear} 
if it is affine and preserves the distinguished element: $f(0)=0$.
\end{defi}

If $A$ is a convex set containing 0 in a vector space or in an
abstract cone, we have 
$ra\in A$ for all $a\in A$ and all $r\in [0,1]$, that is, we have a
multiplication by scalars $r\in[0,1]$. Similarly, we can define a
multiplication by scalars $r\in[0,1]$ for an arbitrary
pointed barycentric algebra $A$. For an element $a\in A$ and $r\in[0,1]$ define
\[r\cdot a\eqdef a+_r 0\]
Scalar multiplication has the usual properties: 
\[0\cdot a
= 0 = r\cdot 0,\ 1\cdot a =a,\ (rs)\cdot a = r\cdot(s\cdot a),\ r\cdot(a +_s b)
=r\cdot a +_s r\cdot b\]
as is straightforward to check (the last
property follows from the distributive law (D)). 
Every linear map $f\colon A\to B$ of pointed barycentric algebras is
\emph{homogeneous} 
in the sense that $f(r\cdot a) =r\cdot f(a)$ for all $a\in A$ and $0\leq r\leq
1$. Indeed, we have: $f(r\cdot a)=f(a+_r 0) = 
f(a) +_r f(0)= f(a) +_r 0 = r\cdot f(a)$.
 
Every cone is a pointed barycentric algebra. Thus, for a map $f$
between cones we have two notions of homogeneity, one where
$r$ ranges over all nonnegative real numbers,  and another where
$r$ only ranges over the unit interval. But the two notions are
equivalent for cones; indeed, if $f(ra) =rf(a)$ for all $r$ in the
unit interval, then $f(a) = f(r^{-1}ra) = r^{-1}f(ra)$ for all $r\geq
1$, whence $rf(a) = f(ra)$ for all $r\geq 1$, too. This shows that our
terminology of linearity for maps between   
cones and pointed barycentric algebras, respectively, is consistent. 

As for barycentric algebras (see \ref{lem:freecone}), there is a simple
description for the free 
pointed barycentric algebra over a set $I$:

 \begin{lem}\label{lem:freepointed}
 The convex set \[S_I = \{x=(x_i)_{i\in I}\in\Rp^{(I)}\mid \sum_i
 x_i\leq 1\}\subseteq \R^{(I)}\]
 of all finitely supported subprobability distributions on $I$ with $(0)_{i \in I}$
 as distinguished element is the free pointed barycentric algebra
 over $I$ with unit $\delta$. 
 \end{lem} 

 \begin{proof} 
Let $A$ be a pointed barycentric algebra with distinguished
 element $0$ and $f\colon I\to A$
 an arbitrary function. We add a new element $i_0$ to $I$ and extend
 $f$ to this new element by $f(i_0) = 0$. There is a unique affine map
 $g$ from the free barycentric algebra 
 $P_{I\cup\{i_0\}}$ to $A$ such that $g\circ \delta = f$. We now
 compose $g$ with the affine bijection from  $S_I$ to
 $P_{I\cup\{i_0\}}$ that maps the subprobability distribution
 $x=(x_i)_{i\in I}$ on $I$ to the probability distribution
 $y$ 
 on $I\cup\{i_0\}$ with 
 $y_i = x_i$ for $i \in I$, and $y_{i_0} = 1 - \sum_{i\in I} x_i$. In this way we obtain an affine
 map $\overline f\colon S_I\to A$ such that $f = \overline f\circ
 \delta$ and $f(0)=0$. Clearly, $\overline f$ is unique with these
 properties.
\end{proof}

We now add a partial order:

\begin{defi}
An \emph{ordered pointed barycentric algebra} is  an
ordered barycentric algebra with a distinguished element $0$. 
A map $f\colon A\to B$ of ordered pointed barycentric algebras is
\emph{sublinear}  (resp., \emph{superlinear}) if it is homogeneous and
convex (resp., concave).  
\end{defi}
The linear maps are those
that are both sublinear and superlinear. 

For general reasons, there is  a  free ordered cone $\Cone(A)$ over any
ordered pointed barycentric algebra $A$
.
We need a concrete
construction of this free object. 
The idea is to stretch the
multiplication by scalars from those in the unit interval to all
nonnegative reals: 


\begin{standardconstr} \label{sc2} {\rm
 Let $A$ be an ordered pointed barycentric algebra. Since every pointed
 barycentric algebra is the image of a free one under a linear map, 
 by Lemma \ref{lem:freepointed} there is a linear surjection $f\colon
 S_I \to A$, where $S_I$ is the pointed barycentric algebra of
 finitely supported subprobability measures on some set $I$. 

 For every
 $x\in \Rp^{(I)}$ there is a greatest real number $0<r\leq 1$ such that $rx\in
 S_I$, and for every real number $s$ with  $0<s\leq r$ we also have
 $sx\in S_I$. 
 We define a relation $\precsim$ on $\Rp^{(I)}$: \[x\precsim x' \mbox{
   if } f(rx)\leq f(rx') \mbox{ for some }r,\ 0<r\leq 1, \mbox{ such that
 } rx\in S_I \mbox{ and }rx'\in S_I.\] 
If this holds  for some $0<r\leq 1$, then it also holds for all $s$ with
 $0<s\leq r$, since $f$ is homogeneous. 
 \begin{lem} 
The relation $\precsim$ is a preorder on the cone $\Rp^{(I)}$ compatible with the cone operations, that is, $x\precsim x'$ implies $x+y\precsim x'+y$ and $rx\precsim rx'$ for every $y$ in the cone and every nonnegative real  number $r$ . 
\end{lem}
\begin{proof}
Clearly, the relation $\precsim$ is reflexive and transitive. For compatibility,
consider elements $x\precsim x'$ in $S_I$. There is an $r$ with $0<r\leq 1$ such that
 $rx,rx'\in S_I$ and  
 $f(rx)\leq f(rx')$. Given $y$, choose $s$ such that  $0<s\leq r$ and $s(x+y),
 s(x'+y)\in S_I$. Using that $f$ is linear on $S_I$ we obtain
 $f(\frac{s}{2}(x+y)) = f(s(x+_{\frac{1}{2}} y)) = 
 f(sx+_{\frac{1}{2}} sy) = f(sx)+_{\frac{1}{2}}f(sy) \leq
 f(sx')+_{\frac{1}{2}}f(sy) = f(\frac{s}{2}(x'+y))$, that is $x+y\precsim
 x'+y$. The property that $tx\precsim tx'$ for $t\in \Rp$ is straightforward. 
 \end{proof}

\begin{cor}
 The relation $x\sim x'$ if $x\precsim x'$ and $x'\precsim x$ is a
 congruence relation on the cone  $\Rp^{(I)}$. \qed
\end{cor}

 We write $\Cone(A)$ for $\Rp^{(I)}/\!\!\sim$, the quotient cone ordered by
  $\rho(x)\leq \rho(y)$ if $x\precsim y$, where $\rho\colon \Rp^{(I)}\to \Rp^{(I)}/\!\!\sim$ is the (linear monotone) quotient map.
Restricted to $S_I$, the quotient map  factors through $f$. Indeed, if $x,y$ are elements of $S_I$ such that $f(x) \leq  f(y)$, then $x\precsim y$ by the definition of $\precsim$. Thus, there  
 is a unique monotone linear map $u\colon A\to \Cone(A)$ such that $\rho|_{S_I} = u\circ f$. 
}
\end{standardconstr}

\makered{We note two important properties of this map\display
 \begin{lem}  \label{two-properties}\hfill
  \begin{enumerate}
 \item If $u(a) \leq u(b)$, for $a,b \in A$, then  $ra \leq rb$ for
   some $0 < r \leq 1$. 
 \item Every $b \in \Cone(A)$ has the form $r u(a)$ for some $a
   \in A$ and $r \geq 1$.  
 \end{enumerate}
\end{lem}
\begin{proof} \hfill
  \begin{enumerate}
  \item Suppose that $u(f(x)) \leq u(f(y))$, for $x,y \in S_I$. Then $\rho(x) \leq \rho(y)$. So $x \precsim y$ and we have $f(rx)\leq f(ry)$ for some $ 0< r\leq 1$, which is to say that $rf(x)\leq rf(y)$ for some $ 0< r\leq 1$.

\item Choose $x\in S_I$ such that $\rho(x)=b \in \Cone(A)$. There is a $y \in S_I$ and an $r  \geq 1$ such that $x = ry$. Set $a \eqdef f(y)$. Then: 
$ru(a) = ru(f(y)) = r \rho(y) = \rho(ry) = \rho(x)= b$. \qedhere
 \end{enumerate}
 \end{proof}}

The two properties of the map $u\colon A\to\Cone(A)$ picked
  out in this lemma yield a 
  characterisation of when a monotone linear map from a pointed
  ordered barycentric algebra to an ordered cone is universal \makeblue{as
  expressed by the freeness property in the following proposition}:
\begin{prop} \label{ocone-embed} Let $u: A \rightarrow C$ be a
  monotone linear map 
  from a pointed ordered barycentric algebra to an ordered cone $C$.
Then $C$  is the  free ordered cone  over
 $A$ with unit $u$ (in the sense that every
 monotone homogeneous map $h$ from $A$ into an ordered cone 
 $D$ has a unique monotone homogeneous extension  
$\widetilde h\colon C\to D$ along $u$) if, and only if, the following two
 properties hold of $u$\display 
 \begin{enumerate}
 \item If $u(a) \leq u(b)$, for $a,b \in A$, then  $ra \leq rb$ for
   some $0 < r \leq 1$. 
 \item Every $x \in C$ has the form $r u(a)$ for some $a \in A$ and $r \geq 1$. 
 \end{enumerate}
For such universal maps $u$, the extension $\widetilde h$ is
sublinear, superlinear, linear, respectively if, and 
 only if, $h$ is. Moreover, $\widetilde h\leq \widetilde g$ if and
 only if $h\leq g$. 
 \end{prop}
\begin{proof} Suppose the two properties hold. For uniqueness of the
  extension use the second assumption and note that, for such an
  $\widetilde h$, we have\display    
\[\widetilde h(r u(a)) = r\widetilde h( u(a)) = r h(a)\]
We then use this property to define $\widetilde h$; to show $\widetilde h$ well-defined and monotone we chose $a,a' \in A$ and $r,r' \geq 1$ and prove that if $ru(a) \leq r' u(a')$ then $ r h(a) \leq r' h(a') $.  For if $r u(a) \leq r'u(a')$  then $u(rs^{-1}a) \leq  u(r's^{-1} a')$, where $s \eqdef \max(r,r')$; so, by the first assumption, $trs^{-1} a \leq  tr's^{-1} a'$ for some $0 < t \leq 1$. 
So, in turn, we have 
$trs^{-1}  h(a) = h( trs^{-1} a) \leq  h(tr's^{-1}  a')  = tr's^{-1}
h(a)$, and so $r h(a)  \leq  r' h(a)$ as required.   It is
straightforward that $\widetilde h$ is homogeneous.
  
That the two properties are necessary follows from
Lemma~\ref{two-properties}, which provides an
example of a map $u\colon A\to \Cone(A)$
possessing them. \makeblue{Suppose indeed that $u'\colon A\to C$ has the
universal property. Then there are mutually inverse monotone homogeneous maps
$\widetilde u\colon C\to\Cone(A)$ and $\widetilde{u'}\colon
\Cone(A)\to C$ such that $u=\widetilde u\circ u'$ and
$u'=\widetilde{u'}\circ u$. For property 1, take $a\leq b$ in $A$ such
that $u'(a)\leq 
u'(b)$. Then $u(a) = \widetilde u(u'(a))\leq \widetilde u(u'(b)) =
u(b)$, hence $ra\leq rb$ for some $0<r\leq 1$. For property 2, take
$x\in C$. Then $\widetilde u(x) = ru(a)$ for some $a\in A$ and $r\geq
1$. Applying $\widetilde{u'}$ yields $x = \widetilde{u'}(ru(a))=
r\widetilde{u'}(u(a))=ru'(a)$.}

Given such a map $u: A \rightarrow C$, suppose $h$ is sublinear. Then so is $\widetilde h $. To see this choose $x,x' \in C$. By the first property $x = r u(a)$ and $x' = r' u(a')$ for some $a,a' \in A$ and $r,r' \geq 1$. Then set $s \eqdef r + r'$ and calculate:  
$\widetilde h (x + y) = \widetilde h (ru(a) + r' u(a')) 
                                = \widetilde h (s u(a +_{r/s} a')) 
                                =  s  h (a +_{r/s} a') 
                                \leq  s ( h (a) +_{r/s} h(a'))
                                (\mbox{as $h$ is sublinear}) 
                                =  r   h (a) + r' h(a') = \widetilde h(x) + \widetilde h(y)$. The converse is evident as $u$ is linear. The assertion for superlinearity is proved similarly, and then the assertion for linearity follows.
                                
Finally we show that $\widetilde h\leq \widetilde g$  if $h\leq g$ (the converse is obvious). Assuming $h\leq g$, for any $a \in A$ and $r \geq 1$ we need only calculate: $\widetilde h(ru(a)) = r h(a) \leq r g(a) = \widetilde g(r u(a))$. 
\end{proof}

Together with Lemma~\ref{two-properties} we now have\display

 \begin{thm}\label{th:sc2}
 For any ordered pointed barycentric algebra $A$, $\Cone(A)$ is the
 free ordered cone  over 
 $A$ with unit $u$ in the following strong sense: Every
 monotone homogeneous map $h$ from $A$ into an ordered cone 
 $D$ has a unique monotone homogeneous extension  $\widetilde h\colon
 \Cone(A)\to D$ along $u$. The extension $\widetilde h$ is sublinear,
 superlinear, linear, respectively if, and 
 only if, $h$ is. Moreover, $\widetilde h\leq \widetilde g$ if and
 only if $h\leq g$. \qed
\end{thm}

It would be nice, if the monotone linear map $u$ from $A$ into the
universal cone would be an order embedding
. 
%
But this is not always the case. We give an example of an ordered pointed barycentric algebra that cannot be embedded in any ordered cone at all. Indeed, in an
ordered cone, if for some $0 < r < 1$
we have $r x \leq r y$, then $x\leq y$ by 
multiplying  with the scalar $r^{-1}$. This need not be true in an ordered pointed barycentric algebra\display 

\begin{exa}
{\rm We consider the unit interval and replace the element $1$ by two
elements $1_1,1_2$. 
On each of the sets $[0,1[\ \cup\ 1_ i,\ i=1,2,$ we take convex
combinations as usual in the unit interval, and the
 set $\{1_1,1_2\}$ is considered as a join-semilattice with $x
+_r y = 1_2$ whenever $x\neq y$ and $0 < r<1$. In this way we obtain a
pointed barycentric algebra with $0$ as distinguished
element. Clearly, $r x = r = r y$ whenever $0<r<1$ and
$x,y\in \{1_1,1_2\}$.}
\end{exa}

We therefore pay attention to the order cancellation law\display 
\[ r x \leq r y \implies x\leq y \quad (\mbox{for  } 0 < r < 1)
\eqno{\rm (OC2)}\]
and its specialised form
\[ r x = r y \implies x= y \quad (\mbox{for  } 0 < r < 1)  \eqno{\rm (C2)}\]
These laws are particular instances of the cancellation laws (OC1)
and (C1). For example (C2) can be rewritten as $x +_r 0 =
y+_r 0\implies x=y$.

In view of the properties of scalar multiplication, the map $x\mapsto 
rx$ of $A$ into itself is linear and monotone. The axiom
(OC2) is equivalent to the statement that this map is also an order
embedding whenever $0<r\leq 1$, which implies that the image $rA$ is
an ordered pointed barycentric algebra isomorphic to $A$.  

Property (OC2) is less restrictive than it seems. Indeed, using Lemma
 \ref{neumann} for the case $c=0$ we obtain: 

\begin{rem}\label{neumann2}
 If in an ordered pointed barycentric algebra $ra\leq rb$ for some
 $0<r<1$ then this holds for all such $r$. 
 \end{rem}

 

 
We now answer the question which ordered pointed barycentric algebras can be embedded into  ordered cones\makered{, where by embedding we mean a
linear order embedding}\display 

\makered{\begin{prop}\label{prop:oc2}
For an ordered pointed barycentric algebra $A$ satisfying {\rm (OC2)}, the
universal map $u\colon A\to\Cone(A)$ is an embedding. An ordered pointed barycentric algebra can be embedded into  an ordered cone if, and only if, it satisfies  {\rm (OC2)}. 
\end{prop}
\begin{proof}
%
Consider $u: A \rightarrow \Cone(A)$, and, given
$a,a'\in A$, suppose that $u(a)\leq u(a')$. Then by
Lemma~\ref{two-properties}.1,  $r a \leq r a'$ for some $0 <
r \leq 1$, and so $a \leq a'$, if $A$ satisfies \rm{(OC2)}. Thus $u$
is an order embedding. Thus, if $A$ satisfies (OC2), it can be
embedded in an ordered cone. Conversely,
if an ordered pointed barycentric algebra can be embedded in an ordered
cone, it satisfies (OC2), since (OC2) is satisfied in every ordered cone.
\end{proof}}

\makered{We can also identify which embeddings are universal\display
\begin{cor} \label{free-oCone} An embedding $u\!: \!A \rightarrow C$ of an ordered pointed barycentric algebra $A$ in an ordered cone $C$ is universal if, and only if, every $x \in C$ has the form $r\cdot u(a)$, for some $a \in A$ and $r \geq 1$. \qed
\end{cor}
Under the assumption {\rm (OC2)},  we may
identify an ordered pointed barycentric $A$ with its image in
$\Cone(A)$ under the 
 embedding $u$ by the previous proposition and we will do so without
 mentioning in the sequel. With this identification for every $c\in
 \Cone(A)$, there is an $0 < r\leq 1$ such that $rc\in A$; we will
 frequently use this fact.} 

%


It would be desirable to embed an ordered pointed barycentric algebra
in an ordered cone as a lower set. \makered{If
an ordered pointed barycentric algebra $A$ is embedded in an ordered cone
$C$ as a lower set, then one has for all $a,b\in A$\display 
\[ a\leq rb \implies \exists a' \in A.\ a=ra' \quad (\mbox{for  } 0 < r
 < 1)  \eqno{\rm (OC3)} \]
Indeed, if $a \leq rb$
for $a,b\in A$ and $0<r<1$, then $a'\eqdef\frac{1}{r}a \leq b$ in $C$. Then
$a' \in A$, as $A$ is a lower set in $C$, and $a = ra'$. Property
(OC3) is not satisfied for all ordered pointed barycentric algebras.}
As an example, let $A$ be the convex hull of the 
points $(0,0), (1,0), (0,1), (2,2)$ in $\Rp^2$. Then $A$ is a pointed
barycentric algebra
embedded in the ordered cone $\Rp^2$, but $A$ is not a lower set. As
it does not satisfy Property (OC3), $A$ cannot be embedded in any
ordered cone as a lower set. 
We also have a converse\display

\begin{lem}\label{lem:opba}
An ordered pointed barycentric algebra $A$ satisfying order
cancellation {\rm (OC2)} is
embedded in the ordered cone $\Cone(A)$ as a lower set if, and only if,
\makered{it satisfies Property {\rm (OC3)}}.
\end{lem}   

\begin{proof}
By Proposition \ref{prop:oc2} we can embed $A$ into the ordered cone
$\Cone(A)$. To show that $A$ is embedded as a lower set, suppose
that $x$ is an element of $\Cone(A)$ such that $x\leq b$ for some
$b\in A$. Then  $rx \in A$  for some $r$ with $0 < r < 1$ and $rx\leq
rb$. By Property (OC3) there is an
$x'\in A$ such that $rx' = rx$ which implies $x = x'\in A$ by
multiplying by $\frac{1}{r}$. 
\end{proof} 

In the presence of the order cancellation property (OC2)
 one has $a'\leq b$ for the element $a'$ whose existence is postulated
 in (OC3).

\begin{rem}\label{rem:hist2}{\rm ({\bf Historical Notes and References})
As far as we know, pointed barycentric algebras have not attracted
much attention. They are identical to the \emph{finitely 
  positively convex spaces} in the sense of Wickenh\"auser, Pumpl\"un,
R\"ohrl, and Kemper \cite{W,Pu,Pu1,PR,Ke}. To define them, one
uses the same setting and the identities used for  convex
spaces (see  historical remark \ref{rem:hist1}), but replaces
the convex sets $P_n$ of 
probability measures on $n$-element sets by the pointed convex sets
\[S_n = \{(q_1,\dots,q_n)\in[0,1]^n\mid \sum_{i=1}^nq_i \leq 1\}\]
of subprobability measures on $n$-element sets. As before, the main
interest of these authors 
was directed towards the positively \emph{superconvex} spaces and
their applications in functional analysis, where
the $S_n$ are replaced by the set $S$ of subprobability measures on
an infinite countable set. Our standard construction \ref{sc2} for
constructing the free ordered cone over an ordered  pointed
barycentric algebra is simpler than Pumpl\"un's construction of a
free cone over an (unordered) positively superconvex space (see
\cite[Definition 4.17 ff.]{Pu}).   

In the same way as join-sem\-ilattices can be considered to be
barycentric algebras,  join-semi\-lattices with a distinguished element
can be considered to be pointed barycentric algebras
. 






Recently, convex and positively convex spaces were taken up by
A.\ Sokolova and H.\ Woracek \cite{SW}. These authors are particularly interested in
finitely generated barycentric and pointed barycentric algebras, that
is, homomorphic images of polyhedra and pointed polyhedra in finite
dimensional vector spaces, and they prove that
finitely generated  barycentric and 
pointed barycentric algebras, respectively, are finitely presented. 
}
\end{rem}

\subsection{\makered{d-Cones and} Kegelspitzen} \label{dCs_KSs}

We now endow partially ordered sets with their Scott topology. In
particular, the sets $\Rp$, $\oRp$, and the unit interval $[0,1]$
are endowed with their usual order and the corresponding Scott
topology. Maps are restricted to Scott-continuous ones. 

Thus, for an ordered cone $C$ it is natural to ask for addition 
$(a,b)\mapsto a+b\colon C\times C\to C$ and scalar multiplication
$(r,a)\mapsto ra\colon \Rp\times C\to C$ to be Scott-continuous (in
both arguments).
Note that the continuity of scalar multiplication in the first
argument implies that $0$ is the least element of $C$; indeed,
Scott-continuity with respect to scalars implies that $r\mapsto
ra\colon \Rp\to C$ is 
monotone, whence $0\leq 1$ implies $0 = 0\cdot a \leq 1\cdot a = a$. 

\begin{defi}
An ordered cone in which addition 
$(a,b)\mapsto a+b\colon C\times C\to C$ and scalar multiplication
$(r,a)\mapsto ra\colon \Rp\times C\to C$ are Scott-continuous (in
both arguments) will be called an \emph{s-cone}. If
in addition the order is  directed complete (resp., bounded directed
complete), we say that $C$   
is a  \emph{d-cone} (resp., a \emph{b-cone}).  
\end{defi}

 We are heading towards a similar connection between the algebraic and
the order structure on ordered pointed barycentric algebras. 
For this we have to restrict the scalars $r$ to the unit interval\display

\begin{defi}\label{def:KS}
A \emph{\KS}\  is a pointed barycentric algebra $K$ equipped with a directed complete
partial order such that, for every $r$ in the unit interval, convex combination
$(a,b)\mapsto a+_r b\colon C\times C\to C$ and scalar multiplication
$(r,a)\mapsto ra\colon [0,1]\times C\to C$ are Scott-continuous in
both arguments. 
 (An alternative name would be \emph{pointed
  barycentric d-algebra}.)   
\end{defi} 
We only need to require scalar multiplication to be continuous in its
first argument in the definition, since $a\mapsto ra = a+_r 0$ is
required to be continuous anyway. 
 \makered{The minimalistic definition of a \KS\
  above may look artificial. In the Historical Notes
  \ref{rem:hist3} below we discuss an equivalent definition that
  looks more natural.} 

Since Scott-continuous maps are monotone, every \KS\ is an ordered
pointed barycentric algebra. As for d-cones, $0$
will be the least element. It is noteworthy that property (OC2) is
always satisfied\display  

\begin{lem}\label{oc2}
Every \KS\ satisfies the order cancellation property {\rm (OC2)}: If $ra\leq rb$
for some $0<r<1$, then $a\leq b$. 
\end{lem} 

\begin{proof}
Indeed, if  $ra\leq rb$ for some $0<r<1$,  the order theoretical
version of Neumann's lemma \ref{neumann} implies that $ra\leq rb$ for
all $r<1$ which implies $a\leq b$ by the Scott continuity of the map
$r\mapsto ra\colon [0,1]\to K$.
\end{proof}

We would like to embed every \KS\ $K$ into a d-cone, where embeddings
of \KSs\ in d-cones are \makered{Scott-}continuous linear maps which
are order embeddings. \makered{We proceed in two steps. In a first
step} we use the embedding of $K$ (considered as a pointed barycentric
algebra) in the ordered cone $\Cone(K)$ according to  
Standard Construction \ref{sc2}. By Proposition
\ref{prop:oc2} this Standard Construction  yields indeed a 
linear order embedding $u$ of $K$ in $\Cone(K)$, since $K$ 
satisfies (OC2) by Lemma \ref{oc2}. \makered{By the following lemma, this
embedding is Scott-continuous: 


\begin{lem}\label{lem:continuous}
Let $u$ be a homogeneous order embedding of a \KS\ $K$ in an s-cone
$C$ in such a way that for every element 
$y\in C$ there is an $r,\ 0<r<1,$ such that $ry\in u(K)$. Then $u$ is
Scott-continuous. 
\end{lem}

\begin{proof}
Let $(x_i)_i$ be a directed family in $K$ and $x$ its supremum in
$K$. Since $u$ is monotone, clearly $u(x_i)\leq u(x)$. In order to
show that $u(x)$ is the supremum of the $u(x_i)$ in $C$,
consider any upper bound $y\in C$ of the 
$u(x_i)$. Choose an $r,\ 0<r<1,$ such that $ry\in u(K)$ and $y'\in K$ with
$u(y')= ry$. Then $ry =u(y')$ is an upper bound of the directed family
$ru(x_i)= u(rx_i)$ (using homogeneity of $u$). Since $u$ is an order
embedding, $y'$ is an upper bound of the $rx_i$ in $K$. By the Scott
continuity of scalar multiplication, $rx$ is 
the least upper bound of the $rx_i$ in $K$, whence $rx\leq y'$. Using
homogeneity and monotonicity of $u$, we deduce $ru(x)=
u(rx)\leq u(y')=ry$ which implies $u(x)\leq y$.
\end{proof}  
}

We now want to show that scalar multiplication and addition
on $\Cone (K)$ are Scott-continuous. This is
no problem for scalar multiplication:

We first recall that $a\mapsto ra\colon \Cone(K)\to \Cone(K)$ is
Scott-continuous for every $r>0$, since this map is monotone
and has a monotone inverse, multiplication by $r^{-1}$.

We now verify that $r\mapsto ra\colon\Rp\to \Cone(K)$ is
Scott-continuous. Suppose indeed that 
$r_i$ is an increasing family in $\Rp$ with $r=\sup_i r_i$. Choose an
$s, 0 < s < 1,$ such that $sr \leq 1$ and $sa\in K$. We then use the
continuity of $r\mapsto ra\colon[0,1]\to K$ to obtain $\sup_i
(sr_i)(sa) = (\sup_isr_i)(sa) =(sr)(sa)$. Now in $\Cone(K)$ we have
$s^2\sup_i r_ia  = \sup_i (sr_i)(sa) = srsa = s^2(ra)$ which implies
$\sup_i r_ia = ra$.

We now turn to addition. To prove that $a\mapsto a+b\colon \Cone(K)\to
\Cone(K)$ is 
Scott-continuous for every fixed $b\in \Cone(K)$, we have to show: If
$a_i$ is a directed system in 
$\Cone(K)$ which has a sup $a=\sup_i a_i$ then the family $a_i+b$ has
a sup and $a +b = \sup_i(a_i+b)$. For this we choose an $s,\ 0<s<1,$ such
that $sa\in K$ and $sb\in K$. Because 
$sa_i\leq sa$, we would like to conclude that $sa_i\in K$, since then
the Scott continuity of convex combination in $K$ implies that
$\sup_i(\frac{1}{2}(sa_i) + \frac{1}{2}(sb)) =
\frac{1}{2}\sup_i(sa_i) + \frac{1}{2}(sb) =\frac{1}{2}(sa) +
\frac{1}{2}(sb)$, whence in $\Cone(K)$ we have  $\sup_i(a_i
+b)=2s^{-1}\sup_i(\frac{1}{2}(sa_i) + \frac{1}{2}(sb))
=2s^{-1}(\frac{1}{2}(sa) + \frac{1}{2}(sb) = a+b$ as desired.

Thus, we would like to use that $K$ is a lower set in $\Cone(K)$. 
By Lemma \ref{lem:opba}, this is equivalent to the requirement that
$K$ satisfies Property (OC3).
\makered{As we often use this
property we make a definition:


\begin{defi}
A \KS\ is said to be \emph{full}  if it
satisfies Property {\rm (OC3)}.
\end{defi}
}




We now can state:

\begin{prop}\label{lem:bcone}
For a full \KS\ $K$, 
the free cone $\Cone(K)$
over $K$ according to Standard Construction \ref{sc2} is a b-cone
and $K$ is \makered{Scott-continuously} embedded in $\Cone(K)$ as a
Scott-closed convex set.  


\makered{The b-cone $\Cone(K)$ is the free b-cone over $K$ w.r.t.\
  Scott-continuous homogeneous maps as, for every such  map $f$ from
  $K$ into a b-cone $D$, the unique 
homogeneous extension $\widetilde f\colon\Cone(K)\to D$ is
Scott-continuous.}
Moreover,
$\widetilde f$ is sublinear, superlinear, or linear, if, and only if, $f$
is. Further, $f\leq g$ if and only if $\widetilde f\leq\widetilde g$.
\end{prop}

\begin{proof}
In the presence of (OC3), Lemma \ref{lem:opba} shows that we have a
linear order embedding of $K$ in $\Cone(K)$ as a lower set \makered{and convexity is evident}.
The embedding \makered{is Scott-continuous}  by Lemma
\ref{lem:continuous} and so $K$ is embedded as a
Scott-closed set.


  
The arguments preceding the statement of the proposition 
show that addition and
scalar multiplication are Scott-continuous on $\Cone(K)$.  
We even have a b-cone: Indeed, if $(x_i)_i$  is a directed family in
$\Cone(K)$ bounded above by some $x$, choose an $r,\ 0 <r<1,$ such
that $rx\in K$. Using that $K$ is a lower set, $(rx_i)_i$ is a
directed family in $K$ and so 
has a sup $y=\sup_i rx_i$ in $K$. We conclude that the $(x_i)_i$ have a sup
namely $r^{-1}y =\sup_ix_i$. 

Now let $f\colon K\to D$ be a homogeneous function into a b-cone
$D$. By Theorem \ref{th:sc2} it has a unique homogeneous extension
$\widetilde f\colon \Cone(K)\to D$. If $f$ is Scott-continuous, then
$\widetilde f$ is Scott-continuous, too: Indeed, let $(x_i)_i$ be a
bounded directed family in $\Cone(K)$ with $x=\sup_i x_i$. Choose any
$r>0$ such that $rx\in K$. Since $K$ is a lower set in $\Cone(K)$ we
also have $rx_i\in K$ and $rx=\sup_irx_i$ in $K$. By the continuity
of $f$ on $K$ we have $f(rx)=\sup_if(rx_i)$, whence $\widetilde
f(x)=r^{-1}\widetilde f(rx)= r^{-1}f(rx) =r^{-1}\sup_if(rx_i)=
r^{-1}\sup_i\widetilde f(rx_i)= \sup_i\widetilde
f(x_i)$. 

The remaining claims follow directly from Theorem \ref{th:sc2}. 
\end{proof}

\makered{In a second step we} use a completion procedure following
Zhang and Fan \cite{ZF}, Keimel and Lawson \cite{KL,KL1} and
\makered{Jung, Moshier, and Vickers \cite{JV}}, in order  
to embed the b-cone $\Cone(K)$ in a d-cone. 

\makered{
\begin{standardconstr}\label{sc3}
{\rm 
A \emph{universal} (or \emph{free}) \emph{dcpo-completion}\footnote{In
the literature \cite{ZF,KL,KL1} the term \emph{dcpo-completion} is
used instead of \emph{universal} 
\emph{dcpo-completion}. For our purposes we prefer the the latter
terminology, since there are Scott-continuous order embeddings of
posets into dcpos which are dense for directed suprema but not
universal, for example the embedding of $\Rp^2$ into the dcpo
obtained by adding a top element.} of a poset
$P$ consists of a dcpo $\overline P$  
and a Scott-continuous map $\xi\colon P\to\overline P$ enjoying the
universal property that every Scott-continuous
map $f$ from $P$ to a dcpo $Q$ has a unique Scott-continuous extension
$\overline f\colon \overline P\to Q$ satisfying $f =\overline f\circ
\xi$. A universal dcpo-completion, if it exists, is evidently unique up to a
canonical isomorphism.

Let us extract relevant information about universal dcpo-completions
of a poset $P$ from the literature\display}
\begin{enumerate}
\item Every poset $P$ has a universal dcpo-completion. {\rm One may, for
  example \cite[Theorem 1]{ZF}, take the least sub-dcpo $\ov{P}$ of
  the dcpo of all 
  nonempty Scott-closed subsets of $P$ (ordered by inclusion) containing
  the principal ideals $\da x$ with $x\in P,$ and $\xi
  \eqdef (x\mapsto \da x)\colon P\to\ov{P}$ as canonical map.}
\item Let $\xi\colon P\to D$ be a topological embedding (for the
  respective Scott topologies) of a poset $P$ into a dcpo $D$. Then
  the least sub-dcpo $\ov{P}$ of $D$ containing the image $\xi(P)$
  together with the corestriction $\xi\colon P\to\ov{P}$ is a universal
  dcpo-completion and the Scott topology of $\ov{P}$ is the subspace
  topology induced by the Scott topology on $D$. {\rm (See \cite[Theorem
  7.4]{KL}).} 
\item A function $\xi\colon P\to D$ of a poset $P$ into a dcpo $D$ is
  a universal dcpo-completion if, and only if,\\
  (i) $\xi$ is a topological embedding (for the Scott topologies)
  and\\
  (ii) the image $\xi(P)$ is dense in $D$\\
  {\rm  (This
  follows from the previous item, since any two universal
  dcpo-completions are isomorphic.)}
\item Let $P_1,\dots,P_n,Q$ be posets and let 
  $\ov{P_1},\dots,\ov{P_n},\ov{Q}$ be universal dcpo-completions
  thereof. Then every Scott-continuous function
  $f\colon P_1\times\dots\times P_n\to Q$ has a unique
  Scott-continuous extension $\ov{f}\colon
  \ov{P_1}\times\dots\times\ov{P_n}\to \ov{Q}$. Further, if $g$ is another
  such function, then $\ov{f}\leq\ov{g}$ if, and only if, $f\leq g$.
  {\rm (Here the products are understood to have the product
      order. Thus, the claim follows from \cite[Proposition 5.6]{KL1},
    since functions 
  defined on finite products are Scott-continuous if, and only if, they
  are separately Scott-continuous in each of their arguments.)}
\item The universal dcpo-completion of a finite direct product of
  posets is the direct product of the universal dcpo-completions of its
  factors. More precisely, if $\xi_i\colon
  P_i\to\ov{P_i}\; (i =1,\dots,n)$  are universal dcpo-completions, then so is
   \[\xi \eqdef \xi_1\times\dots\times\xi_n \colon P_1\times\dots\times P_n\to  \ov{P_1}\times\dots\times \ov{P_n}\]
   %
 {\rm (This follows directly from the
   previous item.)}     
\end{enumerate}
\end{standardconstr}
}

In the characterisation \ref{sc3}(3) of universal dcpo-completions above,
the first condition --- being a topological embedding --- is the
critical one. As we now show, it holds automatically in many
situations. 

\makeblue{To begin with, we remark that, for a Scott-closed subset $C$
  of a poset $P$, the canonical embedding of $C$ into $P$ is
  topological, that is, the intrinsic Scott topology on $C$ is the
  subspace topology induced by the Scott topology on $P$. 

But even on a lower subset $P$ of a dcpo $Q$, the intrinsic Scott
topology of $P$ may be strictly finer than the subspace topology
induced by the Scott topology of $Q$. A simple example for this
phenomenon is given by $P=\Rp\times\Rp$ with the coordinatewise order
and $Q=P^\top$, the dcpo obtained by attaching to $P$ a top
element. Here $[0,1] \times \Rp$ is Scott-closed in $P$, but the Scott
closure of this subset in $Q$ is all of $Q$.} 
   
Following~\cite{KLL03} we say that a dcpo $P$ is \emph{meet continuous} if for any $x \in P$ and any directed set $D \subseteq P$ with  $x \leq \bigvee^{\uparrow} D$, $x$ is in the Scott closure of $\da x \cap \da D$. All domains are meet continuous as are all dcpos with a Scott-continuous meet operation.



%
\begin{lem}\label{lem:meetcont}
The canonical embedding of a lower subset $P$ of a meet continuous
dcpo $Q$ 
is a topological embedding (for the respective Scott topologies).
\end{lem} 
\makeblue{\begin{proof}
One checks that every directed sup in $P$ is also a directed sup in $Q$. Consequently the inclusion of $P$ in $Q$ is Scott-continuous. So, 
for every subset $V$ of $Q$ that is Scott-open \makered{in $Q$}, $V\cap P$ is
open for the Scott topology on $P$. 

In the other direction we have to show: Let $U$ be a
subset of $P$ which is open for the Scott topology on $P$. Then there
is a Scott-open subset $V$ of $Q$ such that $U=V\cap P$. Since 
$\ua U\cap P=U$, it suffices to show that $\ua U$ is Scott-open in
$Q$, where $\ua U$ is the upper set in $Q$ generated by $U$. For this
take any directed subset $D$ in $Q$ with supremum $d$ in $\ua
U$. We have to show that there is a $c \in D$ with $c \in  \ua U$.

To this end, as $d$ in $\ua U$, there is an 
$x\in U$ with $x\leq d$. By meet
continuity, $x$ is in the Scott closure w.r.t. $Q$ of $\da x\cap \da D$. 
Following the remark above on the Scott topologies of closed subsets of partial orders, we see that the intrinsic Scott topology on $\da x$ agrees
with \makered{both} the subspace topology induced by the Scott topology on $Q$, and the subspace topology induced by the Scott topology on $P$.  As 
the $P$-closure of the set $\da x\cap \da D$
intersects 
with the $P$-open set $U$, the set $\da x\cap\da D$ itself
intersects $U$. For an element
$x'$ in the intersection, one has $x'\in U$ and $x'\leq c$ for some
$c\in D$, whence $c\in\ua U$, and we see that this $c$ has the
required properties. 
\end{proof}}

%

The previous lemma yields a sufficient condition for  the 
universality of dcpo-completions:

%

\begin{cor}\label{cor:meetcontinuous}
The canonical embedding of a lower subset $P$ of a meet continuous
dcpo $Q$  
is a universal dcpo-completion if,
and only if, $P$ is dense in $Q$. \qed
%
\end{cor}

We now consider the universal dcpo-completion of an
s-cone.

\begin{prop}\label{prop:sc3}
Let $C$ be an s-cone and $\ov{C}$ a universal dcpo-completion. 
   Then addition and scalar multiplication on $C$ extend 
uniquely to Scott-continuous operations on the dcpo-completion $\overline
C$ which thus becomes a d-cone. The unique Scott-continuous extension
$\ov{f}\colon \ov{C}\to D$ of a Scott-continuous function
$f$ from $C$ to a d-cone $D$ is homogeneous, sublinear,
superlinear, or linear, respectively, if $f$ is. Moreover, $f\leq g$ if
and only if $\ov{f}\leq \ov{g}$.  
\end{prop}   

\begin{proof}
By property \ref{sc3}(4) of universal dcpo-completions, addition and
scalar multiplication on $C$ extend uniquely to  
Scott-continuous operations on $\overline C$. 
As a consequence of  \cite[Proposition 8.1]{KL1}, the extended
operations obey the same equational laws 
as in $C$, that is, $\overline C$ is a d-cone. 
Now let $D$ be any d-cone and $f\colon C\to D$ a Scott-continuous
map. By the universal property, $f$ has a
unique Scott-continuous extension $\ov{f}\colon \ov{C}\to D$. If $f$
is subadditive, let us show that $\ov{f}$ is subadditive, too.
For this we consider the two
maps $g:(a,b)\mapsto f(a+b)$ and $h:(a,b) \mapsto f(a)+f(b)$ from $C\times
C$ to $D$. Clearly,  $(a,b)\mapsto \ov{f}(a+b)$ and $(a,b)\mapsto
\ov{f}(a)+\ov{f}(b)$ are Scott-continuous extensions of $g$ and $h$ to
$\ov{C}\times \ov{C}$. Since $\ov{C}\times \ov{C}$ is the
dcpo-completion of $C\times C$ by property \ref{sc3}(4) of dcpo-completions,
these are the 
unique Scott-continuous extensions $\ov{g}$ and $\ov{h}$,
respectively. The subadditivity of $f$ is equivalent to the statement
that $g\leq h$ which, by the last part of property \ref{sc3}(4) of
dcpo-completions, is equivalent to
$\ov{g}\leq\ov{h}$ which again is equivalent to the subadditivity of
$\ov{f}$. The argument for superadditivity is similar and for
homogeneity it is even simpler. 
\end{proof}


%

We now consider a full \KS\ $K$ and apply the completion procedure above 
to $\Cone(K)$ which, by Proposition
\ref{lem:bcone}, is the universal b-cone  over $K$; we write
$\dCone(K)$ for its universal 
dcpo-completion $\ov{\Cone(K)}$. We know that $K$ is embedded in $\Cone(K)$ as a
Scott-closed convex set. \makeblue{Since $\Cone(K)$ is bounded
  directed complete, it is a lower set in its dcpo-completion. Thus
  $K$ is a lower set in  $\ov{\Cone(K)}$, too. Further, since $K$ is embedded in $\Cone(K)$ and since universal dcpo-completions preserve existing directed sups, $K$ is a sub-dcpo of $\ov{\Cone(K)}$. So
$K$ is Scott-closed in the
universal dcpo-completion $\ov{\Cone(K)}=\dCone(K)$. }

\mycut{\tt GDP added a bit to the above}

\makered{Using Propositions
\ref{lem:bcone} and   \ref{prop:sc3} we then obtain\display}

\begin{thm}\label{prop:KSembed}
Let $K$ be a full \KS. 
 The universal dcpo-completion
$\dCone(K)$ of the b-cone $\Cone(K)$ according to the Standard
Construction \ref{sc2} is a d-cone and $K$ is embedded in this
d-cone as a Scott-closed convex set. 

The embedding 
of $K$ into $\dCone(K)$ is universal in the sense that every
Scott-continuous homogeneous map $f$ from $K$ into a d-cone $D$ has a
unique Scott-continuous homogeneous extension  $\ov{f}\colon
\dCone(K)\to D$. Moreover, the extension $\ov{f}$ is sublinear, superlinear, or
linear if, and only if, $f$ is. Moreover, $f\leq g$ if, and only if, 
$\ov{f}\leq \ov{g}$.   \qed
\end{thm}

\makered{Below, we wish to identify some naturally occurring embeddings as universal. To that end, we begin with a proposition characterising universal embeddings. 
 The \emph{standard factorisation} of a Scott-continuous linear order embedding $K \xrightarrow{e} C$ of a \KS\ in a d-cone is 
 \[K \xrightarrow{u} B \xrightarrow{\xi} C\]
  where $B$ is the sub-cone of $C$ with carrier $\{r e(a) \mid r
  > 1,  a \in K\}$ and the induced order, $u$ is the co-restriction of
  $e$, and $\xi$ is the inclusion. 
\begin{lem} \label{technical} Let $K \xrightarrow{u} B
  \xrightarrow{\xi} C$  be the standard factorisation of a
  Scott-continuous linear order  embedding
  $K \xrightarrow{e} C$ of a full \KS\ $K$ 
in a d-cone. Then $B$ is the free b-cone over $K$, with unit $u$, with respect to Scott continuous homogeneous maps, and $\xi$ is a Scott-continuous linear map. 
\end{lem}

\begin{proof} Clearly $K \xrightarrow{u} B$ is an embedding of an ordered pointed barycentric algebra in an ordered cone. 
Using  Corollary~\ref{free-oCone} and Proposition~\ref{lem:bcone}, we
then see that $B$ is the free b-cone over $K$, with unit $u$, with
respect to both monotone \makeblue{homogeneous maps} and Scott-continuous
homogeneous maps.  

  

The map $\xi$ is a monotone homogeneous map from $B$ into $C$
extending $e$ along $u$ \makeblue{(i.e., $\xi\circ u = e$)}. As $e$ is Scott-continuous and $B$ is the free b-cone over $K$, with unit $u$, with respect to both monotone and Scott-continuous homogeneous maps, $\xi$ is in fact Scott-continuous. It is evidently linear.
%
%
\end{proof}


\begin{prop} \label{d-comp-char} Let $K \xrightarrow{u} B
  \xrightarrow{\xi} C$  be the standard factorisation of an  embedding
  $K \xrightarrow{e} C$ of a  full    
 \KS\ $K$ in a d-cone. Then $e$ is
  universal for Scott-continuous homogeneous maps if, and only if, $C$
  is the universal dcpo-completion over $B$, with unit $\xi$.
\end{prop}

\begin{proof} 
%
In one direction, suppose that $C$ is the universal dcpo-completion
over $B$, with unit $\xi$. By Lemma~\ref{technical} we also have that
$B$ is the free b-cone over $K$, with unit $u$. It then follows from
Proposition~\ref{prop:sc3}  that $C$ is the free d-cone over  $B$,
with unit $\xi$. So $C$ is the free d-cone over  $K$, with unit $\xi\circ u = e$, as required.

In the other direction assume that $e$ is universal, and consider the standard construction of the free d-cone over $K$:
\[K \xrightarrow{u_s} \Cone(K) \xrightarrow{\xi_s} \dCone(K)\] 
As  $\Cone(K)$ is the free b-cone over $K$ with unit $u_s$ and, by
Lemma~\ref{technical}, $B$ is the free b-cone over $K$ with unit $u$,
there is a homogeneous dcpo-isomorphism $\alpha\type \Cone(K) \cong B$ such that $\alpha \circ u_s = u$.
Next, as $\dCone(K)$ is the free d-cone over $\Cone(K)$ with unit $\xi_s$, and as $\alpha$ and $\xi$ are Scott-continuous homogeneous maps (the latter by Lemma~\ref{technical}), there is a Scott-continuous homogeneous map $\beta\type\dCone(K) \rightarrow C $ such that $\beta \circ \xi_s =  \xi \circ \alpha$.
We then have $\beta \circ (\xi_s\circ u_s) =  e$. But, as  $\dCone(K)$ is the free d-cone over $K$ with unit $\xi_s\circ u_s$ and, by assumption, $C$ is the free d-cone over $K$ with unit $e$, it follows that $\beta$ is a dcpo-isomorphism.

Putting all this together, we see that we have dcpo-isomorphisms $\alpha\type \Cone(K) \cong B$ and $\beta\type \dCone(K) \cong C$ such that  
$\beta\circ \xi_s = \alpha \circ \xi$. Then, as  $\Cone(K)
\xrightarrow{\xi_s} \dCone(K)$ is a dcpo-completion, it follows that
$B \xrightarrow{\xi} C$ is a dcpo-completion, as required. 
\end{proof}


As a first application of the proposition, we next give a sufficient
condition for an embedding of a \KS\ in a d-cone to be universal in
the meet continuous case. \makered{The criterion applies in particular
  to continuous d-cones as indicated just before the statement of Lemma
  \ref{lem:meetcont}.} 
%


%
\begin{prop} \label{uni-Keg} Let $e\colon K \rightarrow C$ be a
  Scott-continuous linear order
  embedding of a full \KS\
$K$  in a meet continuous d-cone. Suppose that\display
\begin{enumerate}
\item $e(K)$ is a lower subset of $C$, and
\item 
$B \eqdef \{re(a)\mid  a \in K, r >0\}$ \makered{is dense in $C$}. 
\end{enumerate} 
Then $e$ is universal.
\end{prop}

\begin{proof}  
Endowing $B$ with the ordered  cone structure induced by $C$, we obtain
$K \xrightarrow{u} B \xrightarrow{\xi} C$,
 the standard factorisation of $e$. 
 We may suppose, w.l.o.g.,  that $u$, and so $e$, is an inclusion. 
 As $e(K)$ is a lower subset of $C$, so is $B$. 
 Then, by Proposition~\ref{d-comp-char}, Corollary \ref{cor:meetcontinuous} and the assumption that 
 $B$ is dense in $C$,
 we obtain the desired result.
\end{proof}
}



\makered{As another application of Proposition~\ref{d-comp-char}, we
  check that, given two universal embeddings $K_i
  \xrightarrow{e_i} C_i$ ($i = 1,2$) of  full \KSs\ in  d-cones, their product $K_1 \times K_2 \xrightarrow{e_1 \times e_2} C_1 \times C_2$ is itself a universal embedding (we use the evident definitions of the products of \KSs\  and of d-cones). Let $K_i \xrightarrow{u_i} B_i \xrightarrow{\xi_i} C_i$ be the standard factorisation of  
$e_i$. Then one checks that  the standard factorisation of $e_1 \times e_2$ is $K_1 \times K_2 \xrightarrow{u_1 \times u_2} B_1 \times B_2 \xrightarrow{\xi_1 \times \xi_2} C_1 \times C_2$ (we  use the evident definition of the product of ordered cones). By Proposition~\ref{d-comp-char}, the 
$B_i \xrightarrow{\xi_i} C_i$ are universal dcpo-completions, and so,
by~\ref{sc3}(4), $B_1 \times B_2 \xrightarrow{\xi_1  \times \xi_2} C_1 \times C_2$ is also a universal dcpo-completion. Using Proposition~\ref{d-comp-char} again, we see that $e_1 \times e_2$ is universal.
}

\begin{rem}\label{rem:hist3}{\rm ({\bf Historical Notes and References})
The \emph{abstract probabilistic algebras} of Graham and
Jones~\cite{graham,jones} are barycentric algebras on a dcpo $P$
with a bottom element $0$
such that the map $+: [0,1] \times P^2 = (r,(x,y)) \mapsto x +_r y$ is
continuous, taking the Hausdorff topology on $[0,1]$, the Scott
topology on $P^2 $ 
, and the product topology on 
$[0,1] \times P^2 $. 

Let us add that Jones \cite[Section 4.2]{jones} proves that
an abstract probabilistic algebra can equivalently be defined to be a
dcpo $P$ together 
with  Scott-continuous maps $S_n\times P^n\to P$, $n\in \mathbb N,$
informally written $\sum_{i=1}^nq_ix_i$,
satisfying the equations (A1) and (A2) in remark \ref{rem:hist1}, where
$S_n$ is the domain of subprobability measures on an $n$-element set
as in remark \ref{rem:hist2} with the Scott topology. Jones
explicitly adds the requirement that the operations
$\sum_{i=1}^nq_ix_i$ are commutative in the sense that they are
invariant under any permutation of the indices, a property that other
authors (see the remarks \ref{rem:hist1} and \ref{rem:hist2}) silently
hide in the suggestive notation of a sum.

The notion of an abstract probabilistic algebra $P$ of
  C.\ Jones is equivalent to our notion of a \KS.  Indeed, one
can show that a barycentric algebra over a dcpo is an algebra in this
sense if, and only if, it is a \KS\ in our sense, so the two notions are
equivalent. However, with our definition one can directly use
domain-theoretic methods, for example  for completing \KSs\ to
d-cones.
}  
\end{rem}

\subsection{Preservation results}\label{sec:preserve}

We now turn to the question as to which additional properties of a \KS\ $K$
are inherited by the universal d-cone $\dCone(K)$ over $K$. First we
consider continuity.  
We define a \KS\ $K$ to be \emph{continuous} if it is continuous as a
dcpo.

We will use the following standard lemma:

\begin{lem}\label{lem:aux}
Let $\makered{P}$ be a Scott-closed subset of a continuous poset $\makered{Q}$. Then $\makered{P}$
is a continuous poset and the way-below relation $\ll_\makered{P}$ on $\makered{P}$ is the
restriction of the way-below relation $\ll_\makered{Q}$ on $\makered{Q}$.
\end{lem}

\begin{proof}
Let $x,y$ be elements of $\makered{P}$. Clearly $x\ll_\makered{Q} y$ implies
$x\ll_\makered{P} y$. Thus, if $\makered{Q}$ is continuous, the same holds for
$\makered{P}$. Conversely, let $x\ll_\makered{P} y$. Since $\makered{Q}$ is continuous, the elements
$z\ll_\makered{Q} y$ form a directed set with supremum $y$. Since this directed
set is in $\makered{P}$, there is some
$z\ll_\makered{Q} y$ such that $x\leq z$, whence $x\ll_\makered{Q} y$. 
\end{proof}

In any d-cone, $x\mapsto rx$ is an order isomorphism for $r>0$,
whence $x\ll y \iff rx\ll ry$  for every $r>0$. This statement
is not true for \KSs, in general. 

For elements $x, y$ in a \KS\ $K$ and $0<r<1$ we always have that $rx\ll
ry$ implies $x\ll y$. Indeed if $y\leq \sup_i x_i$, then $ry\leq
r\sup_i x_i  =\sup_i rx_i$, whence $rx\leq rx_i$ for some $i$ which
implies $x\leq x_i$ by Lemma \ref{oc2}. But $x\ll y$ does not imply
$rx\ll ry$ in general, as the counterexample in  Appendix
\ref{subsec:app} shows. 
Fortunately, this difficulty disappears when we require property
(OC3)\display 

\begin{lem}\label{lem:clever}
In a continuous full \KS\ $K$ 
we have $x\ll y$ if and only if $rx\ll ry$ for every $r$ with $0<r<1$. 
\end{lem}

\begin{proof}
Let $K$ be a continuous \KS\ and $0<r<1$. By Property (OC2), the map
$x\mapsto rx$ is an order isomorphism from $K$ onto $rK$ and $rK$ is a
sub-dcpo of $K$. Moreover, if $x$ and $y$ are elements of $K$ such
that $x\ll y$ in $K$, then $rx\ll ry$ in $rK$. Property (OC3) implies
that $rK$ is a lower set, hence, a Scott-closed subset of $K$. By Lemma
\ref{lem:aux}, the way-below relation of the dcpo $rK$ is the
restriction of the way-below relation on $K$. Thus $rx\ll ry$ in $K$. 
\end{proof}

\begin{prop}\label{prop:continuous}
For a full \KS\ $K$, 
 the cone $\dCone(K)$ is
continuous if and only if $K$ is a continuous \KS. \makeblue{If this is the
case, then every element $x\in \dCone(K)$ is the supremum of a directed
family of elements $a_i$ in $\Cone(K)$ with $a_i\ll x$}.
\end{prop}

\begin{proof}
If the d-cone $\dCone(K)$ is continuous, then the Scott-closed subset
$K$ is continuous by Lemma \ref{lem:aux}. Suppose conversely that $K$
is a continuous \KS\ satisfying (OC3). Then $K$ may be considered to
be a Scott-closed convex subset of the ordered cone $\Cone(K)$.

For $x,y$ in $\Cone(K)$, we have $x\ll_{\Cone(K)} y$ iff $rx,ry\in K$
and $rx\ll_K ry$ for some $r>0$. 
Indeed, suppose that $x\ll_{\Cone(K)}
y$. Let $r>0$ be such that $rx,ry\in K$ and let $u_i \in K$ be a 
directed family with $ry\leq\sup_iu_i$. Then $y\leq r^{-1}\sup_i u_i
= \sup_i r^{-1}u_i$ whence $x\leq r^{-1}u_i$ for some $i$ which
implies that $rx\leq u_i$. Conversely, suppose that $rx \ll_K ry$ and
consider a directed family $v_i\in \Cone(K)$ such that $y\leq \sup_i
v_i$. There is an $s>0$ such that $s\sup_i v_i\in K$. Moreover, we may choose
$s< r$ and $s<1$. By Lemma \ref{lem:clever}, we know that $sx
=\frac{s}{r}rx \ll_K  \frac{s}{r}ry = sy$. Since $sy\leq s\sup_i v_i
=\sup_i sv_i$ we conclude that that $sx\leq sv_i$ for some $i$,
whence $x\leq v_i$ by Lemma \ref{oc2}. We conclude that $\Cone(K)$ is a
continuous b-cone.  

For the continuous b-cone $\Cone(K)$, the universal dcpo-completion
$\dCone(K)$ agrees with the round ideal completion which is a
continuous d-cone by \cite[Corollary 8.3]{KL1}. \makeblue{The elements of a
round ideal constitute a directed family of elements way below the
element defined by the round ideal.}  
\end{proof}

Recall that we say that the way-below relation on an ordered cone is
\emph{additive} if
\[a\ll b \mbox{ and } a'\ll b' \implies a+a'\ll b+b'\]
Similarly, in a \KS\ we say that convex combinations preserve the
way-below-relation if
\[a\ll b \mbox{ and } a'\ll b' \implies a +_r a' \ll b +_r
b'\]

\begin{prop}\label{prop:KSembedding1}
Let $K$ be a continuous full \KS. 
 Then 
$\dCone(K)$ has an additive way-below relation if and only if 
 convex combinations preserve the way-below-relation in $K$. 
\end{prop}

\begin{proof}
Suppose first that convex combinations preserve the way-below
relation on $K$. Choose $a \ll a'$ and $b\ll b'$ in $\dCone(K)$. Interpolate
elements $a\ll a''\ll a'$ and $b\ll b''\ll b'$. \makeblue{Then $a''$ and $b''$
belong to $\Cone(K)$, since, by Proposition~\ref{prop:continuous}, $a'$ and $b'$ are suprema of directed
families in $\Cone(K)$, and since $\Cone(K)$ is a lower set in
$\dCone(K)$}. Thus we can find an $r$ with $0<r<1$ such that $ra''\in 
K$ and $rb''\in K$. Since  $ra \ll ra''$ and $rb\ll rb''$ in $K$, we
use the property that convex combinations preserve the way-below
relation to conclude that  $\frac{1}{2}ra+\frac{1}{2}rb 
\ll \frac{1}{2}ra''+\frac{1}{2}rb''$. Multiplying by $2r$ yields
$a+b\ll a''+b''\leq a'+b'$. Thus the way-below relation is additive
on $\dCone(K)$. The converse is straightforward.
\end{proof}

Another noteworthy property is coherence. Recall that a dcpo is called
\emph{coherent} if the intersection of any two Scott-compact saturated
subsets is Scott-compact.

\begin{prop}\label{prop:KSembedding2}
Let $K$ be a continuous full \KS. 
Then the d-cone
$\dCone(K)$ is coherent if, and only if, $K$ is coherent.
\end{prop}

\begin{proof}
Clearly, if $\dCone(K)$ is coherent, then the Scott-closed subset $K$
is coherent, too. For the converse
recall that a continuous poset is said to have property M with respect
to a basis $B$ if, for any $x_1,x_2,y_1,y_2\in B$ with $y_1\ll x_1$
and $y_2\ll x_2$ there is a finite set $F\subset B$ such that $\ua
x_1\cap\ua x_2 \subseteq \ua F\subseteq \ua y_1\cap \ua y_2$. 
By \cite[Proposition III-5.12]{GHK} the following are equivalent for a
continuous dcpo $P$:
\begin{enumerate}
\item $P$ is coherent.
\item $P$ satisfies M for every basis $B$.
\item $P$ satisfies M for some basis $B$. 
\end{enumerate}

\noindent Now let $K$ be a continuous full \KS\ 
which is
coherent. Thus, $K$ satisfies property M with $B=K$. We conclude that
$\Cone(K)$ satisfies property M with $B=\Cone(K)$. Suppose indeed that
$x_1,x_2,y_1,y_2$ are elements of $\Cone(K)$ such that $x_1\ll y_1$
and $x_2\ll y_2$. We can find an $r$ with $0<r<1$ in such a way that $ry_1\in
K$ and $ry_2\in K$. As $rx_1\ll ry_1$ and  $rx_1\ll ry_1$ hold in $K$,
by the coherence of $K$ we can find a finite subset $F$ of $K$ such
that 
$\ua rx_1\cap\ua rx_2 \subseteq \ua F\subseteq \ua ry_1\cap \ua
ry_2$. We then have   
$\ua x_1\cap\ua x_2 \subseteq \ua r^{-1}F\subseteq \ua y_1\cap \ua y_2$.

Since $\Cone(K)$ is a
basis of $\dCone(K)$, we conclude that $\dCone(K)$ is coherent.
\end{proof}

By \cite[Corollary II-5.13]{GHK} we conclude:

\begin{cor}\label{cor:lawsoncompact}
Let $K$ be a continuous full \KS. 
Then $\dCone(K)$ is Lawson compact if, and only if, $K$ is Lawson compact.  \qed
\end{cor}

\subsection{Duality and the subprobabilistic powerdomain} \label{dme}

We next give a brief introduction to duality for d-cones and \KSs,
together  with our main examples,  
function spaces and probabilistic powerdomains. For any dcpos $P$ and
$Q$ we write $Q^P$ for the dcpo of all Scott-continuous maps $f\colon
P \to Q$ with the pointwise order, and note that it is a continuous
lattice whenever $P$ is a domain and $Q$ is a continuous lattice (see,
e.g., \cite{GHK}).

\begin{exa}\label{ex:functionspaces} {\rm
{(d-Cone function spaces)}
Let $P$ be a dcpo. For every d-cone $(C,+,0,\cdot)$, the dcpo  $C^P$ 
is a d-cone when equipped with the pointwise sum and scalar
multiplication\display 
\[(f+g)(x) \eqdef f(x)+g(x), \ \ \ \ (r\cdot f)(x) \eqdef r\cdot f(x)\]
We write $\cL P$ for the d-cone $\oRp^P$ 
of all Scott-continuous functions $f\colon P\to\oRp$. 
We will use the following properties of the d-cone $\cL P$ later on\display
\begin{enumerate}[label=(\alph*)]
\item If $P$ is a domain then 
$\cL P$ is a continuous lattice, hence a continuous d-cone. 

\item For any  domain $P$, the way-below relation of $\cL P$ is additive, if,
and only if, $P$ is coherent (Tix \cite[Propositions 2.28 and 2.29]{TKP09}).
\end{enumerate}
}
\end{exa}

\begin{exa}\label{ex:KSfunctionspaces} {\rm
{(\KS\ function spaces)}
For every \KS\ $(K,+_r,0)$ the dcpo $K^P$ becomes a \KS\
with the pointwise barycentric operations\display
\[(f+_r g)(x) \eqdef f(x) +_r g(x)  \]
and with the constant function with value $0$ as distinguished element.

Taking $K$ to be the unit interval $\makered{\I = }[0,1]$,   we obtain the \KS\
$\makered{\cL_{\leq 1}(P) \eqdef \I}^P$. \makered{It}  is a Scott-closed convex subset of 
$\cL P$ forming a sub-\KS\  of $\cL P$ (considered as a \KS), and
therefore full. 
\makered{Further, $\cL P$ is the
universal d-cone over the \KS\ $\cL_{\leq 1}P$, that is $\cL P \cong \dCone(\cL_{\leq 1} P)$. The inclusion $\cL_{\leq 1}(P) \subseteq \cL(P)$ is the universal embedding, as follows from Proposition~\ref{uni-Keg}, noting that $\cL P$
has a continuous meet, and every $f \in \cL P$ 
is the sup of the  sequence of bounded functions $p_n \! \wedge \! f$, where $p_n$ 
is the projection sending $\oRp$ to $[0,n]$.}
Finally, $\cL_{\leq 1} P$ is  a domain if $P$ is,  and its way-below relation is
inherited from that of $\cL P$ and is
 preserved by the barycentric operations if, and only if,
$P$ is coherent.



}
\end{exa}

For any two d-cones $C$ and $D$, the Scott-continuous linear functions $f\colon
C\to D$ form a sub-d-cone $\Lin(C,D)$ of the function space $D^C$.
Similarly, for any  
\KS\ $K$ and d-cone $D$, the set  $\Lin(K,D)$ of Scott-continuous
linear functions is a sub-d-cone of $D^K$. Note that if $K$ is $C$, 
regarded as a \KS, then $\Lin(C,D)$ and $\Lin(K,D)$ are identical.

Duality will play an important r\^{o}le: every d-cone $C$ and \KS\ $K$ has
a dual d-cone, viz.\   $C^*\eqdef \Lin(C,\oRp)$ and $K^*\eqdef
\Lin(K,\oRp)$, respectively. 

By Theorem \ref{prop:KSembed}, if $K$ is full 
then every   Scott-continuous linear functional $f\colon K\to\oRp$ has a unique
Scott-continuous linear extension $\widetilde f\colon C\to\oRp$, where
$C=\dCone(K)$, and  so
$f\mapsto \widetilde f$ is a natural d-cone isomorphism between $K^*$ and
$C^*$. We will use this isomorphism freely.

Duality leads to the 
 \emph{weak Scott topology} on a d-cone $C$,
 which, as the name implies,  is coarser than the Scott topology.
 It is the coarsest topology on $C$ for which the
Scott-continuous linear functionals on $C$ remain continuous.  The sets
\[U_f \eqdef \{x\in C\mid f(x)> 1\},\ f\in C^*\]
form a subbasis of this topology.  
 
We need the notion of a reflexive cone. For any d-cone $C$ we have
a canonical map $\ev_C$ from $C$ into the  
bidual $C^{**}$. It
assigns the evaluation map $f \mapsto f(x)$ to every $x\in
C$ and  is Scott-continuous and linear. We say that $C$ is 
\emph{reflexive} if $\ev_C$ is an isomorphism of
d-cones. If $C$ is a reflexive d-cone, then its dual $C^*$ is also
reflexive, and its dual is (isomorphic to) $C$.

We will need, as we did in \cite{KP}, the notion of a \emph{convenient}
d-cone. This is a continuous reflexive d-cone $C$ whose weak
Scott topology agrees with its Scott topology, and whose dual
$C^*$ is continuous and has an additive way-below relation. The
valuation powerdomain construction, which we consider
next, exemplifies these strong requirements.   


\begin{exa}\label{ex:Probabilisticpowerdomains}{\rm 
 {(Probabilistic powerdomains)}
Probability measures on dcpos are modelled by valuations, which assign a
probability to Scott-open subsets in a continuous way. This permits
the development of a satisfactory theory of integration of 
Scott-continuous functions. For large classes of dcpos,
Scott-continuous valuations correspond to regular Borel probability
measures and so it does not make much difference whether one works
with Scott-continuous valuations or probability measures. For the sake
of applications in semantics it is useful not to restrict to probabilities, where the
whole space has probability $1$, but to admit subprobabilities, where
the whole space has  probability of at most $1$. For theoretical
purposes, on the other hand, it is more convenient to extend the notion
of  probability to that of a measure, where the measure of the whole
space can be any nonnegative real number, or even $+\infty$. 

A \emph{valuation} on a dcpo $P$ is a map $\mu$ defined on the
complete lattice $\cO P$ of Scott-open subsets of $P$ taking nonnegative
real values, including $+\infty$, 
which (replacing finite additivity) is strict and modular\display
\begin{eqnarray*}
\mu(\emptyset)&=&0\\
\mu(U\cup V)+\mu(U\cap V) &=&\mu(U) +\mu(V)\ \ \mbox{ for all } U,V\in\cO P
\end{eqnarray*}\\
\makeblue{The \emph{point} (or \emph{Dirac}) valuations $\delta_x$ ($x \in P$) are given by:
\[\delta_x(U) = \left\{\begin{array}{ll}
                                     1 & (x \in U)\\
                                     0 & (x \notin U)\\
                                \end{array}\right .\]
The \emph{simple} valuations are the finite linear combinations of point valuations.}
The set $\cV P$ of all Scott-continuous valuations $\mu\colon \cO
P\to\oRp$ forms a sub-d-cone of the function space $\oRp^{\cO P}$, with the
pointwise  partial order and  with pointwise addition and scalar multiplication\display 
\[
(\mu + \nu)(U) \eqdef \mu(U) + \nu(U), \ \ \ (r\cdot\mu)(U) \eqdef r\cdot \mu(U)
\]
We call it the \emph{valuation powerdomain} over $P$ (it is also known as the \emph{extended probabilistic powerdomain}).
The  Scott-continuous valuations $\mu$ such that $\mu(P)\leq 1$ are
called \emph{subprobability valuations}; they form a Scott-closed
convex subset 
$\cV_{\leq 1}P$ of the d-cone $\cV P$, and hence a full  sub-\KS. 
We call it the \emph{subprobabilistic powerdomain} over $P$.  
We will use the following results later on\display


\begin{enumerate}[label=(\alph*)]
\item Let $P$ be a domain. 
The valuation powerdomain $\cV P$ is a continuous d-cone
with an additive way-below relation and (hence)
the subprobabilistic powerdomain $\cV_{\leq 1}P$ is a continuous \KS\
in which the barycentric operations preserve its way-below relation, which is inherited from $\cV P$
(Jones \cite[Theorem 5.2 and Corollary 5.4]{jones}, Kirch \cite{Ki},
Tix \cite{Tix}
and see~\cite[Theorem IV-9.16]{GHK}). 

\item Let $P$ be a coherent domain. Then the powerdomains $\cV P$ and
$\cV_{\leq 1}P$ are also coherent (Jung and Tix \cite{JT98}, Jung \cite{J04}).
\makeblue{In each case the relevant simple valuations form a basis.}

\item For any dcpo $P$ there is a canonical Scott-continuous map
$\delta_P\colon P \to \cV_{\leq 1}P \subseteq \cV P$ which 
assigns to any $x\in P$ the Dirac measure 
$\delta_P(x)$ which has value $1$ for all Scott-open  neighbourhoods
$U$  of $x$ and value $0$ otherwise. 
Then, if $P$ is a domain, $\cV P$ is the free d-cone over $P$ (Kirch
\cite{Ki}, and see~\cite[Theorem 
  IV-9.24]{GHK}), and $\cV_{\leq 1}P$ is the free
 \KS\ over $P$ (Jones \cite[Theorem 5.9]{jones}). 
%
%
 That is,  for every d-cone $C$ and every
 Scott-continuous function $f\colon P\to C$, there is a unique
 Scott-continuous linear map $\ov{f}\colon \cV P\to C$  such that the
 following diagram commutes\display 
 \begin{center}
\begin{diagram}P&&\\
\dTo<{\delta_P} & \SE_{}{\quad \;\; \; f}&\\
\cV P &\rTo_{\ov{\!f}} & C
\end{diagram}
\end{center}
 and for every
 \KS\ $K$ and every Scott-continuous map $f\colon P\to K$ there is a
 unique Scott-continuous linear map $\ov{f}\colon \cV_{\leq 1}P \to K$
 such that the following diagram commutes\display
 \begin{center}
\begin{diagram}P&&\\
\dTo<{\delta_P}& \SE_{}{\quad \;\; \; f}&\\
\cV_{\leq 1} P &\rTo_{\ov{\!f}}&K 
\end{diagram}
\end{center}
 
\item Let $P$ be a domain. Then the valuation powerdomain $\cV P$ is the
universal d-cone over the \KS\ $\cV_{\leq 1}P$, that is:
$\cV P \cong \dCone(\cV_{\leq 1}P)$ \makered{(the inclusion $ \cV_{\leq 1} P \subseteq \cV P$ is the universal arrow).} This is because $\cV P$  and
$\cV_{\leq 1}P$ are, respectively, the free d-cone  and the free \KS\
over $P$. 



\item For any dcpo $P$,  the valuation powerdomain $\cV P$ is
the dual of $\cL P$ up to isomorphism, and, if $P$ is a domain, the d-cone
$\cL P$  is the dual of $\cV P$, which implies that both $\cV P$ and
$\cL P$ are reflexive (Kirch \cite[Satz 8.1 and Lemma 8.2]{Ki}, and Tix
\cite[Theorem 4.16]{Tix}).

The proof is based on the appropriate
notion of an integral of a Scott-continuous function $f\colon
P\to\oRp$ with respect to 
a Scott-continuous valuation $\mu$. Among the various ways to define
this integral, the most elegant is via the Choquet integral\display
\[ \int f\,d\mu \eqdef
\int_0^{+\infty}\mu\big(f^{-1}\big(]r,+\infty]\big)\big)\,dr\, \]   
The Choquet integral should be read as the generalised Riemann
integral of the nonnegative monotone-decreasing function $r \in
[0,+\infty[ \, \mapsto\,  \mu\big(f^{-1}\big(]r,+\infty]\big)\big)$,
which may take infinite values. The integral was originally defined as
a Lebesgue integral by Jones~\cite{jones}; Tix later proved the
Choquet definition equivalent~\cite{Tix}. 

The isomorphism between $\cV P$ and $(\cL P)^*$, the dual  of $\cL P$,  is
the map  
$\mu\mapsto \lambda f.\ \int f\,d\mu$, establishing a kind of Riesz
representation theorem. 
And if $P$ is a domain, the isomorphism between $\cL P$ and the dual
$(\cV P)^*$ is 
the map $f\mapsto \lambda \mu.\ \int f\,d\mu$.

\item For a domain $P$, the Scott topology and the weak Scott topology agree
on both $\cL P$ and $\cV P$ (Tix \cite[Satz 4.10]{Tix}).
\end{enumerate}
}

Summarising the properties reported in Examples \ref{ex:functionspaces}
and  \ref{ex:Probabilisticpowerdomains}, we can say that the
valuation powerdomain $\cV P$ 
over a continuous coherent domain $P$ is a convenient d-cone.
\end{exa}

\section{Power \KSs} \label{Powerkeg}

We now construct three kinds of power \KSs, proceeding analogously to the
constructions of the three types of powercone in~\cite{TKP09}.  Under
various assumptions on a given \KS\ $K$, we will construct its
\emph{lower, upper}, and \emph{convex power \KSs},  $\cH K$, $\cS K$, and
$\cP K$. 
Power \KSs\ and powercones have a choice operation, so we begin by defining the three kinds of \KSs\  and cones enriched with a choice operation that thereby arise.

A \emph{\KS\ semilattice} is a
\KS\ equipped with a Scott-continuous semilattice operation $\nonor$
over which convex combinations distribute, that is, for all $x,y,z \in
K$ and $r \in [0,1]$ we have: 
\[x +_r (y \nonor z) = (x +_r y) \nonor (x +_r z)\]
It is a \emph{\KS\  join-semilattice} if $\nonor$ is the
binary supremum operation (equivalently, if $x \leq x \nonor y$ always
holds). 
It is a \emph{\KS\ meet-semilattice} if $\nonor$
is the binary infimum operation (equivalently, if $x \nonor y \leq x$
always holds). A morphism of  \KS\ semilattices   is a morphism of
\KSs\ which also preserves the semilattice operation. 
\makered{Using the distributivity axiom it is straightforward to show that the semilattice operation is homogeneous and that the barycentric operations are $\subseteq$-monotone (a function  between two  semilattices is said to be
\emph{$\subseteq$-monotone} if it preserves the partial order $\subseteq$ naturally associated to semilattices).  
Further, } 
the following
\emph{convexity} identity holds 
\[x \nonor (x +_r  y) \nonor y = x \nonor y \eqno{\rm (CI)} \]
This can be proved beginning with the equation $x \nonor  y = (x
\nonor  y) +_r (x \nonor  y)$, then expanding out the right-hand
side using the distributivity of $+_r $ over $\nonor$, and then using the inclusion $\subseteq$ associated with the semilattice operation.

There is an analogous notion of \emph{d-cone semilattice}. This is a
d-cone $C$ equipped with a Scott-continuous semilattice operation $\nonor$
over which the cone operations distribute, i.e., for all $x,y,z \in C$ and $r\in \Rp$ we have:
\[x + (y \nonor z) = (x + y) \nonor (x +z) \qquad r\cdot (x \nonor y) = r\cdot x \,\nonor\, r\cdot y\]
Such a cone is a \emph{d-cone  join-semilattice}  (\emph{d-cone meet-semilattice}) if $\nonor$
is the binary supremum operation (respectively, the binary infimum operation).
Every  d-cone semilattice (d-cone  join-semilattice, d-cone meet-semilattice) can be regarded as a     \KS\ semilattice  (respectively,  \KS\  join-semilattice,    \KS\ meet-semilattice). 
\makered{A morphism of  d-cone semilattices   is a morphism of d-cones which also preserves the semilattice operation. Much as before, distributivity implies that $\nonor$ is homogeneous, and that the d-cone operations are $\subseteq$-monotone.}

The lower, upper, and convex power \KSs\   will be, respectively, \KS\  join-semilattices,
meet-semilattices, and  semilattices.
Possibly under further assumptions, they will be the
free such \KS\ semilattices on $K$.
\makered{The analogous result in the lower case for powercones was proved in ~\cite{TKP09}. Freeness results were also proved in the other two cases, but of a different character, having  weaker assumptions and weaker conclusions.}




%
In order to verify the properties of the various power 
\KSs, we will embed the \KSs\ in d-cones, and then use
the embeddings to transfer results from~\cite{TKP09} about powercones
to the power \KSs.

\makered{Another way to proceed is to view the power
\KSs\ as retracts  of the powerdomains of the domains underlying the
\KSs. One can then transfer results about the  
powerdomains (such as the preservation of continuity) to the \KSs. Similar uses of retracts already occur
in ~\cite{TKP09,GL07, GL08,GL10,GL12}.} We have not explored this 
option, but it is certainly  possible to strengthen some of our
results in this way. However, 
the results we present are sufficient for their intended application
in Section~\ref{dPDomains} to mixed powerdomains combining
probabilistic choice and nondeterminism.

\subsection{Lower power \KSs} \label{HpK}
\makered{We first investigate the \emph{convex lower} (or \emph{Hoare})
\emph{power \KS}. We need some closure properties of convex sets\display
\begin{lem} \label{convexity} Let $K$ be a \KS. Then:
 \begin{enumerate}
 \item Any directed union of convex subsets of $K$ is convex.
 \item The Scott closure of a convex subset of $K$ is convex.
 \item If $X$ and $Y$ are convex subsets of $K$ then so is
 \[X +_r Y \eqdef \{x +_r y \mid  x \in X, y \in Y\}\]
 \end{enumerate}
\end{lem}
\begin{proof} 
The first statement is immediate. For the second, note that the Scott closure of a set is obtained by repeating the operations of downwards closure and taking directed unions transfinitely many times. As each of these operations can be seen to preserve convexity, and as taking directed unions does too, it follows that the Scott closure of a convex set is convex. Finally the third statement follows using the entropic law  (E).
\end{proof}}

The lower power \KS\  $\cH K$,
of a given full \KS\  $(K, +_r,0)$, that is, a
  \KS\ satisfying Property (OC3), consists of
the collection of non-empty Scott-closed convex subsets of $K$,
ordered by subset, with zero $\{0\}$ and with convex combination
operators ${+_r}_H$ given by: 
\[X {+_r}_H Y \eqdef \overline{X {+_r} Y}\]
for $r \in [0,1]$, (where, for any $X \subseteq K$, $\overline{X}$ is
the closure of $X$ in the Scott topology). 
That these operators
\makered{are well-defined follows from Lemma~\ref{convexity}.}

As we said above, in order to verify the properties of $\cH K$ we will
make use of the embedding of $K$ into a d-cone $C$  and the properties
of the lower powercone (or Hoare powercone) of $C$. So let us begin by
reviewing the definition and properties of the lower powercone   $\cH
C$ of a d-cone $(C,+,0, \cdot)$~\cite[Section 4.1]{TKP09}. 
As a partial order, it is the collection  
of all nonempty Scott-closed convex subsets of $C$ ordered by
inclusion $\subseteq$. It has arbitrary suprema, given by: 
\[\bigvee_{i\in I}X_i = \cchull{\bigcup_{i\in I}X_i}\]
with directed suprema  given by:
\[\sideset{}{^{\uparrow}}\bigvee_{i\in I} X_i = \overline{\sideset{}{^{\uparrow}}\bigcup _{i\in I} X_i}\]
Addition and scalar multiplication are lifted from $C$  to  $\cH C$ as follows:
\[ X +_{H} Y \eqdef \overline{X+Y}
\qquad
 r \cdot_{H} X \eqdef r\cdot X\]
where  $X + Y = \{x + y \mid x \in X, y \in Y\}$, and $r \cdot X = \{r \cdot x \mid x \in X\}$. Convex combinations are then given by 
$r\cdot_H X +_H (1 - r)  \cdot_H Y =  \overline{X+_rY} $.
Further, the following  is proved in~\cite[Section 4.1]{TKP09}:
\begin{thm} \label{Hcone}
Let  $(C, +, 0, \cdot)$ be a d-cone. 
Then $(\cH C, +_{H} , {0},  \cdot_{H})$ is also  a d-cone, 
and, equipped with binary suprema, it forms a d-cone  join-semilattice.

If $C$ is continuous, then so is $\cH C$. 
The  non-empty finitely generated convex Scott-closed sets
$\cchull{F}$, where $F$ is a finite, non-empty subset of $C$, form a
basis for $\cH C$; further, for any $X, Y \in \cH C$,  $X \ll_{\cH C}
Y$ if, and only if, $X \subseteq \cchull{F}$ and $F \subseteq  \dda
Y$, for some such $F$. 
If, in addition, the way-below relation of $C$ is additive, so is that
of $\cH C$.  \qed
\end{thm}

We can now show:
\begin{thm} \label{HKeg} Let  $(K, +_r,0)$ be a  full   
  \KS. 
 Then $(\cH K, {+_r}_H, \{0\})$ is a full \KS.
%
It has arbitrary suprema, given by:
\[\bigvee_{i\in I}X_i = \cchull{\bigcup_{i\in I}X_i}\]
with directed suprema  given by:
\[\sideset{}{^{\uparrow}}\bigvee_{i\in I} X_i = \overline{\sideset{}{^{\uparrow}}\bigcup _{i\in I} X_i}\]
and, equipped with binary suprema, it forms a  \KS\  join-semilattice.

If, further, $K$ is  a continuous \KS, then so is $\cH K$. 
The  non-empty finitely generated convex Scott-closed sets
$\cchull{F}$, where $F$ is a finite, non-empty subset of $K$, form a
basis for $\cH K$; further, for  any $X, Y \in \cH K$,  $X \ll_{\cH K}
Y$ if, and only if, $X \subseteq \cchull{F}$ and $F \subseteq  \dda
Y$, for some such $F$. 
If, in addition, the way-below relation of $K$ is closed under convex
combinations, so is that of $\cH K$. 
\end{thm}
\begin{proof}
Using Theorem~\ref{prop:KSembed} we can
  regard $K$ as a Scott-closed convex subset of the d-cone $C \eqdef \dCone(K)$, with its
  partial order and algebraic structure inherited from that of $C$.

It is then immediate that $\cH K$ embeds as a sub-partial order of
$\cH C$. Further, as $K$ is Scott-closed and convex, one easily shows,
using the above formulas for the suprema and convex combinations of
$\cH C$, that $\cH K$ is a Scott-closed convex subset of $\cH C$,
bounded above by $K$.  Therefore $\cH K$ has all sups and a \KS\
structure, and they are inherited from $\cH C$. Explicitly, arbitrary
sups and 
and directed sups are given by the claimed formulas, as the
corresponding formulas hold for $\cH C$ (we may, equivalently, take
Scott closure of subsets of $K$ with respect to $C$ or $K$). The
inherited convex combinations and zero agree with those of $\cH C$. 
Convex combinations distribute over the semilattice operation (here
binary sups), as $+$ and $r\cdot-$ do. So, equipped with binary
suprema, $\cH K$ is a  \KS\  join-semilattice. It automatically
satisfies Property (OC3) as it is embedded in the d-cone $\cH C$ as a
Scott-closed 
subset. 

Next, suppose that $K$ is a  continuous \KS. Then, by Proposition~\ref{prop:continuous}, $C$  is continuous, and so, by
Theorem~\ref{Hcone}, $\cH C$ is too. As $\cH K$ is a Scott-closed subset
of $\cH C$, it follows that $\cH K$ is continuous with way-below
relation the restriction of that of $\cH C$ to $\cH K$.   
That the non-empty finitely-generated Scott-closed sets form a basis
of $\cH K$, and the characterisation of the way-below relation of $\cH
K$ then follow from the corresponding parts of Theorem~\ref{Hcone}. 
Further, as $r\cdot - $ preserves the way-below relation of $\cH C$
($r\cdot - $ is an order-isomorphism of cones), it also preserves the
restriction to $\cH K$, and so $\cH K$ is continuous as a \KS.  

Finally, suppose additionally that $K$'s way-below relation is closed
under convex combinations. Then, by Proposition~\ref{prop:KSembedding1},  \makered{$C$ has an additive way-below
relation}. We then see from~\cite{TKP09} that the d-cone operations of
$\cH C$ preserve its way-below relation, and so that convex
combinations preserve the way-below relation of $\cH K$ (it is the
restriction of that of $\cH C$  as $\cH  K$ is a closed subset of $\cH
C$).  
\end{proof}

We now show that $\cH K$ is the free  \KS\  join-semilattice over any
full \KS\ $K$, 
with unit the evident \KS\ morphism $\eta_H: K
\rightarrow \cH K$, where $\eta_H(x) = \da x$. 
\begin{thm} \label{Huni} Let $K$ be a full \KS. Then the map $\eta_H$ is universal. That is,
  for every \KS\  join-semilattice  $\makered{L}$ and \KS\
  morphism $f:K \rightarrow \makered{L}$  
there is a unique  \KS\ semilattice  morphism $f^{\dagger}: \cH K
\rightarrow \makered{L}$ such that the following diagram commutes: 
{\[\begin{diagram}
	K\\
	 \dTo^{\eta_H} & \SE_{}{\quad \;\; \; f}\\
	\cH K & \rTo^{f^{\dagger}} & \makered{L}
	\end{diagram}\]
}
The morphism is given by:
  \[f^{\dagger}(X) =  \sideset{}{_{\!\!\makered{L}}}\bigvee f(X)\]
\end{thm}
\begin{proof} \makered{To show uniqueness, 
choose an $X \in \cH K$. It can be written as a non-empty  sup, viz.\  $\bigvee_{x \in X} \eta_H(x)$.
Then, noting that  continuous binary  join morphisms preserve all non-empty  sups (for such  sups are directed  sups of finite non-empty  sups, and finite non-empty  sups are iterated binary  ones) we calculate:}
\[\begin{array}{lcl}
f^{\dagger}(X) & = & f^{\dagger}(\bigvee_{x \in X} \eta_H(x))\\
                        & = & \bigvee_{x \in X}f^{\dagger}(\eta_H(x))\\
                        & = & \bigvee_{x \in X}f(x)
 \end{array}\]
 \makered{To show} existence we therefore set
  \[f^{\dagger}(X) =  \bigvee_{x \in X}f(x)\]
   and verify that it makes the diagram commute and  is both a \KS\
   and a semilattice map. The first of these requirements holds as we
   calculate: 
 \[f^{\dagger}(\eta_H(x)) = f^{\dagger}(\{y \mid y \leq x\}) =  \bigvee_{y \leq x}f(y) = f(x) \]

 For the second we  need to show that $f^{\dagger}$ is strict and
 continuous and  preserves convex combinations. Strictness is a
 consequence of the diagram commuting, as both $\eta_H$ and $f$ are
 strict. 
  \makered{For continuity we first show that 
  \[\bigvee  f(\overline A) =  \bigvee f(A)\]
 for any $A \subseteq K$. Choose   $A \subseteq K$. We evidently have $\bigvee  f(A) \leq  \bigvee f(\ov{A})$, and it remains to prove the converse inequality $\bigvee  f(\ov{A}) \leq  \bigvee f(A)$.
 By the continuity of $f$, we have $f(\overline A)\subseteq
 \overline{f(A)}$ whence $\bigvee  f(\overline A)\leq \bigvee
 \overline{f(A)}$. So we only have to show that  $x \leq \bigvee f(A)$ for any $x \in  \overline{f(A)}$. This follows from the fact that the Scott closure of a set is obtained by transfinitely many repetitions of the operations of downwards closure and taking directed sups.  }

 %
 %

 We can now calculate:
 \[\begin{array}{lcl}
f^{\dagger}(\bigvee^{\uparrow}_{i\in I}X_i ) & = &  f^{\dagger}( \overline{\bigcup^{\uparrow}_{i\in I}X_i})\\
 & \makered{=} & \makered{\bigvee   f(\ov{\bigcup^{\uparrow}_{i\in I}X_i})}\\
 & \makered{=} & \makered{\bigvee  f(\bigcup^{\uparrow}_{i\in I}X_i)}\\
  & = & \bigvee^{\uparrow}_{i\in I}\bigvee_{x \in X_i}f(x) \\
  & = &\bigvee^{\uparrow}_{i\in I}f^\dagger(X_i) \\ &&
 \end{array}\]
 %
    %
For convex combinations we calculate:
 \[\begin{array}{lcl}
f^{\dagger}(X {+_r}_H Y) & = & f^{\dagger}(\overline{X+_rY}) \\
 & = &  f^{\dagger}(X+_rY) \\
 & = & \bigvee_{z \in X+_rY}f(z)\\
 & = & \bigvee_{x \in X, \,y \in Y} f(x +_r y)\\
 & = & \bigvee_{x \in X, \,y \in Y} f(x) +_r f(y)\\
 & = & (\bigvee_{x \in X} f(x)) +_r (\bigvee_{y \in Y} f(y))\\
 & = & f^{\dagger}(X) +_r f^{\dagger}(Y)\\
    \end{array}\]
\makered{The second equation holds as $f$ is Scott-continuous.
The sixth
equation holds as convex combinations in $\makered{L}$ distribute over
arbitrary \makered{non-empty} sups, as we may see by analysing such sups as before, i.e., as directed sups of iterated binary ones.}

Finally, we need to show that binary sups are also preserved, and so calculate:
\[\begin{array}{lcl}
f^{\dagger}(X \vee Y) & = &  f^{\dagger}(\cchull{(X \cup Y)})\\
                                   & = &  f^{\dagger}(\myconv(X \cup Y))\\
                                   & = &  \bigvee_{x \in X,\, y \in Y}\bigvee_{0\leq r \leq 1}f(x +_r  y)\\
                                   & = &  \bigvee_{x \in X,\, y \in Y}\bigvee_{0 < r < 1} (f(x) \vee f(x +_r  y) \vee f(y))\\
                                   & = &  \bigvee_{x \in X,\, y \in Y}\bigvee_{0 < r < 1}(f(x) \vee (f(x) +_r  f(y)) \vee f(y))\\                                                                      
                                   & = &  \bigvee_{x \in X,\, y \in Y} f(x) \vee f(y)\\    
                                   & = &  \bigvee_{x \in X}f(x) \vee  \bigvee_{ y \in Y}f(y)\\
                                   & = &  f^{\dagger}(X) \vee f^{\dagger}(Y)\\
\end{array}\]
where the sixth equation follows using the convexity identity (CI). 
\end{proof}

\subsection{Upper  power \KSs} \label{SpK}

We next investigate the \emph{convex upper} (or \emph{Smyth})
\emph{power \KS}\ $\cS 
K$, of a given continuous full \KS\  $(K, +_r,0)$. 
 The restriction to the continuous case is necessary as in
  the case of powercones. It consists of the collection of non-empty
  Scott-compact convex saturated (= upper) subsets of $K$,
  ordered by reverse inclusion $\supseteq$
, with zero $K$ and with convex combination
  operators ${+_r}_S$  given by: 
\[X {+_r}_S Y \eqdef \ua\,(X {+_r} Y)\]
for $r \in[0,1]$. To see that these operators  are well-defined, first
note that $X +_r Y$ is Scott-compact, as it is the image under $+_r$ of $X
\times Y$, which is compact in the product topology on $K \times K$,
which   latter is the same as the Scott topology, as $K$ is
continuous. Then note that the upper closure of a Scott-compact
(convex) set is Scott-compact (respectively convex). 

The upper power \KS\ has binary infima, which make it a  \KS\ meet-semilattice. They are given by: 
\[X \wedge Y =  \ua \myconv (X \cup Y)\]

In order to verify the properties of $\cS K$, we follow our general
methodology, using the embedding of $K$ into a d-cone $C$  and  the
properties of the upper powercone (or Smyth powercone) of $C$. Let us
begin by recalling the definition and properties of the upper
powercone   $\cS C$ of a continuous d-cone $(C,+,0,
\cdot)$~\cite[Section 4.2]{TKP09}.  
It consists of all nonempty Scott-compact convex saturated subsets
ordered by reverse inclusion $\supseteq$. It has directed suprema
given by intersection: 
\[\sideset{}{^{\uparrow}}\bigvee_{i\in I} X_i = \sideset{}{^{\downarrow}}\bigcap _{i\in I} X_i\]
and binary infima given by:
\[X \wedge Y = \ua\myconv(X \cup Y)  \]
Addition and scalar multiplication are lifted from $C$ to $\cS C$ as follows:
\[ X +_S Y \eqdef \ua\,(X + Y)
\qquad
r \cdot_S X \eqdef  \ua\,(r \cdot X) \]
Note that $r \cdot_S X = \ua \{0\} = C$ if $r = 0$ and $r \cdot_S X = r \cdot X$ if $r > 0$. 
Convex combinations are given by $r\cdot_S X\, +_S\, (1 - r)  \cdot_S Y =  \ua\,(r\cdot X+  (1 - r) \cdot Y)$.
Further, the following is proved in~\cite[Section 4.2]{TKP09}:

\begin{thm} \label{Scone} Let $(C, +, 0, \cdot)$ be a continuous
  d-cone. Then $(\cS C, +_S , C, \cdot_S)$ is a continuous d-cone 
and, equipped with binary infima, it forms a d-cone meet-semilattice.

The  non-empty finitely generated convex saturated Scott-compact sets
$\ua \myconv {F}$, where $F$ is a finite, non-empty subset of $C$, form
a basis for $\cS C$; further, for any $X, Y \in \cS C$,  $X \ll_{\cS
  C} Y$ if, and only if, $X \supseteq \ua \myconv {F}$ and $ \dua F
\supseteq Y$, for some such $F$. 
If the way-below relation of $C$ is additive,  so is that of $\cS C$.  \qed
%
\end{thm}

For the proof of the next theorem we recall that a \emph{(monotone) retraction pair} between two partial orders $P$ and $Q$ is a pair of monotone maps
\[P \xrightarrow{e} Q \xrightarrow{r} P\]
such that  $r \circ e = \id_P$; it is a \emph{(monotone) closure pair} if, additionally, $e \circ r  \geq \id_Q$.
We can now show: 
\begin{thm}\label{SKeg}  Let $(K, +_r, 0)$ be a continuous
  full \KS. 
 Then $(\cS K, {+_r}_S , K)$ 
is a continuous \KS\  meet-semilattice.  Directed suprema  are given by
intersection: 
\[\sideset{}{^{\uparrow}}\bigvee_{i\in I} X_i = \sideset{}{^{\downarrow}}\bigcap _{i\in I} X_i\]
and binary infima are given by:
\[X \wedge Y = \ua\myconv(X \cup Y)  \]
The  non-empty finitely generated convex saturated Scott-compact sets
$\ua \myconv {F}$, where $F$ is a finite, non-empty subset of $K$, form
a basis for $\cS K$; further, for any $X, Y \in \cS K$,  $X \ll_{\cS
  K} Y$ if, and only if, $X \supseteq \ua \myconv {F}$ and $ \dua F
\supseteq Y$, for some such $F$.
\makered{The  way-below relation $\ll_{\cS K}$ is preserved by $r\cdot_{\cS K} - $, and if}
 the way-below relation of $K$ is closed under convex combinations,
so is that of $\cS K$. 
If $K$ is coherent then $\cS K$ is a bounded-complete domain, hence
coherent too. 
\end{thm}
\begin{proof} Using Theorem~\ref{prop:KSembed}, we can regard $K$ as a Scott-closed convex
  subset of the d-cone $C \eqdef \dCone(K)$, with its partial
  order and algebraic structure inherited from that of $C$. By Proposition~\ref{prop:continuous}, $C$ is continuous, and so, by
  Lemma~\ref{lem:aux}, 
  the way-below relation of $K$ is also inherited from that of $C$. 

To relate $\cS K$ to $\cS C$ we first define a Scott-closed convex
subset $\sub$ of $\cS C$ and then show that $\cS K$ is a closure of
$\sub$ (with partial ordering inherited from $\cS C$). This enables us
to transport structure from $\cS C$ to  $\cS K$ via $\sub$. We take
$\sub$ to be the collection of elements of $\cS C$ intersecting
$K$. It is evidently a lower set \makered{(for $\leq\; = \;\supseteq$)}.
If $X_i$ is a directed subset
of $\sub$ then $X_i \cap K$ is a $\subseteq$-filtered collection of
non-empty saturated Scott-compact subsets of $K$ and so has a non-empty
intersection by a consequence of the Hofmann-Mislove Theorem
\cite[Corollary II-1.22]{GHK}. This shows that $\bigvee_{\cS C} X_i$ is in
$\sub$. The convexity of $L$ follows from that of $K$. 

We order $\sub$ by reverse inclusion $\supseteq$
, i.e., as a sub-partial order of $\cS C$.
As $\sub$ is a Scott-closed subset of $\cS C$,  it is a continuous
sub-dcpo of $\cS C$, with way-below relation inherited from that of
$\cS C$, and with basis $B \cap \sub$, for any basis $B$ of $\cS C$.    
Further, as $\sub$ is a convex Scott-closed subset of $\cS C$, it inherits a 
\KS\ structure from $\cS C$, with zero $C$, and with convex
combinations given by $X{+_r}_{\sub} Y =  r\cdot_S X +_S (1 - r)  \cdot Y =
\ua\,(r\cdot X+  (1 - r) \cdot Y) $. 
Binary infima are also inherited by $\sub$, and as these distribute
over $+_C$ and $\cdot_C$, $\sub$ forms a  \KS\ meet-semilattice.  

We now define a monotone closure pair:
\[\cS K \xrightarrow{e } \sub  \xrightarrow{c} \cS K \]
by setting $e(X) \eqdef \ua_C X$ and $c(Y) \eqdef Y \cap K$. 

Since $(e,c)$ is a  monotone closure pair, the
existence of directed suprema in $L$ implies the existence of directed
suprema in $\cS K$ and $c$ preserves these directed suprema, that is,
$c$ is Scott-continuous. The following shows that the supremum of a
directed collection $X_i$ in $\cS K$ is calculated as expected:
\[\sideset{}{^{\uparrow}}\bigvee_i X_i =
c(\sideset{}{^{\uparrow}}\bigvee_i e(X_i)) =
(\sideset{}{^{\downarrow}}\bigcap_i \ua_C X_i)\cap K =
\sideset{}{^{\downarrow}}\bigcap_i (\ua_C X_i) \cap K =
\sideset{}{^{\downarrow}}\bigcap_i X_i\]  

For the continuity of $e$, we have to show for any
\makered{directed collection $X_i$ in $\cS K$}
  that: 
\[\ua_C(\bigcap_i X_i) \supseteq (\bigcap_i \ua_C X_i)\]
(the other direction holds as $e$ is monotone). Suppose, for the sake
of contradiction, that there is a $y \in C$, with $y$ in every $\ua_C
X_i$, but not in $\ua_C(\bigcap_i X_i)$. We then  have that $\bigcap_i X_i
\subseteq \{x \in C \mid x \not\leq y\} \cap K $. \makered{As the  latter set is
Scott-open in \makered{$K$}, and as $K$ is well-filtered  (this follows from
  the Hofmann-Mislove theorem, see~\cite[Theorem II-1.21]{GHK}),} 
 \makered{there is} an $i$ such that $X_i \subseteq \{x \in C \mid x \not\leq
y\} \cap K $, which contradicts our assumption. Thus $(c,e)$ is a
Scott-continuous closure pair and we can conclude from the continuity
of $\sub$ that $\cS K$ is a continuous dcpo.

Turning to the characterisation of the way-below relation on $\cS K$,
choose $X,Y$ in $\cS K$. By general properties of retractions we know
that $X \ll_{\cS K} Y$ holds iff there is a $U \in \sub$ such that $ X
\leq c(U)$ and $U \ll_{\sub} \ua_C Y$ (equivalently $ e(X) \leq U$ and
$U \ll_{\sub} \ua_C Y$, as $(e,c)$ is an adjoint pair). As the
way-below relation on $\sub$ is the restriction of that on $\cS C$,
Theorem~\ref{Scone} tells us that $U \ll_{\sub} \ua_C Y$ holds iff
there is a non-empty finite subset $F$ of $C$ such that $U \supseteq
\ua_C \myconv F$ and $\dua_C F \supseteq \ua_C Y$.
Putting these together, we have that $X \ll_{\cS K} Y$  holds iff
there is a non-empty finite subset $F$ of $C$ such that  
$\ua_C X \supseteq \ua_C \myconv F$ and $\dua_C F \supseteq \ua_C Y$
(equivalently, $\ua_C X \supseteq \myconv F$ and $\dua_C F \supseteq
Y$).   

As $Y \subseteq K$ and  $\dua_C F \supseteq Y$, and $K$ is a
Scott-closed subset of $C$, $F \cap K$ is non-empty. 
So we have that $X \ll_{\cS K} Y$  holds if, and only if, there is a
non-empty finite subset $F'$ of $K$ such that  
$\ua_C X \supseteq \myconv F'$ and $\dua_C F' \supseteq Y$ (in one
direction, given $F$, set $F' = F \cap K$; in the other direction,
given $F'$, take $F = F'$). 
For such an $F'$ we have $\ua_C X \supseteq \myconv F'$ iff $X \supseteq
\myconv F'$ (as $K$ is a convex subset of $C$) and  
$\dua_C F' \supseteq Y$ iff $\dua_K F' \supseteq Y$, and we have
established the desired characterisation of $\ll_{\cS K}$.

Using this characterisation, and the fact that $K$ is continuous, it follows immediately that $\ll_{\cS K}$ is preserved by
$r\cdot_{\cS K} - $. It also follows immediately that the non-empty
finitely generated convex saturated sets 
form a basis of $\cS K$. 

Next, a calculation now shows that $e$ preserves the convex
combination operation: 
\[\begin{array}{lcl}
e(X) {+_r}_S \, e(Y) & =  & \ua_C (r\cdot (\ua_C X) + (1 - r)\cdot (\ua_C Y))\\
                                & =  &  \ua_C (\ua_C r\cdot X + \ua_C (1 - r)\cdot Y) \\
                                & =  &  \ua_C ( r\cdot X +  (1 - r)\cdot Y) \\
                                & =  &  \ua_C \ua_K ( r\cdot X +  (1 - r)\cdot Y) \\
                                & =  & e (X {+_r}_\makered{K}  \,  Y) \\
\end{array}\]
It follows that the convex combination operation on $\cS K$ can be
defined in terms of that on $\sub$ as we have:  
$X {+_r}_S  \,  Y = ce(X {+_r}_S  \,  Y) = c( e(X) {+_r}_{\sub}  \,  e(Y))$.
So, as $e,c,{+_r}_{\sub}$ are all Scott-continuous, so is ${+_r}_{\cS K}$.

As $(c,e)$ is a closure pair and $\sub$ has binary infima, so does
$\cS K$ and $e$ preserves them. For any $X,Y \in \cS K$ we can then
calculate: 
\[\begin{array}{lcl}
X \wedge Y & = & c(e(X) \wedge e(Y)) \\
                   & = & (\ua_C \myconv (\ua_C X \cup \ua_C Y)) \cap K\\
                   & = & (\ua_C \myconv \ua_C  (X \cup  Y))\cap K\\
                   & = & (\ua_C \myconv  (X \cup  Y))\cap K\\
                   & = & \ua_K \myconv  (X \cup  Y)\\
\end{array}\]
showing that binary meets are given as required. We have also seen
that $\wedge_{\cS K}$ can be defined as a composition of
$e,c,\wedge_{\sub}$, and so is Scott-continuous. 

\makered{As $e$  is an order-mono (i.e., it reflects the partial order) and as it preserves convex combinations and binary
meets, any inequations between these operations holding in $\sub$ also hold in $\cS K$. So (also using the fact that both convex combinations and binary meets are monotone) $\cS K$ is an ordered barycentric algebra, binary meets form a meet-semilattice, and convex combinations distribute over binary meets.
Therefore, as  convex combinations and binary
meets are both Scott-continuous, and as $r\mapsto ra$ is
Scott-continuous, we see that $\cS K$ is a \KS\ meet semilattice with zero $K$.}
%
We also know
that \makered{$\cS K$ is} continuous and \makered{so it}
is a continuous \KS. 

Next, we show that convex combinations in $\cS K$ preserve its
way-below relation, assuming the same is true of $K$.  
Suppose that $X \ll_{\cS K}Y$ and $X' \ll_{\cS K}Y'$.
Using the
characterisation of $\ll_{\cS K}$ we see that there are finite,
non-empty  $F, F' \subseteq K$ such that $X \supseteq \ua \myconv {F}$,
$ \dua F \supseteq Y$, $X' \supseteq \ua \myconv {F'}$, and $ \dua F'
\supseteq Y'$, and then that we need only  show   that $X {+_r}_S X'
\supseteq \ua \myconv (F {+_r} F')$ and $ \dua (F {+_r} F') \supseteq Y
{+_r}_S Y'$, i.e., that $\ua\,(X {+_r} X') \supseteq \ua \myconv (F {+_r}
F')$ and $ \dua (F {+_r} F') \supseteq \ua\,(Y {+_r} Y')$. The first of
these requirements holds as we have: 
 \[\ua\,(X {+_r} X') \supseteq \ua\,(\ua \myconv {F} {+_r} \ua \myconv {F'}) \makered{\supseteq}
 \ua\,(\myconv {F} {+_r}  \myconv {F'})  \makered{\supseteq} \ua \myconv( {F {+_r}   F'})\] 
 \makered{(with the last inclusion holding because of the entropic law). The} second holds as convex combinations in $K$ preserve $\ll_K$.

 Finally if $K$ is 
also \makered{(Scott-)coherent}, i.e., if the intersection of two Scott-compact saturated
sets is again such, then $\cS K$ is a bounded complete dcpo.
\makered{So, as it is also continuous, it is coherent (see~\cite[Proposition III-5.12]{GHK}).}
\end{proof}

 We note that the proof also establishes that the map 
 \[u =  X   \,\mapsto\, \ua{X} \colon \cS K \to \cS (\dCone(K))\]
 is a d-cone meet semilattice embedding (that is, it is a d-cone semilattice morphism that reflects the partial order).

Theorem~\ref{SKeg} does not assert that $\cS K$ satisfies Property (OC3)\makered{, but only that $r\cdot_{\cS K} - $ preserves $\ll_{\cS K}$.
Indeed, $\cS K$ need not satisfy Property (OC3)}\display
\begin{fact} The upper power \KS\ of the subprobabilistic
  powerdomain $\cV_{\leq 1} \{0,1\}$ of the two-element discrete partial order
  does not satisfy Property {\rm (OC3)}. 
\end{fact} 
\begin{proof} For notational convenience we replace $\cV_{\leq 1} \{0,1\}$  by
  the isomorphic  \KS\ obtained from $\mathbb{S} = \{(r,s)\in  [0,1]^2
  \mid r + s \leq 1\}$ by ordering it coordinatewise and equipping it
  with the evident convex combination operators.  

Set $Y \in \cS \mathbb{S}$ to be $\{(1,0)\}$, and take any $0 < r < 1$
and set $X' \in \cS  \mathbb{S}$ to be the saturated convex closure of
$\{(0,1), (r,0)\}$, which is: 
\[\{ (\overline{r}(1 - s),s)\mid s \in [0,1], r \leq \overline{r} \leq 1\}\] 

Clearly $X' \supseteq \ua\,(r\cdot Y)$, i.e., $X' \leq r\cdot_S
Y$. Suppose, for the sake of contradiction, that $X' = r\cdot_S X =
\ua\,(r\cdot X)$ for some $X \in \cS   \mathbb{S}$. Then, as $(0,1) \in
X' = \ua r\cdot X $ we have $(0,1) \geq r\cdot x$ for some $x \in
X$. As $r\cdot x \in r\cdot_S X = X'$, $r \cdot x$ has the form
$(\overline{r}(1 - s),s)$ for some $s \in [0,1]$ and  $r \leq \overline{r}
\leq 1$.  
Then as $(0,1) \geq r\cdot x = (\overline{r}(1 - s),s)$, and so $0 \geq
\overline{r}(1 - s)$, we see that  $s =1$ and so that 
$r \cdot x = (0,1)$. But this cannot be the case as $r < 1$. 
\end{proof}

\makered{The failure of Property (OC3) is a priori a problem, as it obstructs the iteration of the upper \KS\ construction $\cS$. However it is not a problem for this paper as iterating the upper mixed powerdomain does not involve iterating $\cS$.}

We next show that $\cS K$ is the free \KS\ meet-semilattice
 over \makered{any 
 \KS\ $K$ satisfying suitable assumptions}. 
 The
unit $\eta_S: K \rightarrow \cS K$ 
is the evident \KS\ morphism $\eta_S(x) \eqdef \ua_K x$.

\makered{%
\begin{lem} \label{upper_way_below_convex} Let $K$ be a continuous
 full \KS\ 
in which convex combinations preserve the way-below relation. Suppose that  $F,G$ are non-empty subsets of $K$ such that $\dua G \supseteq F$. Then $\ua \myconv G\ll_{\cS K}  \ua \myconv F$.
\end{lem}
\begin{proof}
As convex combinations preserve $\ll_K$ and $\dua G \supseteq  F$, we have $\dua \myconv G \supseteq  \ua \myconv F$. Then, using the compactness of $\ua \myconv F$, we see that there is a non-empty finite subset $H$ of $ \myconv G$ such that  $\dua H \supseteq  \ua \myconv F$.
The conclusion follows by the characterisation of $\ll_{\cS K}$ given in Theorem~\ref{SKeg}.
\end{proof}}

\begin{thm} \label{Suni} \makered{Let $K$ be a continuous full \KS\ 
in which convex combinations preserve the way-below relation.
Then the} map $\eta_S$ is universal. That is,
  for every \KS\ meet-semilattice  $\makered{L}$ and \KS\
  morphism $f:K \rightarrow \makered{L}$ there is a unique  \KS\ semilattice 
  morphism $f^{\dagger}: \cS K \rightarrow \makered{L}$ such that the following
  diagram commutes: 
{\[\begin{diagram}
	K\\
	 \dTo^{\eta_S} & \SE_{}{\quad \; \;\;f}\\
	\cS K & \rTo^{f^{\dagger}} & \makered{L}
	\end{diagram}\]
}
The  morphism is given by:
\[f^{\dagger}(X) = \sideset{}{^{\uparrow}}\bigvee \{ \bigwedge f(F) \mid  F \subseteq_{\mathrm{fin}} K,  F\neq\emptyset, 
\makered{\ua \myconv {F}\ll_{\cS K} X}
\}   \] 
\end{thm}
\begin{proof} 
\makered{Using the basis of  $\ll_{\cS K}$ given in Theorem~\ref{SKeg}},
  for any $X  \in \cS
  K$ we have 
\[X =  \sideset{}{^{\uparrow}}\bigvee \{ \bigwedge \eta_S(F) \mid F \subseteq_{\mathrm{fin}} K, F \neq \emptyset,  
\makered{\ua \myconv {F}\ll_{\cS K} X}
\} \]
where we make use of the easily proved fact that for any finite
non-empty subset $F$ of $K$, we have: $\ua \myconv F  = \bigwedge_{b \in
  F} \eta_S(b)$.   

It then follows for any  \KS\ semilattice  morphism $f^{\dagger}$
which makes the diagram commute that 
\[f^{\dagger}(X) = \sideset{}{^{\uparrow}}\bigvee \{ \bigwedge  f(F) \mid F \subseteq_{\mathrm{fin}} K, F \neq \emptyset,  
\makered{\ua \myconv {F}\ll_{\cS K} X}
\}  \] 
establishing uniqueness.

For existence we define $f^{\dagger}$ by means of this formula and
verify that it makes the diagram commute and is both a \KS\ and
a semilattice map.  
For continuity, it is evident that $f^{\dagger}$ is monotone, and so
it suffices to show that for any directed set $X_i$, $i \in I$, in $\cS K$ and any
finite, non-empty $F \subseteq K$ with 
$\makered{\ua \myconv {F}\ll_{\cS K} \bigvee_i X_i }
$
we have: 
\[\bigwedge f(F)  \leq \bigvee_i  \{ \bigwedge f(\makered{G})
\mid \makered{G} \subseteq_{\mathrm{fin}} K, \makered{G} \neq \emptyset,  
\makered{\ua \myconv {G}\ll_{\cS K} X_i}
\}\] 
\makered{This holds as if $\ua \myconv {F}\ll_{\cS K} \bigvee_i X_i$ then $\ua \myconv {F}\ll_{\cS K} X_i$, for some $i$.}

Next, it is helpful to prove that
\[f^{\dagger}(\bigwedge \eta_S(F)) =  \bigwedge f(F) \tag{$*$}\]
for any finite non-empty set $F$, 
that is, that:
\[\sideset{}{^{\uparrow}}\bigvee \{ \bigwedge f(G) \mid  G
\subseteq_{\mathrm{fin}} K, G \neq \emptyset,  
\makered{\ua \myconv {G}\ll_{\cS K} }
\ua \myconv F \}  =  \bigwedge f(F)\]
\makered{To show the left-hand side is $\leq$ the right-hand side, suppose we have a non-empty finite subset $G$ of $K$ such that
 $\ua \myconv {G}\ll_{\cS K} 
 \ua \myconv F $. 
 Then, by the characterisation of $\ll_{\cS K}$ given in Theorem~\ref{SKeg}, there is a finite non-empty $H \subseteq K$ such that
 $\ua \myconv {G} \supseteq \ua \myconv {H}$ and $ \dua H \supseteq \ua \myconv F$
So, for any  $a \in F$, there is a $b \in H$ such that $b \ll a$, and so a $c \in \myconv{G}$ such that $c \leq a$. Let $G_1$ be the set of such $c$'s.
 We then have\display
 \[\bigwedge f(G)  = \bigwedge f(G)  \wedge \bigwedge f(G_1) \leq  \bigwedge f(G_1) \leq \bigwedge f(F)\]
 where the equality follows from the convexity identity \rm (CI).
 This shows the left-hand side is $\leq$ the right-hand side.}


\makered{Conversely, given $F = \{a_1,\ldots, a_n\}$,  with $n > 0$, choose $b_1 \ll
a_1,\ldots, b_n \ll a_n$ and take $G = \{b_1,\ldots, b_n\}$. 
By Lemma~\ref{upper_way_below_convex} we have $\ua \myconv G\ll_{\cS K}  \ua \myconv F$.
So the left-hand side is $\geq
\bigwedge \makered{f}(G)$, and so $\geq \bigwedge \makered{f}(F)$, as $G$ consists of an
arbitrary choice of an elements way-below each element of $F$. }


Taking $F$ to be a singleton in $(*)$, we see that, as required, the
diagram commutes; this, in turn, implies that $f^{\dagger}$ is strict,
as $\eta_S$ and $f$ are. As regards preservation of the semilattice
operation $\wedge$, as every element is a directed supremum of
non-empty finite infima \makered{of elements of the form $\eta_S(b)$}and as  $\wedge$ is Scott-continuous, we need only
verify it for \makered{such} non-empty finite infima, and that follows immediately
from $(*)$. 

We finally show that $f^{\dagger}$ preserves convex combinations.
Since $f^{\dagger}$ is Scott-continuous and every element of $\cS K$ is a
directed supremum of meets of the form $\bigwedge_{b \in F} \eta_S(b)$
($F \subseteq K$ non-empty and finite), it suffices to show that
$f^{\dagger}$ preserves convex combinations of such finite meets. To
that end, given $F,G \subseteq K$ non-empty and finite, we calculate:
\[\begin{array}{lcl}
f^{\dagger}(\bigwedge_{b \in F} \eta_S(b) {+_r}_{\cS K} \bigwedge_{c \in G} \eta_S(c)) 
                                 & = & f^{\dagger}(\bigwedge_{b \in F, c \in G} (\eta_S(b) \, {+_r}_{\cS K}\, \eta_S(c)))\\
                                 & = & f^{\dagger}(\bigwedge_{b \in F, c \in G}  \eta_S(b +_r c))\\
                                 & = & \bigwedge_{b \in F, c \in G} f(b +_r c)\\
                                 & = & \bigwedge_{b \in F, c \in G} (f(b) +_r f(c))\\
                                 & = & \bigwedge_{b \in F}  f(b)  +_r  \bigwedge_{c \in G}  f(c)\\
                                 & = & f^{\dagger}(\bigwedge_{b \in F} \eta_S(b)) +_r f^{\dagger}(\bigwedge_{c \in G} \eta_S(c)) 
\end{array}
\]
where the third and sixth equalities follow from $(*)$, and  the first
and fifth follow from distributivity.
\end{proof}

\makered{This result contrasts with the corresponding universality result for upper powercones in~\cite{TKP09}. There the assumptions are weaker, but so are the conclusions: there is no assumption of preservation of the way-below relation, but the universality relates only to continuous d-cone semilattices, not to all of them. Further the proof methods for the two theorems are different. It would be interesting to know if the assumption made in Theorem~\ref{Suni} that convex combinations preserve the way-below relation is needed.}

\subsection{Convex power \KSs} \label{CpK}

We next investigate the \emph{convex} (or \emph{Plotkin}) \emph{power \KS}\ $\cP K$, of a
given continuous and coherent (so Lawson compact) full \KS.   
 Note that we have to suppose not
only continuity but also coherence in order to prove the desired
results. First we need some definitions from~\cite{TKP09}. Nonempty
Lawson-compact order-convex subsets of a Lawson-compact domain are
called \emph{lenses}. Both Scott-closed sets and saturated
Scott-compact sets are lenses, as they are both Lawson-compact, and
every lens $X$ can be written as the intersection of  a non-empty
Scott-closed convex set and a non-empty  Scott-compact saturated
convex one, as we have: $X = \overline{X} \cap \ua X$. We also have
$\overline{X} = \da X$ for any lens $X$. If a lens  $X$ of a
continuous Lawson-compact \KS\ is also convex, then so are $\da X$ and
$\ua X$.  
The Egli-Milner  ordering is defined on order-convex subsets of a
partial order $\leq$ by: 
\[X\leqem Y\ \equiv_{\mathrm{def}}\  \forall x \in X.\, \exists y \in
Y.\, x \leq y \,\wedge\, \forall y \in Y.\, \exists x \in X.\, x \leq
y \] 
which can equivalently be written as
\[X\leqem Y\ \equiv\ \da X \subseteq \da  Y \, \wedge\,  \ua  Y \subseteq \ua X \]

We define 
$\cP K$ to be the collection of convex lenses of $K$ ordered by the
Egli-Milner ordering,  
with zero $\{0\}$ and with convex combination operators  ${+_r}_P$ given by:
\[X {+_r}_P Y \eqdef (\da X {+_r}_H \da  Y)  \, \cap \, (\ua  X
{+_r}_S \ua  Y)  \] 
for $r \in [0,1]$. It follows from the above remarks on lenses that
this operator is well-defined. Using the explicit definitions of
convex combinations for the lower and upper power \KSs, one sees that
$X {+_r}_P Y  = \overline{X +_r  Y}  \, \cap \, \ua\, (X +_r  Y)$.  
The convex power \KS\  is a  \KS\ semilattice  when equipped with the
semilattice operator $\nonor_P$ defined by: 
\[X \nonor_P Y \eqdef  (\da X \vee_{\cH K} \da  Y)  \, \cap \, (\ua  X
\wedge_{\cS K} \ua  Y) \]
Using the explicit definitions of the semilattice operations for the
lower and upper power \KSs\ one sees that $X \nonor_P Y =  \cchull{ (X
  \cup Y)}  \, \cap \, \ua \myconv (X  \cup  Y) $; note too that $\da\,(X
\nonor Y) = \da X  \vee_{\cH K}  \da Y$ and $\ua\,(X \nonor_P Y) = \ua X
\vee_{\cS K}  \ua Y$.

In order to verify the properties of $\cP K$ we proceed as before,
via  embeddings into cones. 
Let us begin by recalling the definition and properties of the convex
powercone   $\cP C$ of a continuous Lawson-compact d-cone $(C,+,0,
\cdot)$~\cite[Section 4.3]{TKP09}. This is the collection of all  convex
lenses of $C$ partially ordered by the Egli-Milner ordering.  
It has directed suprema  given by:
\[\sideset{}{^{\uparrow}}\bigvee_{i \in I} X_i =  (\sideset{}{^{\uparrow}}{\bigvee}_{i \in I} \da X_i) \,\cap\, (\sideset{}{^{\downarrow}}{\bigvee}_{i \in I} \ua X_i)\]
where, on the right, we take directed suprema in $\cH C$ and $\cS C$,
respectively. 
More explicitly, we have:
\[\sideset{}{^{\uparrow}}\bigvee_{i \in I} X_i =
(\overline{\sideset{}{^{\uparrow}}{\bigcup}_{i \in I} \da X_i})
\,\cap\, (\sideset{}{^{\downarrow}}{\bigcap}_{i \in I} \ua X_i)\] 
Addition and scalar multiplication are lifted from $C$  to  $\cP C$ as follows:
\[X +_P Y \eqdef  (\da X +_H \da Y)\, \cap \, (\ua X +_S \ua Y) \qquad 
 r\cdot_P  X \eqdef  (r \cdot_H  \da X)\, \cap\, (r \cdot_S  \ua X)\]
Using the explicit definitions of addition and scalar multiplication
in the lower and upper powercones, these definitions simplify to: 
\[X+_P Y= \overline{X+Y} \, \cap \, \ua\, (X+Y) \qquad 
r\cdot_P X=r\cdot X\]
Convex combinations are given by\display
 \[r\cdot_P X\, +_P\, (1 - r)  \cdot_P Y =
\overline{r\cdot X + (1 - r)\cdot Y} \, \cap \, \ua\,(r\cdot  X +  (1 - r)\cdot
Y) \]
 There is also a 
Scott-continuous semilattice operation. It is defined by:
 \[X \nonor_P Y \eqdef (\da X \vee_{\cH C} \da Y) \, \cap \, (\ua X
 \vee_{\cS C}  \ua Y)\] 
which simplifies to $X \nonor_P Y =  \cchull{ (X \cup Y)}  \, \cap \,
\ua \myconv (X  \cup  Y) $. 
Further, the following is proved in~\cite[Section 4.3]{TKP09}:

\begin{thm} \label{Pcone} Let $(C, +, 0, \cdot)$ be a continuous
  coherent  d-cone. Then $(\cP C, +_P , \{0\}, \cdot_P)$ 
is also a continuous coherent d-cone, 
and, equipped with  the semilattice operation $\nonor_P$, it forms a
d-cone semilattice. 

The finitely generated convex lenses $k_C(F) \eqdef  \cchull{F} \,
\cap \ua \myconv F$, where $F$ is a finite, non-empty subset of $C$,
form a basis for $\cP C$, and, for any $X,Y \in \cP C$,  we have $X
\ll_{\cP C} Y$ if, and only if, $X\leqem k_C(F)$ and $F \subseteq \dda
Y$ and $  \dua F \supseteq  Y$ (i.e.,  
$F\llem Y$) for some such $F$.
If the way-below relation of $C$ is additive,  so is that of $\cP C$. \qed
\end{thm}

We can now show: 
\begin{thm}\label{PKeg}  Let $(K, +_r, 0)$ be a continuous
  coherent full \KS. 
 Then \\ $(\cP K, {+_r}_P ,\{0\})$ is also a continuous coherent full   
 \KS\     
 and,
equipped with the Scott-continuous semilattice operation $\nonor_P$, it
forms a \KS\ semilattice.  Directed suprema  are given by: 
\[\sideset{}{^{\uparrow}}\bigvee_{i \in I} X_i =
(\overline{\sideset{}{^{\uparrow}}{\bigcup}_{i \in I} \da X_i})
\,\cap\, (\sideset{}{^{\downarrow}}{\bigcap}_{i \in I} \ua X_i)\] 
The finitely generated convex lenses $k_K(F)  \eqdef  \cchull{F} \,
\cap \ua \myconv F$, where $F$ is a finite, nonempty subset of $K$,
form a basis for $\cP K$, and, for any $X,Y \in \cP K$,  we have $X
\ll_{\cP K} Y$ if, and only if, $X\leqem k_K(F)$ and $F \subseteq \dda
Y$ and $\dua F \supseteq Y$ (i.e., $F\llem Y$) for some such $F$. 
If, in addition, the way-below relation of $K$ is closed under convex
combinations, so is that of $\cP K$. 
\end{thm}

\begin{proof} Applying Theorem~\ref{prop:KSembed}  we can regard
  $K$ as a Scott-closed convex subset of the
  d-cone $C \eqdef \dCone(K)$, with its partial order and algebraic structure
  inherited from that of $C$. Applying Propositions~\ref{prop:continuous} and~\ref{prop:KSembedding2}, we see that $C$ is continuous and coherent. It follows that $K$ is a sub-dcpo of $C$,  that
  its way-below relation is inherited from that of $C$, that a subset
  of $K$ is Scott-compact in the topology of $K$ if, and only if, it
  is Scott-compact in $C$, and that a subset of $K$ is a lens of $K$ if, and
  only if, it is a lens of $C$. 

We therefore see that $\cP K$ is a subset of $\cP C$. It also
evidently inherits its partial order from that of $\cP C$. We next
show that $\cP K$ is a Scott-closed convex subset of $K$. To see it is a lower set,
 suppose  $X\leqem Y \in \cP K$.  Then $X \subseteq
\da_C X \subseteq \da_C Y \subseteq K$ (the last as $Y \subseteq K$
and $K$ is a lower set), and so $X \in \cP K$. For closure under
directed suprema, suppose $X_i$ is a directed subset of $\cP K$. Then
$\sideset{}{^{\uparrow}}{\bigvee_{i \in I}} X_i =  Y \cap Z$, where 
\[Y \eqdef (\overline{\sideset{}{^{\uparrow}}{\bigcup}_{i \in I} \da_C
  X_i}) \qquad \mbox{and} \qquad Z \eqdef
\sideset{}{^{\downarrow}}{\bigcap}_{i \in I} \ua_C X_i\] 
with the closure being taken in $C$. As $X_i \subseteq K$, $Y$ is in
fact a Scott-closed subset of $K$. Therefore the directed supremum is
a subset of $K$ and so a lens of $K$, as required, and we have shown
that $\cP K$ is a Scott-closed subset of $\cP C$. It follows in
particular that $\cP K$ is a sub-dcpo of $\cP C$.  
Noting that we can write $Y$ equivalently taking the lower closure
and the topological closure in $K$, and that we can write $Z \cap K$
as $\bigcap_{i \in I} \ua_K X_i$ we further see that directed suprema
in $\cP K$ are given as claimed.    

For convex closure, recall that convex combinations are given by 
\[\overline{r\cdot X + (1 - r)\cdot Y} \, \cap \, \ua_C (r\cdot  X +  (1 - r)\cdot Y)\]
where the closure is taken in $C$. Taking $X, Y \in \cP K$ we see that
$r\cdot X + (1 - r)\cdot Y$ is a subset of $K$, as $K$ is a convex subset
of $C$. So the closure can equivalently be taken in $K$ and we find
that the convex combination is a subset of $K$ and so, as required,
a lens in $K$. Intersecting  $\ua_C (r\cdot  X +  (1 - r)\cdot Y)$ with
$K$, we note that we can write the convex combination equivalently as
$\overline{r\cdot X + (1 - r)\cdot Y} \, \cap \, \ua_K (r\cdot  X +
(1 - r)\cdot Y)$ with the closure taken in $K$. Thus the convex combination
operators of $\cP K$ are the same as those inherited from $\cP C$.

As $\cP K$ is a Scott-closed convex subset of $\cP C$, it inherits a
continuous coherent \KS\ structure satisfying Property  (OC3) from $\cP
C$, with way-below relation the restriction of that of $\cP C$, and
with basis $B \cap \cP K$, where $B$ is any basis of $\cP C$. As the
zero of $\cP K$ is evidently that of $\cP C$, we see from the above
that the partial order and algebraic structure defined on $\cP K$ is
that inherited from $\cP C$, and so $\cP K$ is indeed a continuous
coherent \KS\ satisfying Property  (OC3). 

One checks that the operation $\nonor$ defined on $\cP K$ is the
restriction of $\nonor_P$ to $\cP K$. It is therefore, as claimed a
Scott-continuous semilattice operation. 
Further as $+_P$ and $r\cdot_P -$ both distribute over $\nonor_P$, we
see that, equipped with $\nonor$, $\cP K$ is, as claimed, a
\KS\ semilattice.

The finitely generated convex lenses $k_C(F) = \cchull{F} \, \cap
\ua_C \myconv F$ that are in $\cP K$ form a basis of $\cP K$. As then $F
\subseteq k_C(F)\subseteq  \cP K$, the closure can be equivalently be
taken in $K$, and we see that $k_C(F) = k_K(F)$. So $\cP K$ has a
basis as claimed. The characterisation of $\ll_{\cP K}$ can then be
read off from the characterisation of $\ll_{\cP C}$, as   $\ll_{\cP
  K}$ is the restriction of $\ll_{\cP C}$ to $\cP K$.  

If  $\ll_K$ is closed under convex combinations, we can assume by
Proposition \ref{prop:KSembedding1} that
$\ll_C$ is closed under sums. Then $\ll_{\cP C}$ is also closed under
sums, and so, too, under convex combinations. As $\cP K$ inherits
convex combinations and its way-below relation  from $\cP C$, we see
that  $\ll_{\cP K}$ is closed under convex combinations, concluding
the proof. 
\end{proof}

\newcommand{\leqemc}{\leqem^{\mathrm{c}}}

We next show that $\cP K$ is the free  \KS\ semilattice  over \makered{any
 \KS\ $K$ satisfying suitable assumptions}. The unit
$\eta_P: K \rightarrow \cP K$
is the evident \KS\ morphism $\eta_P(x) \eqdef \{x\}$.
\makeblue{We first need two lemmas.}

\begin{lem} \label{leqemc} Let $K$ be a continuous  coherent full \KS.
Suppose that $H$, $G$, $F$ are non-empty finite subsets of $K$ such that $H \llem G$ and $k_K(G) \ll_{\cP K}  k_K(F)$. Then there are finite sets $H_1 \subseteq \myconv H$ and 
$F_1 \subseteq \myconv F$ such that $H \cup H_1 \leqem F \cup F_1$. 
\end{lem}
\begin{proof}
By the characterisation of $\ll_{\cP K}$ given in Theorem~\ref{PKeg}, there is a  non-empty finite $I \subseteq_{\mathrm{fin}} K$ such that $k_K(G) \leqem k_K(I)$ and $I \llem k_K(F)$.

We first show:
\[\forall c \in H.\, \exists a \in \myconv F.\, \makeblue{c} \leq a \tag*{$(*)$}\]
Choosing $c \in H$, as $H \llem G$ we find a $b \in G$ with $c \ll b$. Then, as $k_K(G) \leqem k_K(I)$ we find an $i'$ in the closed set $\ov{\myconv I}$ such that $b \leq i'$. So, as $c \ll i' \in \ov{\myconv I}$, there is an $i \in \myconv I$ such that $c \ll i$, and so $c \leq i$.
As convex combinations are monotone, it is then enough to show that every $i'' \in I$ is below an element of $\myconv F$. This follows as, since $I \llem k_K(F)$, every such $i''$ is way-below an element of the closed set $\ov{\myconv F}$.

We next show:
\[\forall i \in I.\,\exists c \in \myconv H.\, c \leq i \tag*{$(**)$}\]
Choosing $i \in I$, as $k_K(G) \leqem k_K(I)$ there is a $b \in \myconv G$ with $b \leq i$. As $H \llem G$ and convex combinations are monotone, we then find the required $c \in \myconv H$.

\makeblue{{We now build a finite set $H_1$ of elements of $\myconv H$  
by picking one
below each element of $I$, 
 as}} guaranteed by $(**)$. 
 We then have:
\[\forall c \in H \cup H_1.\, \exists a \in \myconv F.\, \makeblue{c} \leq a \tag*{$({**}*)$}\]
\makeblue{Indeed, for $c\in H$, the conclusion is given by $(*)$, and, for $c\in H_1$, we use that $c\leq i$ for some $i\in I$ and that every element of $I$ is below an element of $\myconv F$, as in the argument proving $(*)$.}  
\makeblue{
{We next  build a finite set $F_1$ of elements of $\myconv F$  
by picking one above each element of 
$H\cup H_1$, as}} guaranteed by $({**}*)$.

We claim that $H \cup H_1 \leqem F \cup F_1$, as required. This follows as, on the one hand, by $({**}*)$, every element of $H \cup H_1$ is below an element of $F_1$,
and, on the other hand, as $I \llem k_K(F)$, every element of $\myconv F$ (and so of $F \cup F_1$) is above an element of $I$ and so, by 
\makeblue{ 
  {the choice of $H_1$}}, above an element of $H_1$.
\end{proof}

\makered{\begin{lem} \label{convex_way_below_convex} Let $K$ be a continuous
  coherent full \KS\ 
in which convex combinations preserve the way-below relation,
Suppose that $F$, $G$ are finite non-empty subsets of $K$ such that $G
\llem  F$. Then $k_K(G)\ll_{\cP K}  k_K (F)$. 
\end{lem}
%

%
\makeblue{\begin{proof}
By the characterisation of $\ll_{\cP K}$ given in
Theorem~\ref{PKeg} it suffices to find a finite set $H$ such that
$k_K(G)\leqem k_K(H)$, $H\subseteq \dda k_K(F)$ and $\dua H\supseteq k_K(F)$.

As convex combinations preserve $\ll_K$ and $ G  \llem  F$, we have $
\myconv G  \llem  \myconv F$. Then, using the compactness of $ \myconv F$,
we see that there is a non-empty finite subset $G'$ of $ \myconv G$ such
that  $\dua G' \supseteq  \myconv F$. Let $H=G\cup G'$. Since
$G\subseteq H\subseteq \myconv G$, we have $k_K(G)=k_K(H)$. 
We also have $H \subseteq \dda \myconv F \subseteq
\dda k_K(F)$, whence $H \subseteq\dda k_K(F)$. Finally, $\dua H \supseteq \dua
G' \supseteq \myconv F$, and so we have $\dua H \supseteq \ua\myconv F \supseteq k_K(F)$. 
\end{proof}}}

\begin{thm} \label{Puni} 
\makered{Let $K$ be a continuous coherent full \KS\
and in which convex combinations preserve the way-below relation.
Then the} map $\eta_P$ is universal. That is,
  for every  \KS\ semilattice  $\makered{L}$ and \KS\ morphism $f:K \rightarrow
  \makered{L}$ there is a unique  \KS\ semilattice  morphism $f^{\dagger}: \cP K
\rightarrow \makered{L}$ such that the following diagram commutes: 
{\[\begin{diagram}
	K\\
	 \dTo^{\eta_P} & \SE_{}{\quad \; \;\;f}\\
	\cP K & \rTo^{f^{\dagger}} & \makered{L}
	\end{diagram}\]
}
The  morphism is given by:
\[f^{\dagger}(X) = \sideset{}{^{\uparrow}}\bigvee \{
\sideset{}{_{\makered{L}}}\bignonor f(F)\mid F \subseteq_{\mathrm{fin}} K, F \neq \emptyset, F\llem
X\}  \] 
\end{thm}

\begin{proof}
For any non-empty finite set $F \subseteq K$ we have
\[k_K(F) = \sideP\bignonor \eta_P(F)\]
as  $\da \sideP\bignonor \eta_P(F) = \bigvee_{\cH K} \{ \da
\eta_P(b)\mid b \in F \} = \bigvee_{\cH K} \{ \da b \mid b \in F \} =
\cchull F = \da k_K(F)$, and (proved similarly)  $\ua \sideP\bignonor \eta_P(F) =  \ua k_K(F)$. 

Using this \makered{and the basis given in 
Theorem~\ref{PKeg}},  for any $X  \in \cP K$
we then have: 
\[X =  \sideset{}{^{\uparrow}}\bigvee \{ \sideP\bignonor \eta_P(F) \mid F \subseteq_{\mathrm{fin}} K, F \neq \emptyset, \makered{k_K(F) \ll_{\cP K} X}\}\]
It follows that
\[f^{\dagger}(X) = \sideset{}{^{\uparrow}}\bigvee \{ \sideset{}{_{\makered{L}}}\bignonor f(F) \mid F \subseteq_{\mathrm{fin}} K, F \neq \emptyset, 
\makered{k_K(F) \ll_{\cP K} X}\}  \] 
establishing uniqueness.

For existence we define $f^{\dagger}$ by means of this formula and
then verify that it makes the diagram commute and is both a \KS\ and a
semilattice map. 
  It is clearly continuous.
%
%
Next, as in the proof of Theorem~\ref{Suni} , it helpful to prove
that, for any non-empty $F \subseteq_{\mathrm{fin}} K$, we have: 
\[f^{\dagger}(\sideP\bignonor \eta_P(F)) =  \sideset{}{_{\makered{L}}}\bignonor f(F)
\tag*{$(*)$}\]
that is, that:
\[\sideset{}{^{\uparrow}}\bigvee \{ \sideset{}{_{\makered{L}}}\bignonor f(G) \mid  G
\subseteq_{\mathrm{fin}} K, G \neq \emptyset,  \makered{k_K(G) \ll_{\cP K}  k_K(F)} \}  =   
\sideset{}{_{\makered{L}}}\bignonor f(F)\]

\makered{To show that the left-hand side is $\leq \bignonor_L f(F)$, suppose given $G = \{b_1,\ldots, b_n\} \subseteq K$, with $n > 0$, such that $k_K(G) \ll_{\cP K}  k(F)$. Choose $c_1 \ll b_1,\ldots, c_n \ll b_n$ and set $H = \{c_1,\ldots, c_n\}$. } 
\makeblue{By Lemma~\ref{leqemc}, there are finite sets $H_1 \subseteq \myconv H$ and $F_1 \subseteq \myconv F$ such that $H \cup H_1 \leqem F \cup F_1$. We then have:
\[\bigcup_L f(H) = \bigcup_L f(H) \cup_L \bigcup_Lf(H_1)  \leq \bigcup_L f(F) \cup_L \bigcup_Lf(F_1) = \bigcup_Lf(F)\]
where the two equalities follow using the fact that $L$ satisfies the convexity identity (CI) several times. So $\bigcup_L f(G) \leq \bigcup_L(F)$ as $H$ consists of an
arbitrary choice of  elements way-below each element of $G$.}

\makered{Conversely, supposing $F = \{a_1,\ldots, a_n\}$, with $n > 0$, choose $b_1 \ll
a_1,\ldots, b_n \ll a_n$ and take $G = \{b_1,\ldots, b_n\}$. 
By Lemma~\ref{convex_way_below_convex} we have $k_K(G) \ll_{\cP K}  k_K(F)$.
So the left-hand side is $\geq
\bigcup_L \makered{f}(G)$, and so $\geq \bigcup_L \makered{f}(F)$, as $G$ consists of an
arbitrary choice of  elements way-below each element of $F$. }


Given $(*)$, the rest of the proof follows exactly as did that of
Theorem~\ref{Suni}. 
\end{proof}

\makered{Similarly to  the case of upper semilattices, this universality result contrasts with the corresponding universality result for convex powercones in~\cite{TKP09}. As before, it would be interesting to know if the preservation assumption made here is needed.}

\subsection{Powerdomains combining probabilistic choice and nondeterminism} \label{dPDomains}

Powerdomains combining probabilistic choice and nondeterminism exist
on arbitrary dcpos for general reasons. That is, there is
always a free  \KS\ semilattice  over any dcpo, and the same is
true for  \KS\   join- and  meet-semilattices. This is because
each of these kinds of structure can be axiomatised by inequations over a
signature of finitary operations, possibly  (Scott-)continuously parameterised by an
auxiliary dcpo, and free algebras over dcpos satisfying such
inequations always exist (this can be shown using the \makered{General Adjoint Functor
Theorem}, and see~\cite{HP06}).  \makered{These various free semilattices over a dcpo are automatically continuous if the dcpo is, as follows from~\cite{St95}
(but not from the less general results on free algebras in~\cite{AJ94}, which do not apply when there is  parameterisation).}

Free  \KS\ semilattices   are given by the
inequational theory with: a binary operation symbol  $+_{r}$, for each $r \in [0,1]$; a unary operation symbol $\cdot_r$, continuously parameterised by $r$, ranging over the dcpo $[0,1]$; 
a binary operation
symbol $\nonor$;  and a constant $0$. The equations \makered{consist of: equations for a \KS,  by which we mean the barycentric algebra equations for $+_{r}$, as given in Section~\ref{Keg} and the equation $\cdot_r(x) = x +_r 0$};
equations  asserting that $\nonor$ is
associative, commutative, and idempotent; and the equation 
\[x +_r (y \nonor z) = (x +_r y) \nonor (x +_r z)\]
saying that  $+_r$ distributes over $\nonor$ in its second argument,
for any $r \in [0,1]$ (and so also in its first one). For \KS\  join-semilattices   one adds the inequation $x \leq x \nonor y$; for
meet-semilattices one instead adds the inequation  $x \nonor y \leq x$. \\

While we do not know any general characterisation of these various
free constructions, by making use of our previous results we can
characterise them for domains (assumed also coherent in the convex
case). From the discussion in Section~\ref{dme} we know that the
subprobabilistic power domain $\cV_{\leq 1} P$ 
over a dcpo is a full \KS; 
 that, in case $P$ is a domain,
it is a continuous \KS\ with convex combinations preserving the
way-below relation; and that,  in case $P$ is also coherent, then so
is $\cV_{\leq 1} P$.  We further know that, if $P$ is a domain, then
the subprobabilistic powerdomain $\cV_{\leq 1} P$ is the free \KS\
over $P$, with unit $x \mapsto \delta_x$, where 
$\delta_x$ is the Dirac distribution, with mass 1 at $x$ (given a
Scott-continuous $f: P \rightarrow K$, we write $\overline{f}: \cV_{\leq 1} P
\rightarrow K$ for its extension to a \KS\ map).  

Therefore  we can form the three power \KSs\ $\cH \cV_{\leq 1} P$, $\cS \cV_{\leq 1} P$,
and $\cP \cV_{\leq 1} P$, assuming that  $P$ is a domain (and a coherent one,
in the convex case); it is immediate from the above remarks on the
subprobabilistic powerdomain  and Theorems~\ref{Huni},~\ref{Suni},
and~\ref{Puni} that these yield, respectively,  the free  \KS\  join-semilattice, the free  \KS\ meet-semilattice, and the free
 \KS\ semilattice  over a given domain.  We record these results as
corollaries. 

\begin{cor} \label{HVuni} Let $P$ be a domain. Then the map
  $\eta_{HV} \eqdef \da
  \delta_x: P \rightarrow \cH\cV_{\leq 1} P$ 
  is universal. That is, for every  \KS\  join-semilattice 
  $\makered{L}$ and Scott-continuous map $f:P \rightarrow \makered{L}$  
there is a unique  \KS\ semilattice  morphism $f^{\dagger}: \cH \cV_{\leq 1} P
\rightarrow \makered{L}$ such that the following diagram commutes: 
{\[\begin{diagram}
	P\\
	 \dTo^{\eta_{HV}} & \SE_{}{\quad \; \;\;f}\\
	\cH\cV_{\leq 1} P & \rTo^{f^{\dagger}} & \makered{L}
	\end{diagram}\]
}
The morphism is given by:
  \[f^{\dagger}(X) =  \bigvee \overline{f}(X)\tag*{\qEd}\]
\end{cor}

\begin{cor} \label{SVuni}  Let $P$ be a domain. Then the map
  $\eta_{SV}  \eqdef \ua \delta_x: P \rightarrow \cS\cV_{\leq 1} P$
  is universal. That is, for every \KS\ meet-semilattice 
  $\makered{L}$ and Scott-continuous map $f:P \rightarrow \makered{L}$  
there is a unique  \KS\ semilattice  morphism $f^{\dagger}: \cS \cV_{\leq 1} P
\rightarrow \makered{L}$ such that the following diagram commutes: 
{\[\begin{diagram}
	P\\
	 \dTo^{\eta_{SV}} & \SE_{}{\quad \; \;\;f}\\
	\cS\cV_{\leq 1} P & \rTo^{f^{\dagger}} & \makered{L}
	\end{diagram}\]
}
The morphism is given by:

\[f^{\dagger}(X) = \sideset{}{^{\uparrow}}\bigvee \{ \bigwedge \overline{f}(F) \mid F \subseteq_{\mathrm{fin}} \cV_{\leq 1} P, F \neq \emptyset,  \dua F \supseteq X\} \tag*{\qEd} \]
\end{cor}

\begin{cor} \label{PVuni} Let $P$ be a coherent domain. Then the
  map $\eta_{PV} \eqdef \{\delta_x\}: P \rightarrow \cP\cV_{\leq 1} P$
  is universal.  That is, for every   \KS\ semilattice 
  $\makered{L}$ and Scott-continuous map $f:P \rightarrow \makered{L}$  
there is a unique  \KS\ semilattice  morphism $f^{\dagger}: \cP \cV_{\leq 1} P \rightarrow \makered{L}$ such that the following diagram commutes: 
{\[\begin{diagram}
	P\\
	 \dTo^{\eta_{PV}} & \SE_{}{\quad \; \;\;f}\\
	\cP\cV_{\leq 1} P & \rTo^{f^{\dagger}} & \makered{L}
	\end{diagram}\]
}
The morphism is given by:

\[f^{\dagger}(X) = \sideset{}{^{\uparrow}}\bigvee \{\sideset{}{_{\makered{L}}}\bignonor \overline{f}(F) \mid F \subseteq_{\mathrm{fin}} \cV_{\leq 1} P, F \neq \emptyset,  \dua F \supseteq X\} \tag*{\qEd} \]
\end{cor}

\section{Functional representations} \label{funcrep}

In \cite[Sections 4 and 6]{KP}, the various powercones
over a d-cone were represented by 
functionals. We will use those results to obtain similar functional
representations of the corresponding  power \KSs, and then deduce
corresponding functional representations for mixed powerdomains. 

\makeblue{Some context may help. For a functional representation of a monad $T$ one chooses a test space $O$, say, and represents an object $T(X)$ by a suitable collection of functionals, with domain a space of `test functions' from $X$ to $O$ and range $O$. One general such method is to work in a symmetric monoidal closed category, when one has available the `continuation' or `double-dualisation' monad $[[X,O],O]$ (writing $[X,Y]$ for the function space). Assuming that the monad $T$ is strong, there is then a 1-1 correspondence between $T$-algebras $\alpha:T(O) \rightarrow O$ and morphisms  $T \rightarrow [[-,O],O]$ of strong monads~\cite{KP93,Koc70,Koc12}. If there are sufficiently many test functions, this morphism will be a monomorphism, and, perhaps with further restrictions on the functionals, it may corestrict to an isomorphism; one may also have to restrict to certain objects $X$.

In our case, we would work with the category of \KSs\ and continuous linear maps, when $[K,L]$ would be the \KS\ formed from such maps with the pointwise order and algebraic structure, and the extended reals $\oRp$ provide a natural test space. 
As we will see below, there is a natural choice of functionals for all three of our power-\KSs, but in no case are such functionals generally  linear: for example in the Hoare case they are rather sublinear. When, later,  we apply our results to obtain functional representations of mixed powerdomain monads, we are working in the cartesian-closed category of dcpos, the above general framework does apply and our functional representations are, in fact, submonads of the relevant continuation monads (but, for non-essential reasons, with some minor differences in the convex case).}

Throughout this section, we generally work with full \KSs, that is,
those satisfying Property
(OC3). We consider such \KSs\ $K$ to be embedded in their 
 universal d-cones $C=\dCone(K)$ 
as Scott-closed convex sets (Theorem~\ref{prop:KSembed}) and we recall
that the universal d-cones are continuous 
whenever the \KSs\ are (Proposition \ref{prop:continuous}).
From Section~\ref{dme}  we recall
that the  subprobabilistic  powerdomain  $\cV_{\leq 1}
P$ of a dcpo $P$ is a full \KS, 
 that $\dCone({\cV_{\leq 1} P}) \cong \cV P$, the valuation
powerdomain of $P$, and that, in case  $P$ is a
domain,  $\cV_{\leq 1} P$ is continuous.

We will make use of \emph{norms} on d-cones $C$, taking them to be
Scott-continuous sublinear functionals $\norm{\hspace{0.5pt}\mbox{-}\hspace{0.5pt}}\colon C \to
\oRp$ such that $\norm{x}>0$ for every
  $x\neq 0$;
\emph{normed d-cones} are then d-cones equipped with a norm. A map
$f\colon C \to D$ from one normed d-cone $C$ to another $D$ is
\emph{nonexpansive} if   $\norm{f(x)} \leq \norm{x}$ for all $x
\in C$; it is a morphism of normed d-cones if it is nonexpansive and a
morphism of d-cones (i.e., Scott-continuous and linear).

Various function space d-cones will be involved in our development. As
well as those considered in Section~\ref{dme}  we note that, for any
d-cones $C$ and $D$, the subsets $\SubL(C,D)$, and $\SuperL(C,D)$ of
$D^C$ of, respectively, the sublinear, and superlinear functions form
sub-d-cones of $D^C$. 
Regarding d-cone semilattices, if $D$ is a d-cone semilattice
(respectively, join-semilattice, meet-semilattce) then, with the
pointwise structure, so is $D^P$, for any dcpo $P$.  Further, for any
d-cone $C$, if $D$ is a  join-semilattice (meet-semilattice) then, as
is easily checked,  $\SubL(C,D)$ (respectively, $\SuperL(C,D)$) is a
sub-d-cone join-semilattice (respectively,  sub-d-cone
meet-semilattice) of $D^C$. 

Supposing additionally the cones $C$ and $D$ to be normed, the
collection $\SubL^{\leq1}(C,D)$ of all Scott-continuous sublinear 
nonexpansive functions, with
$D$ a d-cone join-semilattice, forms  
a  \KS\ join-semilattice; indeed it is a sub-\KS\ join-semilattice of
$\SubL(C,D)$, regarding the latter as a \KS\ join-semilattice. 

A trivial example of a normed cone is  $\oRp$ \makered{with norm  the identity function:
\[\norm{x} = x\]
A less trivial example is provided by} the dual cone $K^*$ of a \KS\ $K$
equipped with the  sup norm, defined by\display  
\[\norm{f}^*_{K} = \sup_{x\in K}f(x)\]
%
where the index $K$ indicates the dependency of this norm on
the d-cone $K^*$ on  the \KS\ $K$. Notice that $\norm{f}^*_{K}=+\infty$ if there is an $x\in
K$ such that $f(x)=+\infty$.  

If $K$ satisfies Property (OC3) then we can define a norm on $C^*$, where $C
\eqdef \dCone(K)$ by\display 
\[\norm{f}^*_{K} \eqdef \norm{f \!\upharpoonright \! K}^*_{K}  = \sup_{x \in K} f(x)\]
where the index now indicates  the dependency of this norm on $C^*$  on $K$.
%
%
With this norm, the d-cone isomorphism between $K^*$ and $C^*$
given in Section~\ref{dme}, Example \ref{ex:KSfunctionspaces} becomes
an isomorphism of normed d-cones. 

\makeblue{Recall that, for any element $x\in C = \dCone(K)$, the evaluation 
map $\ev_C(x)\colon C^*\to\oRp$ sends $f$ to $f(x)$.
We note that}
$\ev_C(x) \leq \norm{\hspace{0.5pt}\mbox{-}\hspace{0.5pt}}^*_{K}$ if
$x\in K$, with the converse 
holding if $K$ is continuous. For,  if $x\in K$ then we have  $\ev_C(x)\leq
\norm{\hspace{0.5pt}\mbox{-}\hspace{0.5pt}}^*_{K}$, since  
$f(x)\leq  \sup_{x\in K}f(x) = \norm{f}^*_{K}$, for $f \in C^*$. And
if $x\not\in K$, then using the Strict Separation
Theorem~\cite[Theorem 3.8]{TKP09}, 
we obtain an $f\in C^*$ such that
$f(y)\leq 1$ for $y\in K$ but $f(x)>1$, whence $\norm{f}^*_{K}=\sup_{y\in
  K}f(y)\leq 1<f(x) =\ev_C(x)(f)$, and so $\ev_C(x)\not\leq
\norm{\hspace{0.5pt}\mbox{-}\hspace{0.5pt}}^*_{K}$. Thus, if the d-cone $C = \dCone(K)$ is continuous and
reflexive, the Scott-continuous linear functionals $\varphi\leq
\norm{\hspace{0.5pt}\mbox{-}\hspace{0.5pt}}^*_{K}$  on $C^*$
are given by evaluations at points $x\in K$.

\subsection{The lower power \KS}\label{lower}

We regard $\oRp$ as a d-cone join-semilattice  with the 
semilattice operation $r\vee s =\max(r,s)$ 
(as such it is isomorphic to $\cH \oRp$, now regarding $\oRp$ as a
d-cone). We take an arbitrary
d-cone $C$ with its lower powercone $\cH C$ and its dual
$C^*$. 
Consider the map $\Lambda_C\colon \cH C \to \oRp^{C^*}$
where\display 
\[ \Lambda_C(X)(f) \eqdef \sup_{x\in X} f(x)\] 
Fixing $f\in C^*$,  we obtain the map $\Lambda_C(-)(f)\colon \cH
C\to\oRp$, which is  
the unique d-cone join-semilattice morphism extending
$f$ along the canonical embedding $\eta\colon C\to\cH C$
by \cite[Proposition 3.2]{KP}.

Fixing $X\in \cH C$, we obtain the functional
\[\Lambda_C(X)\colon C^*\to \oRp\]
where\display
\[ \Lambda_C(X)(f) = \sup_{x \in X} f(x)\]
As the pointwise supremum of the Scott-continuous
linear functionals 
$\ev_C(x)$ ($x \in X$), $\Lambda_C(X)$ is Scott-continuous and sublinear.
In this way we obtain a  d-cone join-semilattice morphism
\[ \Lambda_C \type \cH C \longrightarrow \SubL(C,\oRp)\]
which represents the lower convex powercone by the Scott-continuous
sublinear functionals on the dual cone $C^*$.
\begin{thm}[{\cite[Proposition 6.1 and Theorem 6.2]{KP}}]\label{th:lower}  
Let $C$ be a  d-cone. Then we have a d-cone join-semilattice morphism
$\Lambda_C\colon \cH C\to \SubL(C,\oRp)$, where\display
\[\Lambda_C(X) = \sup_{x \in X} f(x)  \]
If, in addition, $C$ is continuous then $\Lambda_C$ is an order embedding; 
if, further, $C$ is reflexive
with a continuous dual then it is an isomorphism.
\qed \end{thm} 

To  apply the
above considerations to the universal d-cone $C=\dCone(K)$ over $K$ we
 now consider a full \KS\  $K$. 
The power \KS\ $\cH K$ is a Scott-closed convex join-subsemilattice of
the powercone $\cH C$. The functionals
$\Lambda_C(X)$ representing Scott-closed 
convex subsets $X$ of $K$ are  the sublinear functionals
$\Lambda_C(X)$ dominated by the 
norm $\norm{\hspace{0.5pt}\mbox{-}\hspace{0.5pt}}^*_{K}=\Lambda_C(K)$, and all of them if $K$ is continuous.
For certainly if $X \subseteq K$ then $\Lambda_C(X) \leq \Lambda_C(K)
= \norm{\hspace{0.5pt}\mbox{-}\hspace{0.5pt}}^*_{K}$, and, assuming the converse, for any $x \in
X$ we have $\ev_C(x) \leq \Lambda_C(X) \leq \norm{\hspace{0.5pt}\mbox{-}\hspace{0.5pt}}^*_{K}$,
and so $x \in K$ by the above discussion (assuming $K$ continuous). 

Recalling that $\SubL^{\leq1}(K^*,\oRp)$ is the collection of
nonexpansive functionals in $\SubL(K^*,\oRp)$, we therefore have a
\KS\ join-semilattice morphism 
 \[\Lambda_K\colon \cH K \to  \SubL^{\leq1}(K^*,\oRp)\] 
 viz.\  the composition
 \[\cH K \xrightarrow{\Lambda_C\upharpoonright \cH K} \SubL^{\leq1}(C^*,\oRp) \cong \SubL^{\leq1}(K^*,\oRp)\]
of the restriction of $\Lambda_C$ to $\cH K$ with  the isomorphism
$\SubL^{\leq1}(C^*,\oRp) \cong \SubL^{\leq1}(K^*,\oRp)$ arising from
the normed d-cone isomorphism between $K^*$ and $C^*$. 
Theorem \ref{th:lower} then yields the desired functional representation theorem, adapting its hypotheses to \KSs
\display  
\begin{thm}\label{th:lowerKS}
Let $K$ be a full \KS. 
Then  we have a   \KS\  join-semilattice morphism  $\Lambda_K\colon \cH K \to \SubL^{\leq1}(K^*,\oRp)$. It is given by\display
\[ \Lambda_K(X)(f) \eqdef \sup_{x \in X} f(x)\]
 If $K$ is continuous then $\Lambda_K$ is an order embedding. If, further, the dual cone $K^*$ is 
continuous and the universal d-cone $\dCone(K)$ is reflexive, then
$\Lambda_K$ is an isomorphism. 
\qed \end{thm}

With the aid of this theorem we can obtain a corresponding result for the lower mixed powerdomain. For any dcpo $P$, making use of Section~\ref{dme}, we see that the predicate extension and restriction maps
\[\mathrm{EXT}_{P} \eqdef f \mapsto \ov{f}\colon \cL P \to (\cV_{\leq
  1} P)^* \quad \mbox{and}  
\quad \mathrm{RES}_P \eqdef f \mapsto f \circ \delta \colon
(\cV_{\leq 1} P)^* \to \cL P\]
 are d-cone  morphisms, and mutually inverse isomorphisms if $P$ is a domain. 
 
 Next, for any dcpo $P$, we equip the d-cone $\cL P$  with the sup
 norm, i.e., the one  defined by: $\norm{f}_{\infty} = \sup_{x \in P}f(x)$. 
\begin{lem}\label{isonorm} Let $P$ be a dcpo. Then the extension map
  $\EXT_P\colon \cL P \to (\cV_{\leq 1} P)^*$ preserves the norm. The
  restriction map $\RES_P\colon (\cV_{\leq 1} P)^* \to \cL P$ is
  nonexpansive and preserves the norm if $P$ is a domain. 
So $\EXT_P$ and $\RES_P$ are normed d-cone morphisms, and  mutually inverse  isomorphisms if $P$ is a domain. 
\end{lem}
\begin{proof} We wish first to show that $\norm{f}_{\infty} =
  \norm{\ov{f}}^*_{(\cV_{\leq 1}P )}$ for a given $f\in \cL P$ (where
  $\ov{f}(\mu) = \int\!f\, d\mu$). In one direction, we have
  \makered{$\norm{f}_\infty \leq \norm{\ov{f}}^*_{(\cV_{\leq 1}P )}$} as,
  for any $x \in P$, $f(x) = \int\!f\,d\delta_x = \ov{f}(\delta_x)$.  
In the other direction it suffices to show that $\ov{f}(\mu) \leq
\norm{f}_{\infty}$ for all $\mu \in \cV_{\leq 1}P$. This holds as we have: 
$\ov{f}(\mu)  = \int\!f\, d\mu \leq \int\! ( x \mapsto
\norm{f}_{\infty}) \,d\mu =  \mu(P)\norm{f}_{\infty} \leq \norm{f}_{\infty}$. 
Next, the restriction map is evidently nonexpansive. If $P$ is continuous it preserves the norm as then it is right inverse to the extension map, and that preserves the norm.
\end{proof}%

We will make use of the mapping\display
\[\Phi_{P}\colon  \oRp^{(\cV_{\leq 1} P)^*} \to  \oRp^{\cL P}\]
where $P$ is a dcpo and $\Phi_{P}(F) = F\circ \EXT_P$. It is a d-cone morphism, and preserves pointwise joins and meets. 
If $P$ is a domain it  is an isomorphism, with inverse $\Phi^{r}_P \eqdef F \mapsto  F\circ \RES_P$.  %
\begin{cor}\label{th:lowerdomain}
Let $P$ be a dcpo. 
Then  we have a  \KS\  join-semilattice  morphism\display
\[\Lambda_P\colon \cH\mathcal{V}_{\leq 1} P \;\longrightarrow\;
\SubL^{\leq1}(\cL P,\oRp)\]
It is given by\display
\[ \Lambda_P(X)(f) \eqdef \sup_{\mu \in X} \int \!f\,d\mu\]
If $P$ is a domain then $\Lambda_P$ is an isomorphism.
\end{cor}   
\begin{proof} 
We first check that both $\Phi_{P}$ and $\Phi^r_{P}$ preserve sublinearity and nonexpansiveness, the latter by Lemma~\ref{isonorm}. So $\Phi_{P}$  cuts down to a morphism 
\[\SubL^{\leq1}((\cV_{\leq 1} P)^*,\oRp) \to \SubL^{\leq1}(\cL P,\oRp)\]
 of \KS\  join-semilattices that is an isomorphism if $P$ is a domain.

Next, as discussed in Section~\ref{dme}, $\cV_{\leq 1} P$ is a full    
\KS, 
 and, if $P$ is a domain, then $\cV_{\leq 1} P$ is continuous and  $\dCone({\cV_{\leq 1} P}) \cong \cV P$; further, 
if $P$ is continuous then 
$(\cV P)^*$ (which is isomorphic to $ (\cV_{\leq 1} P)^*)$ is continuous (being isomorphic
to $\cL P$), and $\cV P$ is reflexive. So if $P$ is a domain then $\cV_{\leq 1} P$ satisfies all the other  various hypotheses of Theorem~\ref{th:lowerKS}.

An easy calculation then displays $\Lambda_P$ as the following composition of \KS\ join-semilattice morphisms that are isomorphisms if $P$ is a domain\display
 \[\cH\mathcal{V}_{\leq 1} P \xrightarrow{\Lambda_K}
 \SubL^{\leq1}((\mathcal{V}_{\leq 1} P)^*,\oRp) \xrightarrow{\Phi_{P}}
 \SubL^{\leq1}(\cL P,\oRp) \tag*{\qEd} \]
\def\popQED{}
\end{proof}

We remark that 
nonexpansiveness has a simple
formulation  for monotone homogeneous functionals $F\colon \cL P \to
\oRp$, viz.\  that 
\[F(\mathbf 1_P)\leq 1\]
where $\mathbf 1_P$ is the
constant function on $P$ with value $1$. The condition is evidently is
a special case of nonexpansiveness, as $\norm{\mathbf 1_P}_{\infty} =
1$. Conversely, for any $g \in \cL P$, noting that  $g \leq
\norm{g}_{\infty} \mathbf{1}_P$, we have: $\norm{F(g)} \leq
\norm{F(\norm{g}_{\infty} \mathbf{1}_P)} = \norm{g}_{\infty}
\norm{F(\mathbf 1_P)} \leq \norm{g}_{\infty}$.

\subsection{The upper power \KS} \label{upper}

We regard $\oRp$ as a d-cone meet-semilattice  with the 
semilattice operation $r\wedge s =\min(r,s)$ 
(as such it is isomorphic to $\cS \oRp$, now regarding $\oRp$ as a
d-cone). Take a continuous
d-cone $C$ with its upper powercone $\cS C$ and its dual
$C^*$. 
Consider the map $\Lambda_C\colon \cS C \to  \oRp^{C^*}$
where\display 
\[ \Lambda_C(X)(f) \eqdef \inf_{x\in X} f(x)\] 

Fixing $f\in C^*$,  we obtain the map $\Lambda_C(-)(f)\colon \cS
C\to\oRp$, which  is the unique Scott-continuous linear 
meet-semilattice homomorphism extending 
$f$ along the canonical embedding $\eta\colon C\to\cS C$ which maps
$x$ to $\ua x$ (this follows from \cite[Proposition 3.5]{KP}, using
the above isomorphism). %

Fixing $X\in \cS C$, we obtain the functional
\[\Lambda_C(X)\colon C^*\to \oRp\]
where\display
\[ \Lambda_C(X)(f) = \inf_{x \in X} f(x)\]
As the pointwise infimum of linear functionals, $\Lambda_C(X)$ is
superlinear. It is also Scott-continuous. \makeblue{Indeed, for a Scott-compact
set $X$, the image $f(X)$ is Scott-compact in $\oRp$, hence has a
smallest element $\min f(X) = \inf_{x\in X}f(x)=\Lambda_C(X)(f)$; thus  
$\ua \Lambda_C(X)(f) = \cS (f)(X)$; since $f\mapsto \cS(f)$ is
Scott-continuous, $f\mapsto  \cS(f)(X)\colon C^*\to\cS(\oRp)$ is
Scott-continuous, too;  composing with the isomorphism
$\cS(\oRp)\iso\oRp$ yields the Scott continuity of $\Lambda_C(X)$.}  

In this way we obtain a d-cone meet-semilattice morphism
\[ \Lambda_C \type \cS C \longrightarrow \SuperL(C^*,\oRp)\]
representing the upper powercone by the Scott-continuous superlinear functionals
on the dual cone $C^*$.
We need the quite strong hypothesis of a convenient d-cone (see
Section \ref{dme}) to obtain the
analogue of Theorem \ref{th:lower}\display

\begin{thm}[{\cite[Proposition 6.4 and Theorem 6.5]{KP}}]\label{th:upper}
Suppose that $C$ is a continuous d-cone. Then we have a d-cone meet-semilattice  morphism $\Lambda_C \type \cS C \to \SuperL(C^*,\oRp)$, which is an order embedding, where\display
\[ \Lambda_C(X)(f) = \inf_{x \in X} f(x) \]
Further, if $C$ is convenient then $\Lambda_C$ is an isomorphism.
\qed \end{thm}

We now consider a continuous full \KS\ $K$. 
 The universal d-cone
$C=\dCone(K)$ is also then continuous and we apply the above
considerations to it.  
The upper power \KS\ $\cS K$ 
consists of all nonempty Scott-compact saturated convex subsets $X$ of
$K$.  As discussed in Section~\ref{SpK} the map $u\colon \cS K \to \cS C$, where $u(X) = \ua{X}$ is a d-cone meet-semilattice morphism which is an order-embedding. 
The functions $\Lambda_C(u(X))\colon C^*\to\oRp$ ($X\ \in \cS K$) are
Scott-continuous and 
superlinear. We want to characterise the Scott-continuous and
superlinear functionals $F$ on $C^*$ that represent the elements of $\cS
K$ in this way. It turns out that, unlike the case of the lower power
\KS,  being nonexpansive is not
sufficient. We notice that, for any
$X\in \cS K$, the representing functional $F\colon C^* \to \oRp$ has  a remarkable property: it is \emph{strongly
  nonexpansive}, by which we mean that
%
\[F(f+g)\leq F(f) + \norm{g}^*_{K} \]
holds for all $f,g \in C^*$ (setting $f = 0$, we see that strong nonexpansiveness implies nonexpansiveness).
%
Indeed, we have: $F(f+g) = \inf_{x\in X}(f+g)(x) = \inf_{x\in X}
(f(x) + g(x)) \leq  \inf_{x\in X}(f(x) + \sup_{x\in
  K}g(x)) = \inf_{x\in X} (f(x)+\norm{g}^*_{K}) = F(f) +
\norm{g}^*_{K}$.   

For any normed d-cone $D$ we write  $\SuperL^{\rm{sne}}(D,\oRp)$
for the collection of all Scott-continuous
superlinear functionals $F\colon D \to \oRp$ that are \emph{strongly
nonexpansive} in the sense that\display
\[F(x + y)\ \leq\ F(x) + \norm{y}\]
holds for all $x,y \in D$. The collection forms a  \KS\ meet-semilattice, indeed it is a sub-\KS\
 meet-semilattice of  
$\SuperL(D,\oRp)$. One easily checks that $\SuperL^{\rm{sne}}(C^*,\oRp)$
and $\SuperL^{\rm{sne}}(K^*,\oRp)$ are isomorphic as \KS\ meet-semilattices
via the normed d-cone isomorphism between $K^*$ and $C^*$. 

Putting all this together, we define
\[\Lambda_K\colon  \cS K \longrightarrow \SuperL^{\rm{sne}}(K^*,\oRp)\]
to be the \KS\  meet-semilattice morphism given by the composition\display
\[\cS K \xrightarrow{\Lambda_C \,\circ\, u} \SuperL^{\rm{sne}}(C^*,\oRp) \cong \SuperL^{\rm{sne}}(K^*,\oRp)\]
Theorem~\ref{th:upper} then yields the desired functional
representation theorem, adapting its hypotheses to \KSs\display   

\begin{thm}\label{th:upperKS}
Let $K$ be a continuous full \KS. 
 Then we have a
\KS\  meet-semilattice  morphism $\Lambda_K\colon \cS K\to
\SuperL^{\rm{sne}}(K^*,\oRp)$. It is given by\display 
\[ \Lambda_K(X)(f) \eqdef \inf_{x \in X} f(x)\]
If, further,  $\dCone(K)$ is convenient, then $\Lambda_K$ is an isomorphism.
\end{thm} 
\begin{proof}
 We note first that, for any $X \in \cS K$ and $f \in K^*$, we have\display
\[\Lambda_K(X)(f)  = \Lambda_C(\ua{X})(\widetilde f) = \inf_{x\in \ua X}\widetilde{f}(x) =\inf_{x\in X}\widetilde{f}(x)  = \inf_{x\in X}f(x) \]
Next, as $\dCone(K)$  is continuous since $K$ is, Theorem~\ref{th:upper} tells us that $\Lambda_C$ is an order-embedding; so, as $u$ is also one, so too is $\Lambda_K$. 

For the isomorphism, assuming that $\dCone(K)$ is convenient, we need
then only show that $\Lambda_C\circ u\colon \cS K\to
\SuperL^{\rm{sne}}(C^*,\oRp)$ is onto. 
So take any strongly nonexpansive Scott-continuous superlinear $F\colon C^*\to\oRp$.
By Theorem \ref{th:upper} we know
that there is a $Y\in \cS C$ such that 
$\Lambda_C(Y)  =  F$, that is such that $F(f)=\inf_{y\in Y}f(y)$ for all $f\in
C^*$. Let $X = Y\cap K$. Clearly, $X$ is a Scott-compact convex set
saturated in $K$. 
We want to show that $X$ is non-empty and 
 $\Lambda_C(u(X)) = \Lambda_C(Y)$, that is, 
$\inf_{y\in Y}f(y)=\inf_{x\in X}f(x)$ for all $f\in C^*$. 

Since the d-cone $K^* \cong C^*$ is assumed to be continuous, 
strong nonexpansiveness
allows us to apply the Main Lemma \cite[Lemma 5.1(1)]{KP} to $C^*$.
We learn that $\Lambda_C(Y)(f) =\inf \varphi(f)$,
where $\varphi$ ranges over the Scott-continuous linear functionals on
$C^*$ such that $\Lambda_C(Y) \leq \varphi\leq \norm{\hspace{0.5pt}\mbox{-}\hspace{0.5pt}}^*_{K}$. Using the
hypotheses of continuity and reflexivity for $C$, and the discussion at the beginning of this section, this can be rewritten in the 
form    $\Lambda_C(Y)(f)=\inf_{x\in Q}f(x)$,
 where $Q$ is the set of those elements $x\in K$ that satisfy
$\Lambda_C(Y)(f)\leq f(x)$ for all $f\in C^*$. 
Note that $Q$ is non-empty, as, taking $f$ to be constantly $0$, we have $\inf_{x\in Q}f(x) = \Lambda_C(Y)(f) = 0$ (recalling that $Y$ is non-empty).

As $Q$ is non-empty and   $\Lambda_C(Y)(f)=\inf_{x\in Q}f(x)$, it only remains, therefore, to show that $X=Q$. Clearly, $X\subseteq Q$.
For the reverse containment, suppose that $x \in Q \subseteq K$. \\
 We cannot have $x \not\in Y$ as otherwise, by the  Strict Separation
 Theorem~\cite[Theorem 3.8]{TKP09},
 \makered{there is an $r > 1$ such that $f(x) \leq 1$ and $f(y) > r$ for every $y \in Y$, 
  whence $\Lambda_C(Y)(f) =\inf_{y\in Y}f(y)\geq r \not\leq f(x)$. } So $x \in K \cap Y = X$ as required.
\end{proof}

We can now  specialise to domains\display
\begin{cor}\label{th:upperdomain}
Let $P$ be a domain.
Then  we have a  \KS\  meet-semilattice isomorphism
\[\Lambda_P\colon \cS\mathcal{V}_{\leq 1} P \;\cong\;
\SuperL^{\rm{sne}}(\cL P,\oRp)\]
It is given by\display
\[ \Lambda_P(X)(f) \eqdef \inf_{\mu \in X} \int \!f\,d\mu\]
\end{cor}   
\begin{proof} Using Lemma~\ref{isonorm}, we can check that both
  $\Phi_{P}$ and its inverse  preserve strong nonexpansiveness,
and so that $\Phi_{P}$  cuts down to an isomorphism 
\[\SuperL^{\rm{sne}}((\cV_{\leq 1} P)^*,\oRp) \cong \SuperL^{\rm{sne}}(\cL P,\oRp) \]
 of \KS\ meet-semilattices.
  Next, from the discussion of the  valuation  powerdomain  in Section~\ref{dme}, 
  we know that $\dCone(\cV_{\leq 1}P) \cong \cV P$ is convenient, and so we can apply Theorem~\ref{th:upperKS}. 
One then displays $\Lambda_P$ as the following composition of \KS\ meet-semilattice isomorphisms\display
 \[\cS\mathcal{V}_{\leq 1} P \xrightarrow{\Lambda_{(\cV_{\leq 1} P)}}
 \SuperL^{\rm{sne}}((\cV_{\leq 1} P)^*,\oRp) \xrightarrow{\Phi_{P}}
 \SuperL^{\rm{sne}}(\cL P,\oRp) \tag*{\qEd} \]
\def\popQED{}
\end{proof}

Strong nonexpansiveness
has a simple formulation for 
Scott-continuous homogeneous functionals
$F\colon \cL P \to \oRp$, viz.\  that
\[  F(f+\mathbf 1_P) \leq  F(f) + 1 \]
holds for all  $f\in \cL P$.\footnote{Goubault-Larrecq calls such functionals 
\emph{subnormalised previsions} in~\cite{GL07,GL08,GL10,GL12}.} 
Clearly a strongly nonexpansive functional satisfies this condition,
since $\norm{\mathbf 1}_{\infty} = 1$. Suppose conversely that the
second condition is satisfied and 
take any $f, g\in\cL P$. For $g=0$ there is nothing to prove. So let $g\neq
0$, and suppose that $g$ is bounded, i.e., that $\norm{g}_{\infty} < \infty$. 
Then, using homogeneity and then monotonicity and then the simplified
condition, we have\display 
\[F(f+g)=  \norm{g} F(\frac{1}{\norm{g}}f  +  \frac{1}{\norm{g}}g) \leq \norm{g}F(\frac{1}{\norm{g}}f+ \mathbf 1_P) \leq
\norm{g}(F(\frac{1}{\norm{g}}f)+1) = F(f)+\norm{g}\]
As every non-zero $g \in \cL P$ is the directed sup of bounded
non-zero such $g$'s, using the continuity of $F$ we then see that
$F$ is strongly nonexpansive.  

\subsection{The convex power \KS} \label{convex}

Here our representations employ 
functionals with values not in $\oRp$, but rather in $\cP \oRp $, the convex 
powercone of the extended nonnegative reals; this consists of the closed
intervals $a=[\un{a},\ov{a}]$, with $\un{a}\leq\ov{a}$ in $\oRp$, 
ordered by
the Egli-Milner order, where: $[\un{a},\ov{a}]\leqem [\un{b},\ov{b}]$ if
$\un{a}\leq \un{b}$ and $\ov{a}\leq \ov{b}$, with
addition given by 
$[\un{a},\ov{a}]+[\un{b},\ov{b}]=[\un{a}+\un{b},\ov{a}+\ov{b}]$, and 
scalar multiplication  given by $r[\un{a},\ov{a}]=[r\un{a},r\ov{a}]$. The semilattice operation
on $\cP \oRp$ is
$[\un{a},\ov{a}]\nonor[\un{b},\ov{b}]=[\min(\un{a},\un{b}),\max(\ov{a},\ov{b})]$;
 the semilattice order is containment
$\subseteq$. \makered{We define a norm on $\cP \oRp$ by setting $\norm{a} \eqdef \ov{a}$.}

We also use notation adapted from \cite[Section 4]{KP} setting up, for
any dcpo $P$, a bijection between functions $F \colon P \to \cP\oRp$
and pairs of functions $G,H\colon P \to \oRp$ with $G \leq H$. In
one direction, given such an $F$, we write $\un{F}(x)$ and 
$\ov{F}(x)$ for the lower and upper ends
of the image $F(x)$ of any $x\in P$, obtaining such a pair of functions $\un{F}$ and $\ov{F}$; conversely, given such a pair of functions $G$ and $H$, we set $[G,H](x)$ equal to the interval $[G(x),H(x)]$. 
The function $F$ is Scott-continuous if and only if both $\un{F}$
and $\ov{F}$ are. In case we are considering functionals $F\colon D \to \cP \oRp$ where $D$ is a d-cone, then $F$ is said to be
$\subseteq$-\emph{sublinear}, if $F$ is homogeneous and 
$F(x + y)\subseteq F(x)+F(y)$ for all $x,y \in D$ (which is equivalent to $\un{F}$ being
superlinear and  $\ov{F}$ sublinear), and it is said to be \emph{medial}\footnote{This property is called `canonicity'
in~\cite[Section 4.3]{KP}, and `Walley's condition' in~\cite[Definition 7.1]{GL10}.  Indeed, Walley includes 
this property among those 
  of his 
  `coherent previsions'
  in his
  book on reasoning with imprecise 
  probabilities~\cite[Section 2.6]{Wal}.}  if we have
%
%
\[\un{ F}(x + y)\leq \un{ F}(x)+\ov{ F}(y)\leq
\ov{ F}(x + y)\]
for all $x,y \in D$. 

Let us recall the (diagonal) functional representation of the convex powercone
$\cP C$ of a coherent continuous d-cone $C$ from \cite[Sections 4 and 6]{KP}.
It combines the representations of the lower and upper powercones.
Take a continuous coherent d-cone $C$ with its its dual
$C^*$ and its convex powercone $ \cP C$. 
Consider the map $\Lambda_C\colon \cP C \to  \cP \oRp^{C^*}$ where\display
\[ \Lambda_C(X)(f) \eqdef  [\inf_{x\in X} f(x), \sup_{x\in X} f(x)] \] 
Fixing $f\in C^*$,  we obtain the map $\Lambda_C(-)(f)\colon \cP C\to \cP \oRp$, which  is is the unique d-cone semilattice morphism $\cP f\colon \cP C\to\cP \oRp$ extending
$\eta_{\oRp}\!\circ\! f$ along the canonical embedding $\eta_C$, where, for any continuous coherent cone $D$, the canonical embedding  $\eta_D\colon D \to\cS D$ maps $x$ to $\{x\}$ 
(this follows from \cite[Proposition 3.8]{KP}).

Fixing $X\in \cP C$ we obtain the Scott-continuous $\subseteq$-sublinear medial functional
\[\Lambda_C(X) \colon C^*\to \cP\oRp\]
where 
\[\un{\Lambda_C(X)}(f)=\inf_{x\in X}f(x),\ \ \
\ov{\Lambda_C(X)}(f)=\sup_{x\in X}f(x)\]  
%
The collection
$\mathcal{L}_{\subseteq,\rm{med}}(D,\cP \oRp)$ of all Scott-continuous  
$\subseteq$-sublinear medial functionals $F: D^*\to \cP \oRp$  
forms a d-cone semilattice, for any d-cone $D$; indeed, it is a sub-d-cone semilattice of ${\cP \oRp}^{\!\!{D}}$.\\
%

In this way we obtain a d-cone semilattice morphism
\[\Lambda_C \colon \cP C \longrightarrow \mathcal{L}_{\subseteq,\rm{med}}(C^*,\cP \oRp)\]
that represents the convex powercone by the Scott-continuous $\subseteq$-sublinear medial functionals
$F\colon C^*\to \cP \oRp$.

\begin{thm}[{\cite[Proposition 6.4, Theorem 6.8]{KP}}]\label{th:convex}  
Suppose that $C$ is a continuous coherent d-cone. Then we have a d-cone semilattice morphism
$\Lambda_C\colon \cP C \longrightarrow \mathcal{L}_{\subseteq,\rm{med}}(C^*,\cP \oRp)$,
which is an order embedding, where\display
\[\Lambda_C(X)(f)  = [\inf_{x\in X} f(x), \sup_{x\in X} f(x)] \]
Further, if $C$ is convenient then $\Lambda_C$ is an isomorphism.
\qed \end{thm}

We now consider a continuous coherent full \KS\ $K$. 
 The continuous universal d-cone 
$C=\dCone(K)$ is then also continuous and coherent (this last by
Proposition~\ref{prop:KSembedding2}).  
As in
the proof of Theorem \ref{PKeg}, the power \KS\ $\cP K$ can be considered 
to be the collection of all convex lenses in $C$ that are contained in
$K$, which latter is itself  a convex lens of $C$. In this way $\cP K$ can be
seen as a Scott-closed convex $\nonor$-subsemilattice of $\cP C$. 

We then note that a convex lens $X$ is in $\cP K$ iff 
$\ov{\Lambda_C(X)}$ is dominated by the norm functional $\norm{\hspace{0.5pt}\mbox{-}\hspace{0.5pt}}^*_{K}$ (for $X$ is in $\cP K$  iff $\downarrow \!X \subseteq K$ iff $\ov\Lambda_C(\downarrow \!X) \leq \norm{\hspace{0.5pt}\mbox{-}\hspace{0.5pt}}^*_{K}$, by the discussion in Section~\ref{lower}, and we have ${\Lambda_C(X)} =  \Lambda_C(\downarrow \!X)$).
Now, for any normed d-cone $D$, the collection $\mathcal{L}^{\leq 1}_{\subseteq,\rm{med}}(D,\cP \oRp)$ of all Scott-continuous and $\subseteq$-sublinear functionals $F\colon D \to \cP \oRp$ 
with $F$ nonexpansive and medial
 forms a \KS\ semilattice, indeed it is a sub-\KS\ semilattice of $\mathcal{L}_{\subseteq,\rm{med}}(D,\cP \oRp)$   (regarding the latter as a \KS\ semilattice).

We therefore have a \KS\ semilattice morphism
\[\Lambda_K\colon \cP K \longrightarrow \mathcal{L}^{\leq 1}_{\subseteq,\rm{med}}(K^*,\cP \oRp)\]
viz.\  the composition\display
\[\cP K \xrightarrow {\Lambda_C \upharpoonright K} \mathcal{L}^{\leq 1}_{\subseteq,\rm{med}}(C^*,\cP \oRp) \cong \mathcal{L}^{\leq 1}_{\subseteq,\rm{med}}(K^*,\cP \oRp)\]
From
Theorem \ref{th:convex} we then  immediately obtain the following functional
representation theorem\display
\begin{thm}\label{th:convexKS} 
Let $K$ be a continuous coherent full \KS.
  Then we have a \KS\ semilattice morphism
$\Lambda_K\colon \cP K \longrightarrow \mathcal{L}^{\leq 1}_{\subseteq,\rm{med}}(K^*,\cP \oRp)$ which is an order embedding. It is given by\display
\[\Lambda_K(X)(f) \eqdef [\inf_{x\in X} f(x)\, , \sup_{x\in X} f(x)\, ] \]
Further, if $\dCone(K)$ is convenient then $\Lambda_K$ is an isomorphism.
\qed \end{thm}

We specialise this result to domains. For any domain $P$
\[(\Phi_c)_{P}\colon  \cP\oRp^{(\cV_{\leq 1} P)^*} \to  \cP\oRp^{\cL P}\]
where $(\Phi_c)_{P}(F) = F \mapsto  F\circ \EXT_P$, is a d-cone semilattice isomorphism
with inverse $ F \mapsto  F\circ \RES_P$.  %
We then obtain\display
\begin{cor}\label{th:convexdomain}
Let $P$ be a coherent domain. 
Then  we have a  \KS\ semilattice isomorphism
\[\Lambda_P\colon \cP\mathcal{V}_{\leq 1} P \;\cong\;
\mathcal{L}^{\leq 1}_{\subseteq,\rm{med}}(\cL P,\cP \oRp)\]
It is given by\display
\[ \Lambda_P(X)(f) \eqdef [\inf_{\mu \in X} \int \!f\,d\mu\, ,
\sup_{\mu \in X} \int \!f\,d\mu\, ]\]
\end{cor}   
\begin{proof} 
Checking that $(\Phi_c)_{P}$ and its inverse preserve
$\subseteq$-sublinearity and mediality
, we see that ${\Phi_c}_P$ cuts down to an isomorphism 
\[ \mathcal{L}^{\leq 1}_{\subseteq,\rm{med}}((\cV_{\leq 1} P)^*,\cP \oRp) \cong \mathcal{L}^{\leq 1}_{\subseteq,\rm{med}}(\cL P,\cP \oRp)\]
 of \KS\ semilattices.
 As $P$ is a coherent domain, then, following Section~\ref{dme}, we see that $\cV_{\leq 1}P$ is continuous and coherent, and that $\dCone(\cV_{\leq 1}P)$ is convenient. So we may apply Theorem~\ref{th:convexKS}. 
One then displays $\Lambda_P$ as the following composition of \KS\ semilattice isomorphisms\display
 \[\cP\mathcal{V}_{\leq 1} P \xrightarrow{\Lambda_{(\cV_{\leq 1} P)}}
 \mathcal{L}^{\leq 1}_{\subseteq,\rm{med}}((\cV_{\leq 1}
 P)^*,\cP \oRp)   \xrightarrow{(\Phi_c)_{P}} \mathcal{L}^{\leq
   1}_{\subseteq,\rm{med}}(\cL P,\cP \oRp) \tag*{\qEd} \]
\def\popQED{}
\end{proof}

\section{Predicate transformers} \label{pred-tran}

We are ready now to achieve a goal that we can summarise under the
slogan \emph{`the equivalence of state transformer and predicate transformer
semantics'}. 
Let us begin by describing the general framework for the lower and
upper cases. 
In Section \ref{Powerkeg}, in both these cases, we modelled mixed
probabilistic and 
nondeterministic phenomena by a monad $S$  over a full subcategory
$\mathsf C$ of the category of dcpos and Scott-continuous maps. The
Kleisli category of  $S$ has as morphisms the Scott-continuous maps
\[s\colon P\to S(Q)\] We  name these \emph{state transformers}.

In Section \ref{funcrep}, we considered functional representations of $S$, 
which led to a monad $T$ with isomorphisms $S(P)\cong T(P)$. This
monad is a submonad of the continuation monad $\oRp^{{\oRp}^{\! P}}$,
and so its Kleisli category is \makered{faithfully embedded in} the Kleisli
category of the continuation monad whose morphisms are the
Scott-continuous maps
%
 %
  \[t\colon P\to \oRp^{{\oRp}^{\! Q}}\] 
 These morphisms provide our
general notion of state transformer. \makered{The collection
  $(\oRp^{\oRp^Q})^P$ of such state transformers  can be regarded as
  either a  d-cone join-semilattice  or a d-cone meet-semilattice with
  respect to the pointwise structure obtained from $\oRp$, depending
  on whether $\oRp$ is viewed as a d-cone join-semilattice or a d-cone
  meet-semilattice.}

In this setting, it makes sense to think of $\oRp$ as a space of
\emph{truthvalues} and then to call 
 Scott-continuous maps $f\colon P\to\oRp$ on a dcpo $P$
 \emph{predicates}, so that the function 
space $\makeblue{\oRp^P=}\cL P$ 
becomes the dcpo of predicates on $P$. A \emph{predicate
transformer} is then a Scott-continuous map 
\[p\colon \cL Q \to \cL P\]
\makered{and, as before, the collection $(\cL P)^{\cL Q}$ of such
  predicate transformers can  be regarded as either a d-cone
  join-semilattice or a  d-cone meet-semilattice depending on how
  $\oRp$ is viewed.}

There is an evident natural \makered{bijection} 
$\PT\colon (\oRp^{\oRp^Q})^P \cong (\cL P)^{\cL Q}$, where\display
\[\PT(t)(g)(x) \eqdef t(x)(g)  \quad (g \in \cL Q, x \in P)\]
\makered{This bijection is both a d-cone join-semilattice isomorphism
  and a  d-cone meet-semilattice isomorphism, depending on which of
  the above semilattice structures are taken on the state and
  predicate transformers.} 
It is \makered{then} our aim to characterise the `healthy' predicate
transformers, 
that is, those $p$ that correspond to the state transformers 
  $t\colon P\to T(Q)\subseteq \oRp^{{\oRp}^{\!\! Q}}$ 
 arising from \makered{the two} monads for mixed nondeterminism.

For the convex case there is a \makered{natural} 
modification
of this general framework where the role of $\oRp$ is taken by over by
$\cP\oRp$, the convex powercone over $\oRp$. As signalled in
Section~\ref{funcrep}, the uniformity at hand is that, up to
isomorphism, we are making use of the three powercones $\cH \oRp$,
$\cS \oRp$, and $\cP \oRp$. All three are based on $\oRp$, which is
$\cV \mathbf{1}$, the free valuation powerdomain on the
one-point dcpo.

\makeblue{In all cases considered here the `healthy' predicate
  transformers do not preserve the natural algebraic operations on the
  function spaces, that
  is, they are not homomorphisms. In particular, in the lower and upper cases they are   respectively sublinear and superlinear.  This phenomenon is explained from
  a general point of view in \cite{KK1,KK2}.}

As indicated above, we restrict ourselves to predicate
  transformers for the power 
  \KSs\ over domains. There are, nevertheless, related results for
  \KSs\ more generally.  For example, in the lower and upper cases, one
  takes predicates on a \KS\ $K$ to be elements of $K^*$, the sub-\KS\
  of $\oRp^K$ of all Scott-continuous linear functionals.  
Predicate  transformers are suitable Scott-continuous maps 
\[p\colon L^*\to K^*\] 
and state transformers  are linear Scott-continuous maps
\[t\colon K\to\oRp^{L^*}\]  
 In all three cases one obtains \KS\ isomorphisms between \KSs\ of state transformers and \KSs\ of suitably healthy predicate transformers. This differs from the domain case, where one rather obtains \KS\  \emph{semilattice} isomorphisms; the difference arises as there seems to be no general reason why, for example, a dual \KS\ $K^*$ should be a semilattice, whereas, if $K = \cV P$ then $K^*$ is  $\cL P$ which is a semilattice.
%



%


\subsection{The lower case} \label{trans-lower}

Consider two dcpos $P$ and $Q$. 
For any state transformer $t\colon P\to\oRp^{\cL Q}$, the
corresponding predicate transformer $\PT(t)\colon \cL Q\to \cL P$ is given by 
$\PT(t)(g)(x)= t(x)(g)$ for $g\in \cL Q$ and $x\in P$. One can check
directly that
$\norm{\PT(t)(g)}_{\infty}\leq \norm{g}_{\infty}$ for every $g\in \cL Q$ if and only if
$t(x)\leq \norm{\mbox{-}}_{\infty}$ for every $x\in P$, and that $\PT(t)$
is sublinear if and only if $t(x)$ is sublinear for every $x\in P$.
Thus the state transformers $t\colon P\to \SubL^{\leq 1}(\cL Q,\oRp)$
correspond bijectively via $\PT$ to the nonexpansive sublinear predicate
transformers $p\colon \cL Q\to \cL P$.
So $\PT$ cuts down to a \KS\  \makered{join-semilattice} isomorphism 
\[ \SubL^{\leq1}(\cL Q,\oRp)^P \cong \SubL^{\leq1}(\cL Q,\cL P)\]
\makered{taking the pointwise \KS\ join-semilattice structure on $\SubL^{\leq1}(\cL Q,\oRp)^P$.}
Finally, to make the link between  state transformers   \makered{and the healthy} predicate transformers,  we  set \\
$\PT_{{P},{Q}}(s)  \eqdef \PT( \Lambda_{Q} \circ s)$, ($s\colon {P}\to \cH {\cV_{\leq 1} Q}$), 
 where $\Lambda_{Q}$ is as in Section~\ref{lower}, and calculate that
\[\PT_{{P},{Q}}(s)(g)(x) = \PT(\Lambda_{Q} \circ s)(g)(x) = \Lambda_{Q}(s(x))(g)  =  \sup_{\mu \in s(x)}\int\! g\, d\mu\]
Combining the above discussion with Corollary~\ref{th:lowerdomain} we then obtain\display   
%
%


\begin{cor} \label{th:lowerdomainPred}
Let $P$ and $Q$ be  dcpos. To every state
transformer $s\colon P \to \cH \cV_{\leq 1}Q$ we can assign  a
predicate transformer $\PT_{P,Q}(s)\colon \cL Q\to \cL P$  by\display
\[\PT_{P,Q}(s)(g)(x) \eqdef \sup_{\mu \in s(x)}\int\! g\, d\mu\quad
(g\in \cL Q, x \in P)\]
The predicate transformer $\PT_{P,Q}(s)$ is sublinear and nonexpansive.
The assignment $\PT_{P,Q}$ is a \KS\ \makered{join-semilattice} morphism 
\[(\cH  \cV_{\leq 1} Q)^P \longrightarrow \SubL^{\leq1}(\cL Q,\cL P)\]
 If $Q$ is 
a domain then it is an isomorphism. 
\qed\end{cor}  


Similar to the simplification of nonexpansiveness for functionals discussed in Section~\ref{lower} (and an immediate consequence of it), the condition of nonexpansiveness has a simple formulation for homogeneous predicate transformers  $p \colon   \cL Q\to \cL P$,  viz.\  $p(\mathbf{1}_Q)\leq \mathbf{1}_P$.

\subsection{The upper case} \label{trans-upper}

Consider two dcpos, $P$ and $Q$. In agreement with the terminology
introduced in Section \ref{upper} we will say that a
predicate transformer $p\colon \cL Q \to \cL P$ is \emph{strongly
  nonexpansive} if we have
\[p(f+g)\leq p(f) + \norm{g}_{\infty}\cdot {\mathbf 1}_P\]
%
for all $f,g \in \cL Q$, where ${\mathbf 1}_P$ is the constant function on $P$ with value $1$. 
Strongly nonexpansive predicate transformers are nonexpansive. Indeed,
for $f=0$ the inequality yields $p(g)(x)\leq\norm{g}_{\infty}$
for all $x\in P$, whence $\norm{p(g)}_{\infty}\leq \norm{g}_{\infty}$.
%
In the case of homogeneous predicate transformers this can be
simplified to the equivalent condition \display 
\[p(f +{\mathbf 1}_Q)\leq p(f)+{\mathbf 1}_P\quad   (\mbox{for all } f \in \cL Q)\]
as follows immediately from the corresponding simplification for
functionals  in  Section~\ref{upper}.

For any state transformer $t\colon P\to\oRp^{\cL Q}$, we have
$\PT(t)(g)(x)=t(x)(g)$, ($g\in 
\cL Q, x\in P$). This firstly implies that $t(x)$ is superlinear for
every $x$ if, and only if, $\PT(t)$ is superlinear. It secondly
implies that $t(x)$ 
is strongly nonexpansive for every $x$ if, and only if, $\PT(t)$ is
strongly nonexpansive.
For we have $t(x)(f + g) \leq t(x)(f) + \norm{g}_{\infty}$ for every $x \in P$ if, and only if, $\makered{\PT(t)}(f+g)(x) \leq \makered{\PT(t)}(f)(x)+ \norm{g}_{\infty}$ for every $x \in P$, that is, if, and only if, $\makered{\PT(t)}(f+g)\leq \makered{\PT(t)}(f)+ \norm{g}_{\infty}\cdot \mathbf{1}_P$.

We write $\SuperL^{\rm{sne}}(\cL Q,\cL P)$ for the set of strongly nonexpansive
superlinear predicate transformers. 
and note that it forms a sub-\KS\ \makered{meet-semilattice} of
$({\cL P})^{\cL Q}$ \makered{(taking the pointwise \KS\ meet-semilattice structure on $({\cL P})^{\cL Q}$)}. So $\PT$ cuts down to a \KS\  \makered{meet-semilattice}  isomorphism  
\[ \SuperL^{\rm{sne}}(\cL Q,\oRp)^P \cong \SuperL^{\rm{sne}}(\cL Q,\cL P)\]

Finally, to make the link between state transformers \makered{and the healthy} predicate transformers, we set $\PT_{P,Q}(s)
\eqdef \PT( \Lambda_Q \circ s)$ ($s\colon P\to \cS\cV_{\leq 1} Q$), where $\Lambda_Q$ is as in
Section~\ref{upper} (and assuming
now that $Q$   is a domain), and calculate that 
\[PT_{P,Q}(s)(g)(x) =  
\inf_{\mu \in s(x)}\int\! g \, d\mu\]
Combining the above discussion with Corollary~\ref{th:upperdomain} we  then
obtain\display    
%
.


%
%

\begin{cor} \label{th:upperdomainPred}
Let $P$ be a dcpo and let $Q$ be a domain. To every state
transformer $s\colon P \to \cS \cV_{\leq 1}Q$ we can assign  a predicate transformer $\PT_{P,Q}(s)\colon
\cL Q\to \cL P$  by\display
\[\PT_{P,Q}(s)(g)(x) \eqdef \inf_{\mu \in s(x)}\int\! g \, d\mu\quad
(g\in \cL Q, x \in P)\]
The predicate transformer $\PT_{P,Q}(s)$ is superlinear and strongly
nonexpansive. 
The assignment $\PT_{P,Q}$ is a \KS\ \makered{meet-semilattice} isomorphism 
\[(\cS \cV_{\leq 1} Q)^P \cong \SuperL^{\rm{sne}}(\cL Q,\cL P) \tag*{\qEd}\]
\end{cor}

\subsection{The convex case} \label{trans-convex}

For this case we have
to modify our framework. First we need a function space construction.
%
\makered{For d-cone semilattices $C$ and $D$, the collection $\cL_{\mathrm{mon}}(C,D)$ of Scott-continuous $\subseteq$-monotone maps from $C$ to $D$} is a sub-d-cone semilattice of the d-cone semilattice of all Scott-continuous maps from $C$ to $D$ equipped with the pointwise d-cone semilattice structure.

%


In the modified framework, the r\^{o}le of $\oRp$, considered as join- and meet-semilattices in the
lower and upper cases, is  taken over by $\cP\oRp$, the
convex powercone over $\oRp$.
Predicates on a dcpo $P$ are no longer  functionals with
values in $\oRp$ but are now rather Scott-continuous 
functionals of the form $f\colon P \to \cP\oRp$; \makered{they form a d-cone semilattice with the pointwise structure}.  We define a norm on predicates $f\type P \to \cP\oRp$ by\display $\,\norm{f} \eqdef \norm{\ov{f}}_{\infty} \makered{(= \bigvee_{x \in P}\ov{f}(x))}$. 
Employing the notation of Section~\ref{convex} we have a bijection $f \rightarrow (\un{f},\ov{f})$
between predicates and pairs of linear functionals $g,h \in \cL P$
with $g \leq h$. 
\makered{Note that $(\un{f},\ov{f}) \leq (\un{f'},\ov{f'})$ if, and only if, $\un{f} \leq \un{f'}$ and 
$\ov{f} \leq \ov{f'}$, that $(\un{f},\ov{f}) \cup (\un{f'},\ov{f'}) = (\un{f} \wedge \un{f'},\ov{f} \vee \ov{f'})$, and that $(\un{f},\ov{f}) \subseteq (\un{f'},\ov{f'})$ if, and only if, $\un{f} \geq \un{f'}$ and 
$\ov{f} \leq \ov{f'}$.}


%
We take general state transformers to be maps\display
  \[t\colon P\to {\cP \oRp}^{{\oRp}^{\! Q}}\] 
One might rather have expected, $t\colon P\to \cP \oRp^{{\cP \oRp}^{Q}}$, uniformly replacing $\oRp$ with $\cP \oRp$; we chose our definition to be closer to the functional representation.

Predicate transformers are taken to be Scott-continuous maps
 \[p\colon (\cP\oRp)^Q \to (\cP\oRp)^P\] 
 which are, in addition, required to
preserve the partial order $\subseteq$, i.e., to be $\subseteq$-monotone. 
\makered{This requirement is a technical condition to achieve an isomorphism between general state transformers and predicate transformers (see below).}
\makered{ The predicate transformers form a d-cone semilattice $\cL_{\mathrm{mon}}((\cP\oRp)^Q, (\cP\oRp)^P)$.}
  Note that a predicate transformer $p$ is nonexpansive if, and only if, $\norm{\ov{p}(f)}_{\infty} \leq \norm{\ov{f}}_{\infty}$, for any predicate $f$.

We link these predicate transformers to  \makered{general state transformers via} `predicate transformers of diagonal form' which we take to be Scott-continuous functions:
\[q\colon \cL Q \longrightarrow  (\cP\oRp)^P\]
State transformers $t\colon P\to (\cP\oRp)^{{\cL Q}}$ are connected to predicate transformers of diagonal form by the map
\[\mathrm{T}\colon (\cP\oRp^{{\cL Q}})^P\longrightarrow ((\cP\oRp)^P)^{{\cL Q}} \]
where $\mathrm{T}(t)(g)(x) \eqdef t(x)(g)$. \makered{This map is evidently an isomorphism of d-cone semilattices, with respect to the pointwise structures.}

\makered{To connect predicate transformers of diagonal form to predicate transformers} we first extend some definitions from predicates to functions $F\colon D \to (\cP\oRp)^P$, with $D$ a d-cone. Let $F$ be such a function.
We define $\un{F},\ov{F}\colon D \to  \cL P$  by setting $\un{F}(x) = \un{F(x)}$ and $\ov{F}(x) = \ov{F(x)}$, for $x \in D$.
%
 %
%
Then define a map $\mathrm{P}$ between the two kinds of predicate transformers\display
 \[\mathrm{P}\colon ((\cP\oRp)^P)^{\cL Q} \longrightarrow  \cL_{\mathrm{mon}}((\cP\oRp)^P, (\cP\oRp)^Q)\]
 by
 $P(q)(f) \eqdef  [\un{q(\un{f})},\ov{q(\ov{f})}])$.
 \begin{lem} $\mathrm{P}$ is an isomorphism of \makered{d-cone semilattices}.
 \end{lem}
 \begin{proof} It is routine to verify that $\mathrm{P}$ is a morphism of \makered{d-cone semilattices}. To see that $\mathrm{P}$ is an order embedding suppose that 
 $\mathrm{P}(q) \leq \mathrm{P}(q')$ and choose $g\in L^*$ to show that $q(g) \leq q'(g)$. Then we have $[\un{q(\un{f})},\ov{q(\ov{f})}]) \leq [\un{q'(\un{f})},\ov{q'(\ov{f})}])$ where $f = [g,g]$. So $\un{q(g)} \leq \un{q'(g)}$ and $\ov{q(g)} \leq \ov{q'(g)}$, and so $q(g) \leq q'(g)$, as required.
 
 To see that $\mathrm{P}$ is onto, choose a predicate transformer $p$ to find a $q$ with $p = \mathrm{P}(q)$. We claim 
that \makered{$\un{p([\un{f},\ov{f}])} = \un{p([\un{f},\un{f}])}$} and 
\makered{$\ov{p([\un{f},\ov{f}])} = \ov{p([\ov{f},\ov{f}])}$}. 
\makered{For the first of these claims, as $[\un{f},\un{f}] \leq [\un{f},\ov{f}]$ we have $\un{p([\un{f},\un{f}])} \leq \un{p([\un{f},\ov{f}])}$, since $p$
preserves the order $\leq$, and as $[\un{f},\un{f}] \subseteq [\un{f},\ov{f}]$ we have $\un{p([\un{f},\un{f}])} \geq \un{p([\un{f},\ov{f}])}$, since 
$p$ preserves the order $\subseteq$. 
The proof of  the second of these claims is similar: as $[\un{f},\ov{f}] \leq [\ov{f},\ov{f}]$ we have $\ov{p([\un{f},\ov{f}])} \leq \ov{p([\ov{f},\ov{f}])}$, since $p$
preserves the order $\leq$, and as $[\un{f},\ov{f}] \supseteq [\ov{f},\ov{f}]$ we have $\ov{p([\un{f},\ov{f}])} \geq \ov{p([\ov{f},\ov{f}])}$, since 
$p$ preserves the order $\subseteq$. }


Defining $q = g \mapsto p([g,g])$, we then see that $p = \mathrm{P}(q)$, as required.
 \end{proof}
 
 So we have a \KS\ semilattice isomorphism between general state transformers and predicate transformers\display
 \[\mathrm{P} \circ \mathrm{T}\colon (\cP\oRp^{{\cL Q}})^P \cong \cL_{\mathrm{mon}}((\cP\oRp)^P, (\cP\oRp)^Q)\]
 and we seek the relevant healthiness conditions on the predicate transformers.

Define a function $F\colon D \to (\cP\oRp)^P$, with $D$ a d-cone and $P$ a dcpo, to be \emph{medial} if\display
%
\[\un{ F}(x + y)\leq \un{ F}(x)+\ov{ F}(y)\leq
\ov{ F}(x + y)\]
for all $x,y \in D$,
and define a function $F\display D \to C$, where $D$ is a d-cone and $C$ is a 
 d-cone \makered{semilattice}, to be 
\emph{$\subseteq$-sublinear} if it is homogeneous and $F(x + y) \subseteq F(x) + F(y)$, for all $x,y \in D$ (this generalises the definition of $\subseteq$-sublinearity in Section~\ref{convex}, and in the case where $C$ is $(\cP\oRp)^P$, it is equivalent to $\ov{F}$ being sublinear and $\un{F}$ being superlinear).

Now fix a state transformer $s$ and set $q = \mathrm{T}(s)$ and $p = \mathrm{P}(q)$. 
We have $\ov{t(x)}$ sublinear for every $x\in P$ iff $\ov{q}$ is sublinear iff $\ov{p}$ is sublinear  
and, similarly, $\un{t(x)}$ is superlinear for every $x\in P$ iff  $\un{p}$ is superlinear.
\makered{So $t(x)$ is $\subseteq$-sublinear for every $x\in P$ iff $q$ is $\subseteq$-sublinear iff $p$ is $\subseteq$-sublinear.}
Next, $t(x)$ is medial 
for all $x\in P$ iff $q$ is medial 
iff $p$ is medial. 
Finally, \makered{$t(x)$ is nonexpansive for all $x \in P$}
 iff $\ov{t(x)}\leq \norm{\mbox{-}}_{\infty}$ for all $x\in P$ 
 iff $\norm{\ov{q}(g)}_{\infty} \leq \norm{g}_{\infty}$ for all $g\in {\cL Q}$, 
  iff $\norm{\ov{p}(f)}_{\infty} \leq \norm{\ov{f}}_{\infty}$ for all predicates $f$,
  that is, iff $p$ is nonexpansive. 

We write $ \cL^{\leq  1}_{\mathrm{mon},\subseteq,\rm{med}}((\cP\oRp)^Q,(\cP\oRp)^P)$
for the set of $\subseteq$-monotone, $\subseteq$-sublinear, medial,
nonexpansive  predicate 
transformers. 
\makered{As is straightforwardly checked}, it forms a sub-\KS\ \makered{semilattice} of $ \cL_{\mathrm{mon}}((\cP\oRp)^Q,(\cP\oRp)^P)$,
 and, from the above considerations, we see that $\mathrm{P} \circ \mathrm{T}$ cuts down to a \KS\ semilattice isomorphism\display
\[\mathcal{L}^{\leq 1}_{\subseteq,\rm{med}}(\cL Q,\cP \oRp)^P \cong \cL^{\leq  1}_{\mathrm{mon},\subseteq,\rm{med}}((\cP\oRp)^Q,(\cP\oRp)^P) \]

Finally, to make the link between state transformers and predicate transformers  we set \\
 $\PT_{P,Q}(s) \eqdef \mathrm{P} ( \mathrm{T} (\Lambda_Q \circ s))$, ($s\colon P\to \cP \cV_{\leq 1}Q$ ), where $\Lambda_Q$ is as in Section~\ref{convex} (and assuming now that $Q$  is a coherent domain), and calculate that 
\[\PT_{P,Q}(s)(g)(x) =  [\inf_{\mu \in s(x)}\int\! \un{g} \, d\mu, \sup_{\mu \in s(x)}\int\! \ov{g} \, d\mu\ ]  \]
Combining the above discussion with  Corollary~\ref{th:convexdomain}, 
we then obtain\display

\begin{cor} \label{th:convexdomainPred}
Let $P$ be a dcpo and let $Q$ be  a coherent domain. To every state
transformer $s\colon P \to \cP \cV_{\leq 1}Q$ we can assign  a predicate transformer $\PT_{P,Q}(s)\colon
\cP\oRp^Q \to \cP\oRp^P$  by\display
\[\PT_{P,Q}(s)(g)(x) \eqdef [\inf_{\mu \in s(x)}\int\! \un{g} \, d\mu,
\sup_{\mu \in s(x)}\int\! \ov{g} \, d\mu\ ]        \quad (g\in \cP\oRp^Q, x \in P)\]
The predicate transformer $\PT_{P,Q}(s)$ is  nonexpansive,
$\subseteq$-monotone, $\subseteq$-sublinear, and medial. 
The assignment $\PT_{P,Q}$ is a \KS\ \makered{semilattice} isomorphism 
\[(\cP\cV_{\leq 1} Q)^P \cong  \cL^{\leq
    1}_{\mathrm{mon},\subseteq,\rm{med}}(\cP\oRp^Q,\cP\oRp^P) \tag*{\qEd}\]
\end{cor}  

%


\section{The unit interval} \label{unit-interval}

\newcommand{\cLh}{\cL_{\mathrm{hom}}}

In this section we consider replacing the extended positive reals
$\oRp$ by the unit interval $\I$. In the lower and upper cases,
functional representations will involve maps to $\I$; $\I$ will play
the r\^{o}le of truth values for predicates; and predicate
transformers will be functions from  
$\cL_{\leq 1}Q$ to $\cL_{\leq 1}P$. In the convex case, functional
representations will involve maps to $\cP \I$; $\cP \I$ will play the
r\^{o}le of truth values;  and predicate transformers will be
functions from  $\cP \I^Q$ to $\cP \I^P$. As we shall see, the results
obtained are the same as those with $\oRp$, except that
nonexpansiveness requirements are dropped.

First, we slightly weaken the notion of a norm introduced in Section
\ref{funcrep} deleting the requirement that nonzero elements have
nonzero norm. A \emph{seminorm} on a \KS\ $K$ is
defined to be a Scott-continuous sublinear map from $K$ to $\oRp$
and a \emph{seminorm} on a cone $C$ is a Scott-continuous sublinear
map from $C$ to $\oRp$. A seminormed \KS\ is a \KS\ equipped with a
seminorm~$\norm{\!\!-\!\!}$, and similarly for 
seminormed cones. A function $f:K \rightarrow L$ between seminormed
\KSs\  is \emph{nonexpansive} if, for all $a \in K$ we have: 
\[\norm{f(a)} \leq \norm{a}\]
and similarly for seminormed cones.
A \emph{seminormed d-cone semilattice} is a d-cone semilattice equipped with a
seminorm such that the operation $\nonor$  is nonexpansive, by which we mean that
 $\norm{a\nonor b}\leq \max(\norm{a},\norm{b})$.

We next need some function space constructions.
For any \KS\ $K$ and \KS\ $L$ (\KS\ semilattice $L$) we write
$\cLh(K,L)$ for the Scott-continuous homogeneous functions from $K$ to
$L$. Equipped with the pointwise structure, $\cLh(K,L)$ forms a
sub-\KS\ (respectively, sub-\KS\ semilattice) of $L^K$; further, for
any d-cone $C$ (d-cone semilattice $C$),  $\cLh(K,C)$ (regarding $C$
as a \KS) forms a sub-d-cone (respectively sub-d-cone semilattice) of
$C^K$ when equipped with the pointwise structure.

For seminormed d-cones $C$ and $D$, 
we write $\cLh^{\leq 1} (C,D)$ for the collection of all
Scott-continuous, homogeneous, nonexpansive functions from $C$ to
$D$. Equipped with the pointwise structure it forms a sub-\KS\  
of  $\cLh(C,D)$
; further, 
if $D$ is a seminormed d-cone semilattice, $\cLh^{\leq 1}
(C,D)$  forms a sub-\KS\ semilattice of $\cLh(C,D)$. 

We have a basic function space isomorphism as an immediate consequence
of the universal embedding in a d-cone of a  full \KS\ 
given by  Theorem~\ref{prop:KSembed}.  
Let $e\type K \rightarrow C$ be a universal \KS\ embedding (in the
sense of Section~\ref{Keg}) of a   \KS\ $K$ in a d-cone $C$. Then
Theorem~\ref{prop:KSembed} tells us that, for any d-cone $D$, function
extension $f \mapsto \ov{f}$ yields a dcpo isomorphism  
\[\cLh(K,D) \cong \cLh(C,D)\]
with inverse given by restriction $g \mapsto g\circ e$ along the
universal arrow. Moreover, as restriction preserves the pointwise
structure,  the isomorphism is an isomorphism of d-cones; further, if
$D$ is additionally equipped with a semilattice structure, then the
isomorphism is an isomorphism of d-cone semilattices.

To connect nonexpansiveness and \KSs\ we make use of particular
seminorms. For any Scott-closed convex subset $X$ of a cone $C$, we define
the \emph{(lower) Minkoswki functional} $\nu_X\type C \rightarrow
\oRp$ by\display  
\[\nu_X(a)  \eqdef \inf \{r \in \Rp \mid a \in r\cdot X\}\]
Minkowski functionals were previously considered in~\cite{Plo06} and
in \cite{KK}.
%
\begin{prop}\label{prop:mink}
 Let $X$ be a Scott-closed convex subset of a d-cone $C$. Then\display 
\begin{enumerate}
\item $\nu_X$ is Scott-continuous and sublinear.
\item If $0 < \nu_X(a) < \infty$ then, for some $x \in X$, we have $a =
  \nu_X(a)\cdot x $. 
\item $X = \{a \in C\mid \nu_X(a) \leq 1\}$.
\item If $C$ is a d-cone semilattice and $X$ also a subsemilattice,
  then we have\display \[\nu_X(a\cup b) \leq \max(\nu_X(a),\nu_X(b))\]
\end{enumerate}
\end{prop}
\begin{proof}\hfill
\begin{enumerate}
\item 
  \begin{enumerate}[label=(\alph*)]
  \item  For monotonicity, suppose $a \leq b \in C$. Then if $b
    \in r\cdot X$, we have $a \in r\cdot X$, since $r\cdot X$ is a lower set
    and so  $\nu_X(a) \leq \nu_X(b)$. 
  \item Having established monotonicity, for continuity it
    remains  to show that $\nu_X(\bigvee_i a_i) \leq \bigvee_i
    \nu_X(a_i)$, for any directed set $a_i \, (i \in I)$  of elements
    of $C$. Suppose that   $\bigvee_i \nu_X(a_i)< r$ for some $r \in
    \Rp$. Then, for every $i$, $\nu_X(a_i) < r$ and so $a_i \in r\cdot
    X$. Since $r\cdot X$ is Scott-closed, we also have  $\bigvee_ia_i \in r\cdot
    X$ and consequently $\nu_X(\bigvee_i a_i) \leq r$. As $r$ is an arbitrary
    element of $\Rp$ with $\bigvee_i \nu_X(a_i)< r$ this shows that
    $\nu_X(\bigvee_i a_i) \leq \bigvee_i \nu_X(a_i)$, as required. 
   \item For homogeneity, choose $a \in C$ and $r \in ]0,1[$ to
     show that $\nu_X(r \cdot a) = r \cdot \nu_X(a)$. This follows
     from the observation that, for any positive $s \in \Rp$,  $a \in
     s\cdot X$ iff $r\cdot a \in rs\cdot X$.  
  \item For subadditivity, choose $a$, $b$ in $C$ and
    $r>\nu_X(a)$, $s>\nu_X(b)$. Then $a\in r\cdot X$ and $b\in s\cdot
    X$ whence $a+b\in r\cdot X + s\cdot X = (r+s)\cdot X$, since $X$
    is convex. So we have $r+s>\norm{a+b}$ and, since this holds for
    all  $r>\nu_X(a)$ and $s>\nu_X(b)$, we conclude that
    $\nu_X(a+b)\leq \nu_X(a) +\nu_X(b)$. 
\end{enumerate} 
\item As $0 < \nu_X(a) < \infty$ there are sequences $r_n \in \Rp$ and
  $x_n \in X$, with $r_n$ decreasing and positive, such that $a =
  r_n\cdot x_n $ and  $\nu_X(a) = \inf r_n$.  
So $x_n = r_n^{-1}\cdot a $ is an increasing sequence, and taking sups
we see that  $\sup x_n = (\sup r_n^{-1})\cdot a = \nu_X(a)^{-1}\cdot
a$. We have
$x \eqdef\sup x_n\in X$, as $X$ is Scott-closed, and so $\nu_X(a)\cdot
x = a$. 
\item Evidently  $X \subseteq \{a \in C\mid \nu_X(a) \leq
  1\}$. Conversely, suppose that we have $a \in C$ with $\nu_X(a) \leq
  1$. If $\nu_X(a) < 1$ then clearly $a\in 1\cdot X=X$.  Otherwise we
  have $\nu_X(a) = 1$. In this case, by the second part we have $a =
  \nu_X(a)\cdot x$ for some $x \in X$. Then, as $\nu_X(a) = 1$, we see
  that $a \in X$, as required.   
\item As for subadditivity, choose any real number $r >
  \max(\nu_X(a),\nu_X(b))$. Then $a$ and $b$ are both in $r\cdot
  X$. Since $X$ is supposed to be a subsemilattice and $x\mapsto
  r\cdot x$ is a semilattice homomorphism, $r\cdot X$ is
  subsemilattice, too, so that $a\cup b\in r\cdot X$, 
  that is $\nu_X(a\cup b)\leq r$. Since this holds for
    all  $r>\max(\nu_X(a), \nu_X(b))$, we conclude that
    $\nu_X(a\cup b)\leq \max(\nu_X(a) ,\nu_X(b))$.   \qedhere
\end{enumerate}
\end{proof}

%

So all Minkowski functionals are seminorms. For every full   
\KS\ $K$
and every universal \KS\ embedding $K
\xrightarrow{e} C$, we write $\norm{\hspace{0.5pt}\mbox{-}\hspace{0.5pt}}_K$ for the seminorm
$\nu_{e(K)}\type C \rightarrow \oRp$, and on $K$ we use the same
notation for the seminorm $\norm{a}_K \eqdef \norm{e(a)}_K$. 
When we do not
mention below which seminorm we use we mean the relevant one of these. 
We have the following pleasant facts\display 


\begin{fact} \label{Kegne}
\makered{For full \KSs\ $K$ and $L$, l}et $K \xrightarrow{e} C$ and $L \xrightarrow{e'} D$ be universal
\KS\ embeddings. Suppose that $f\colon K\to L$ and and $g\colon C\to
D$ are homogeneous maps 
such that the following diagram commutes: 
\[\begin{diagram}
	{C} & \rTo^{g} & D\\
	 \uTo^{e} & & \uTo_{e'} \\
	 {K} & \rTo_{f} & L
	\end{diagram}\]
%
Then both $f$ and $g$ are nonexpansive.  
\end{fact}

\begin{proof}
It suffices to show that $g$ is nonexpansive. So, 
for $a\in C$ we have to show that $\norm{g(a)}_L\leq \norm{a}_K$. This is
certainly true if $\norm{a}_K= +\infty$. Otherwise take any real number
$r>\norm{a}_K$. Then $a\in r\cdot e(K)$. We deduce that $g(a) \in
g(r\cdot e(K)) = r\cdot g(e(K))=r\cdot e'(f(K))\subseteq r\cdot e'(L)$
which implies that $\norm{g(a)}_L\leq r$. Since this holds for all
$r>\norm{a}_K$, we have the desired inequality.  
\end{proof}

The next proposition is at the root of our results for the unit
interval. It enables nonexpansiveness requirements to be dropped when
using the unit interval. 
First define $K \xrightarrow{e} C$ to be a \emph{universal \KS\
  semilattice embedding} if $K$ is a full \KS\ semilattice, $D$ is a d-cone
semilattice, $e$ preserves the semilattice operation, and $K
\xrightarrow{e} C$ is a universal \KS\ embedding. Note that then the
norm $\norm{\hspace{0.5pt}\mbox{-}\hspace{0.5pt}}_K$ on $C$ satisfies property (4) of
Proposition \ref{prop:mink}.

\begin{prop} \label{fund-unit} \makered{For full \KSs\ $K$
    and $L$, l}et $K \xrightarrow{e} C$ and $L
  \xrightarrow{e'} D$ be  universal \KS\  embeddings.  
Then there is a \KS\ isomorphism:
\[\cLh(K,L) \cong \cLh^{\leq 1} (C,D)\] 
The isomorphism sends $f\in \cLh(K,L)$ to $\ov{e'\circ f}$; its
inverse sends $g \in  \cLh^{\leq 1} (C,D)$ to the restriction of
$g\circ e$ along $e'$; and $f$ and $g$ are related by the isomorphism if, and only if, $g\circ e = e'\circ f$.

In case $L \xrightarrow{e'} D$ is additionally  a universal \KS\
semilattice embedding, the isomorphism is a \KS\ semilattice
isomorphism. 
\end{prop}

\begin{proof} Composing with the \KS\ embedding $e'$ and then function
  extension, viewed as a \KS\  isomorphism, we obtain a  \KS\
  embedding\display 
\[\cLh(K,L) \xrightarrow{e' \circ -} \cLh(K,D) \xrightarrow{\ov{\cdot}}  \cLh(C,D)\] 
We see from 
Fact~\ref{Kegne}  that every function in the
range of the embedding is nonexpansive. 
Conversely let  $g: C \rightarrow D$ be a nonexpansive
Scott-continuous homogeneous function. Then, in particular, for every
$a \in K$ we have  
$\norm{g(e(a))} \leq \norm{e(a)} \leq 1$ and so there is a
(necessarily unique) $b \in L$ such that $g(e(a)) = e'(b)$. So we have
a function $f\type K \rightarrow L$ such that $e'(f(a)) = g(e(a))$ for
all $a \in K$; this function is Scott-continuous and homogeneous as
$g$ is and $e$ and $e'$ are \KS\  embeddings. 
 As $g$ extends $f\circ e'$ along $e$, the \KS\ embedding $f \mapsto
 \ov{f\circ e'}$ of $\cLh(K,L)$ in $ \cLh(C,D)$ sends $f$ to $g$,  
  and so   cuts down to a bijection, and so a \KS\ isomorphism,
  between $\cLh(K,L)$ and $\cLh^{\leq 1}(C,D)$, with inverse as
  claimed. 
 
 That $f$ and $g$ are related by the isomorphism if, and only if, $g\circ e = e'\circ f$ is clear.
 
%
%

With the extra semilattice assumptions, $D$ is a d-cone
semilattice and so $\cLh^{\leq 1} (C,D)$ is a \KS\ semilattice;
further,  the isomorphism preserves the semilattice structure as $e'$
does. 
\end{proof}

\makered{We will typically apply this result by first restricting to a subclass (e.g., to sublinear functions in the lower case) and then specialising to dcpos or domains.  
 
 We could also obtain general results for \KSs\ $K$ by adding to the
assumptions considered above the assumption that the evident embedding
of $K^*$ in $\dCone(K)^*$ is universal, where now by $K^*$ we mean the
\KS\ of Scott-continuous linear functions from $K$ to $\I$. By
Proposition~\ref{uni-Keg}, an equivalent assumption, assuming
$\dCone(K)^*$ continuous, is that every Scott  continuous linear
function from $K$ to $\Rp$ is a directed sup of bounded such
functions. 

 One obtains general functional representation and predicate transformer results, except for predicate transformer results in the convex case. (The obstacle in that case is that the equational proof below that mediality transfers in the domain case is not available at a general level, since there seems to be no general reason why dual \KSs\ or dual d-cones should be semilattices.) 

 As remarked in the introduction,  one might prefer a development  not involving d-cones at all. Another improvement, perhaps easier to achieve, would be a development where the assumptions on \KSs\  involved only $\I$.}

%


There is a pleasant induction principle for d-cones given by \KS\ universal
embeddings. Say that a property  of a d-cone is \emph{upper
  homogeneous} if it is closed under all actions $r\!\cdot\! -$ with
$r \geq 1$. Then, for any \makered{full \KS\ $K$ and any} universal
\KS\ embedding $K \xrightarrow{e} 
C$, if a property of $C$  is closed under directed sups,  and is upper
homogeneous, then it holds for all of $C$ if it holds for all  of
$e(K)$. This can be proved by reference to  the construction of
universal embeddings in Section~\ref{Keg}. 

There is an $n$-ary version
of this induction principle. Given  $n > 0$ universal embeddings
$K_i \xrightarrow{e_i} C_i$ \makered{ of full \KSs\ $K_i$} ($i = 1,\dots,n$),  if a relation on
$C_1,\ldots, C_n$ is closed under directed sups and is upper
homogeneous (in  an evident sense), then the relation holds for all of
$C_1 \times \ldots \times C_n$ if it holds for all  of $e_1(K_1)\times
\ldots \times e_n(K_n)$. This follows from the unary principle since, as shown in Section~\ref{dCs_KSs}, universal \KS\ embeddings are closed under finite non-empty products.

\subsection{The lower case}

For any full \KSs\ $K$ and $L$ we take $\SubL(K,L)$ to be the collection of Scott-continuous sublinear functions from $K$ to $L$ equipped with the pointwise structure and so forming a sub-\KS\ of $L^K$, and a sub-\KS\ join-semilattice if $L$ is a \KS\ join-semilattice.

 Let  $e\type K \rightarrow C$ and $e'\type L \rightarrow D$ be   universal  \KS\  embeddings,  where, additionally, $L$ is a \KS\ join-semilattice and $D$ is a d-cone join-semilattice (when $e$ is automatically a \KS\ semilattice morphism and so $e'\type L \rightarrow D$ is  a universal \KS\   semilattice embedding). Then the  \KS\  semilattice isomorphism
of Proposition~\ref{fund-unit} restricts to a \KS\ semilattice isomorphism 
\[\SubL(K,L) \cong \SubL^{\leq 1} (C,D)\] 
as an $f\in \cLh(K,L)$ is sublinear iff $e'\circ f$ is iff (using Theorem~\ref{prop:KSembed}) $\ov{e'\circ f}$ is.

We saw in Section~\ref{dme} that the inclusion $\cL_{\leq 1}P
\subseteq \cL P$ is a universal embedding for any dcpo $P$, and it is easy to check that the norm $\norm{\hspace{0.5pt}\mbox{-}\hspace{0.5pt}}_{\infty}$ is the Minkowski seminorm, thereby ensuring consistency with the previous two sections. Further, $\cL_{\leq 1}P$ is a \KS\ 
join-semilattice and $ \cL P$ is a d-cone join-semilattice. 
The analogous remarks apply to the inclusion $\I \subseteq \oRp$.  

So, for any dcpo $P$ we obtain the \KS\ semilattice isomorphism\display
\[\SubL(\cL_{\leq 1}P,\I) \cong \SubL^{\leq 1} (\cL P,\oRp)\] 
and for any dcpos $P$ and $Q$ we obtain the \KS\ semilattice  isomorphism\display
\[\SubL(\cL_{\leq 1} Q, \cL_{\leq 1}P) \cong \SubL^{\leq 1} (\cL Q, \cL P)\] 
As immediate consequences of Corollaries~\ref{th:lowerdomain} and~\ref{th:lowerdomainPred} we then obtain\display

\begin{cor}\label{th:lowerdomain-unit}
Let $P$ be a dcpo. 
Then  we have a  \KS\  join-semilattice  morphism\display
\[\Lambda_P\colon \cH\mathcal{V}_{\leq 1} P \;\longrightarrow\;
\SubL(\cL_{\leq 1}  P,\I)\]
It is given by\display
\[ \Lambda_P(X)(f) \eqdef \sup_{\nu \in X} \int \!f\,d\nu\]
If $P$ is a domain then $\Lambda_P$ is an isomorphism.
\qed \end{cor}

\begin{cor} \label{th:lowerdomainPred-unit}
Let $P$ and $Q$ be  dcpos. To every state
transformer $s\colon P \to \cH \cV_{\leq 1}Q$ we can assign  a
predicate transformer $\PT_{P,Q}(s)\colon \cL_{\leq 1} Q\to \cL_{\leq 1} P$  by\display
\[\PT_{P,Q}(s)(g)(x) \eqdef \sup_{\nu \in s(x)}\int\! g\, d\nu\quad
(g\in \cL_{\leq 1} Q, x \in P)\]
The predicate transformer $\PT_{P,Q}(s)$ is sublinear.
The assignment $\PT_{P,Q}$ is a \KS\  join-semilattice morphism 
\[(\cH  \cV_{\leq 1} Q)^P \longrightarrow \SubL(\cL_{\leq 1} Q,\cL_{\leq 1} P)\]
If $Q$ is 
a domain then it is an isomorphism. 
\qed\end{cor}
\subsection{The upper case}

 Suppose we are given \makered{full \KSs\ $K$ and $L$ and} universal \KS\ embeddings $e\type K \rightarrow C$ and  $e'\type L \rightarrow D$, where $L$ has a top element (which we write as $1$). Then we say
 that  a  function $f\type K\rightarrow L$ is \emph{strongly nonexpansive} if\display 
 \[f(a +_r b)\ \leq\ f(a) +_r \norm{b}_K\cdot 1\]
holds for all $a,b \in K$ and $r \in [0,1]$,  
and  that a  function $g\type C\rightarrow D$ is \emph{strongly nonexpansive} if\display 
 \[g(x + y)\ \leq\ g(x) + \norm{y}_K\cdot e'(1)\]
holds for all $x,y \in C$; note that this last definition is consistent with the corresponding definitions in previous sections.
\makered{Also, a homogeneous function $g\type C\rightarrow D$  is
  strongly nonexpansive iff it is when considered as a function between \KSs.}




  
  We write  $\cL_{\rm{sup}}^{\rm{sne}}(C,D)$
for the collection of all Scott-continuous
superlinear strongly nonexpansive functions $g\colon C \rightarrow D$; this collection forms a  \KS\ with the pointwise structure, and  a  \KS\ meet-semilattice if $D$ is a d-cone meet-semilattice.
 We further write  $\cL_{\rm{sup}}^{\rm{sne}}(K,L)$
for the collection of all Scott-continuous
superlinear strongly nonexpansive functionals $f\colon K \rightarrow L$; this collection forms a \KS\ with the pointwise structure, and a \KS\ meet-semilattice if $L$ is one.

Now suppose, additionally, that $L$ is a \KS\ meet-semilattice and $D$ is a d-cone meet-semilattice (when $e'$ is automatically a \KS\ semilattice morphism and so $e'\type L \rightarrow D$ is  a universal \KS\   semilattice embedding). 
Then the \KS\ semilattice  isomorphism of Proposition~\ref{fund-unit} restricts to a \KS\ semilattice isomorphism\display
\[\cL_{\rm{sup}}^{\rm{sne}}(K,L) \cong \cL_{\rm{sup}}^{\rm{sne}}(C,D)\]
which sends $f$ to $g \eqdef \ov{f\circ e'}$. Regarding superadditivity, $f$ is superadditive iff $f\circ e'$ is iff (by Theorem~\ref{prop:KSembed}) $\ov{f \circ e'}$ is.
Also
$g$ is strongly nonexpansive iff
$g(x +_r y)\ \leq\ g(x) +_r \norm{y}_K\cdot e'(1)$ for all $x,y \in C$ and $r \in [0,1]$,
iff, by the binary induction principle for universal embeddings,  
$g(e(a) +_r e(b))\ \leq\ g(e(a)) +_r \norm{e(b)}_K\cdot e'(1)$ for $a,b \in K$, $r \in [0,1]$, 
iff $e'(f(a +_r b))\ \leq\ e'(f(a) +_r \norm{b}_K\cdot 1)$, for $a,b \in K$, $r \in [0,1]$  (as $g\circ e = e'\circ f$),
iff $f$ is strongly nonexpansive.

The \KSs\  $\I$, $\oRp$, $\cL_{\leq 1} P$ and $\cL P$ ($P$ a dcpo) are
all \KS\ meet-semilattices and so the inclusions $\I \subseteq \oRp$
and $\cL_{\leq 1} P \subseteq \cL P$ are universal \KS\ semilattice
embeddings. So, in particular, for any dcpo $P$ we obtain a \KS\ semilattice isomorphism\display
\[\cL_{\rm{sup}}^{\rm{sne}}(\cL_{\leq 1}P,\I) \cong \cL_{\rm{sup}}^{\rm{sne}}(\cL P,\oRp)\] 
and for any dcpos $P$ and $Q$ we obtain  a \KS\ semilattice isomorphism\display
\[\cL_{\rm{sup}}^{\rm{sne}}(\cL_{\leq 1} Q, \cL_{\leq 1}P) \cong \cL_{\rm{sup}}^{\rm{sne}}(\cL Q, \cL P)\] 
As immediate consequences of Corolleries~\ref{th:upperdomain} and~\ref{th:upperdomainPred} we then obtain\display
\begin{cor}\label{th:upperdomain-unit}
Let $P$ be a domain.
Then  we have a  \KS\  meet-semilattice isomorphism
\[\Lambda_P\colon \cS\mathcal{V}_{\leq 1} P \;\cong\;
\cL_{\rm{sup}}^{\rm{sne}}(\cL_{\leq 1} P,\I)\]
It is given by\display
\[ \Lambda_P(X)(f) \eqdef \inf_{\nu \in X} \int \!f\,d\nu \tag*{\qEd}\]
\end{cor}

\begin{cor} \label{th:upperdomainPred-unit}
Let $P$ be a dcpo and let $Q$ be a domain. To every state
transformer $s\colon P \to \cS \cV_{\leq 1}Q$ we can assign  a predicate transformer $\PT_{P,Q}(s)\colon
\cL_{\leq 1}Q\to \cL_{\leq 1}P$  by\display
\[\PT_{P,Q}(s)(g)(x) \eqdef \inf_{\nu \in s(x)}\int\! g \, d\nu\quad
(g\in \cL_{\leq 1}Q, x \in P)\]
The predicate transformer $\PT_{P,Q}(s)$ is superlinear and strongly
nonexpansive.  
The assignment $\PT_{P,Q}$ is a \KS\ meet-semilattice isomorphism
 \[(\cS \cV_{\leq 1} Q)^P \cong \cL_{\rm{sup}}^{\rm{sne}}(\cL_{\leq
     1}Q,\cL_{\leq 1}P) \tag*{\qEd}\]
\end{cor}

Finally we show that, as in previous cases, the strong
nonexpansiveness condition can be simplified for homogeneous
functions.

\begin{fact} Let $K \xrightarrow{e} C$ and $L \xrightarrow{e'} D$ be universal embeddings where both $K$ and $L$ have top elements and where $\norm{\hspace{0.5pt}\mbox{-}\hspace{0.5pt}}_K$ is a norm on $C$. Then:
\begin{enumerate}
\item A homogeneous function $f: K  \rightarrow L$ is strongly nonexpansive iff $f(x +_r 1) \leq f(x) +_r 1$, for all $x \in K$ and $r \in [0,1]$
\item A homogeneous function $g: C \rightarrow D$ is strongly nonexpansive iff $g(x + 1) \leq g(x) + 1$, for all $x \in C$
\end{enumerate}
\end{fact}
\begin{proof} We assume, without loss of generality, that the embeddings are inclusions.
\begin{enumerate}
\item We have to show that $f(a +_r b)\ \leq\ f(a) +_r \norm{b}_K\cdot 1$ for all $a,b \in K$ and $r \in [0,1]$. If  $\norm{b}_K = 0$  then, as 
$\norm{\hspace{0.5pt}\mbox{-}\hspace{0.5pt}}_K$ is a norm, $b = 0$. Otherwise, setting $s \eqdef \norm{b}_K$, we see by Proposition~\ref{prop:mink} that $b  = s\cdot c$ for some $c \in K$ and so that $b \leq s\cdot 1$. Then,  taking $t \eqdef 1 - (1 -r)s$, we note that $r \leq t$ and calculate\display
\[f(a +_r b) \leq f(a +_r s\cdot 1 ) =  f(r/t\cdot a +_t 1) \leq f(r/t\cdot a) +_t 1 = r/t \cdot f(a) +_t 1 = f(a) +_r \norm{b}_K\cdot 1\]
\item  We have to show that $g(x + y) \leq\ g(x) + \norm{y}_K\cdot 1$ for all $x,y \in C$. This is trivial if  $\norm{y}_K = \infty$. If it is $0$ then, as 
$\norm{\hspace{0.5pt}\mbox{-}\hspace{0.5pt}}_K$ is a norm, $y = 0$. Otherwise, setting $s \eqdef \norm{y}_K$ and noting that $s^{-1}\cdot y \in K$, we calculate\display
\begin{align*}g(x + y) =  r \cdot g(s^{-1}\!\cdot x \,+ \, s^{-1}\!\cdot y)
  &~\leq r \cdot g(s^{-1}\!\cdot x \,+\, 1) \\ &~\leq r\cdot(g(s^{-1}\cdot x) \,+\, 1) =  g(x) + \norm{y}_K\cdot1\tag*{\qEd}\end{align*}
\end{enumerate}
\def\popQED{}
\end{proof}
\subsection{The convex case}

\newcommand{\myd}{\mathrm{d}}
\newcommand{\myu}{\mathrm{u}}

We begin with some definitions. Given a full \KS\ $K$ and
a full \KS\ semilattice  $L$,  say that  a function $f\type K \rightarrow L$ is \emph{$\subseteq$-sublinear} if it is homogeneous and, for all $a,b \in K$ and $r \in [0,1]$ we have\display
\[f(a +_r b) \subseteq f(a) +_r f(b)\]
Note that a function $g: C \rightarrow D$ from a d-cone to a d-cone semilattice is $\subseteq$-sublinear, as defined in Section~\ref{trans-convex},  iff it is when considered as a function from a \KS\ to a \KS\ semilattice.

Next, given \KSs\  $K$, $L$, and $M$ and two functions $\myd_L,\myu_L\type L \rightarrow M$, we say that a function $f: K \rightarrow L$ is \emph{medial (w.r.t.\ $\myd_L,\myu_L$)} if, for all $a,b \in K$ and $r \in[0,1]$ we have:
\[f_{\myd}(a +_r b)  \leq f_{\myd}(a) +_r f_{\myu}(b) \leq f_{\myu}(a + _rb)\]
where $f_{\myd} \eqdef  \myd_L \circ f$ and $f_{\myu} \eqdef  \myu_L \circ f$.
Similarly, given d-cones $C$, $D$, and $E$ and two  functions $\myd_D,\myu_D\type D \rightarrow E$, we say that a function $g: D \rightarrow E$ is \emph{medial (w.r.t.\ $\myd_D,\myu_D$)} if, for all $x,y \in C$, we have:
\[g_{\myd}(x + y)  \leq g_{\myd}(x) + g_{\myu}(y) \leq g_{\myu}(x + y)\]
where $g_{\myd} \eqdef  \myd_D \circ g$ and $g_{\myu} \eqdef  \myu_D \circ g$;  this is equivalent to it being medial when considered as a function between \KSs, provided that it is homogeneous, as  are $\myd_D$, and $\myu_D$.

Now we suppose  given: 
\begin{itemize}
\item full \KSs\ $K$ and $M$ and a full \KS\ semilattice $L$,

\item d-cones $C$ and $E$ and a d-cone semilattice $D$,
\item universal \KS\ embeddings  $K \xrightarrow{e_1} C$ and $M \xrightarrow{e_3} E$ and a universal \KS\ semilattice embedding $L \xrightarrow{e_2}  D$, and
\item  Scott continuous homogeneous functions $\myd_L,\myu_L\type L \rightarrow M$ and $\myd_D,\myu_D\type D \rightarrow E$ such that both the following two diagrams commute\display
\[\begin{diagram}
	{D} & \rTo^{\myd_D} & E\\
	 \uTo^{e_2} & & \uTo_{e_3} \\
	 {L} & \rTo_{\myd_L} & M
	\end{diagram}
\qquad
\begin{diagram}
	{D} & \rTo^{\myu_D} & E\\
	 \uTo^{e_2} & & \uTo_{e_3} \\
	 {L} & \rTo_{\myu_L} & M
	\end{diagram}	\]
\end{itemize}

Then, if $f\type K \rightarrow L$ and $g\type C \rightarrow D$ are Scott-continuous homogeneous functions such that $e_2 \circ f = g\circ e_1 $ then $f$ is $\subseteq$-sublinear iff $g$ is, and $f$ is medial iff $g$ is. The proofs are straightforward using the binary induction principle for universal \KS\ embeddings.

We now suppose further that

\begin{itemize}
\item $M$ is both a \KS\ meet- and join-semilattice, and $d_L, u_L$ are both \KS\ semilattice morphisms, with $M$ taken, accordingly, as  a \KS\ meet- or join-semilattice, and 
\item $E$ is both a d-cone meet-semilattice and  join-semilattice, and $d_D, u_D$ are both d-cone semilattice morphisms, with $M$ taken, accordingly, as  a d-cone meet- or  join-semilattice.
\end{itemize}

We write $\mathcal{L}^{\leq 1}_{\subseteq,\rm{med}}(C,D)$ for the set of Scott-continuous, $\subseteq$-sublinear, medial, nonexpansive functions from $D$ to $E$; equipped with the pointwise structure it forms a sub-\KS\ semilattice of $\mathcal{L}^{\leq 1}_\mathrm{hom}(C,D)$. We further write 
$\mathcal{L}_{\subseteq,\rm{med}}(K,L)$ for the set of Scott-continuous, $\subseteq$-sublinear, medial functions from $K$ to $L$; equipped with the pointwise structure it forms a sub-\KS\ semilattice of $\mathcal{L}_{\mathrm{hom}}(K,L)$.

As we have seen that $\subseteq$-sublinearity and mediality transfer along  the \KS\ semilattice  isomorphism of Proposition~\ref{fund-unit}, we now see that that, under our several suppositions, this isomorphism restricts to a \KS\ semilattice isomorphism\display
\[\mathcal{L}_{\subseteq,\rm{med}}(K,L) \cong \mathcal{L}^{\leq 1}_{\subseteq,\rm{med}}(C,D)\]

Let us now consider the particular case where $K \xrightarrow{e_1} C$ is the inclusion $\cL_{\leq 1} P \subseteq \cL P$, for some dcpo $P$, $L \xrightarrow{e_2} D$ is the inclusion $\cP \I \subseteq \cP \oRp$,  $M \xrightarrow{e_3} E$ is the inclusion $\I \subseteq  \oRp$, $\myd_{ \cP \oRp}(x) \eqdef  \un{x}$, $\myu_{ \cP \oRp}(x) \eqdef   \ov{x}$, and $\myd_{\cP \I}$ and $\myu_{\cP \I}$ are defined similarly. 

Then we already know that $e_1$ is a universal \KS\ embedding and it is evident that $e_3$ is too; that $e_2$ is follows from Proposition~\ref{uni-Keg}, and so it is evidently a universal \KS\ semilattice embedding. 
It is then clear that all the above assumptions hold, and so we have  a \KS\ semilattice isomorphism\display
\[\mathcal{L}_{\subseteq,\rm{med}}(\cL_{\leq 1} P,\cP \I) \cong
\mathcal{L}^{\leq 1}_{\subseteq,\rm{med}}(\cL P, \cP \oRp)\]
where, on the left nonexpansiveness is defined relative to the Minkowski seminorms on $\cL P$ and $\oRp$. However these seminorms are the same as those considered before: we already know this for $\cL P$, and it is easy to show that for $\oRp$ (i.e., to show that $\norm{x}_{\cP \I} = \ov{x}$).  As an immediate corollary of Corollary~\ref{th:convexdomain} we then obtain\display
\begin{cor}\label{th:convexdomain-unit}
Let $P$ be a coherent domain. 
Then  we have a  \KS\ semilattice isomorphism
\[\Lambda_P\colon \cP\mathcal{V}_{\leq 1} P \;\cong\;
\mathcal{L}_{\subseteq,\rm{med}}(\cL_{\leq 1} P,\cP \I)\]
It is given by\display
\[ \Lambda_P(X)(f) \eqdef [\inf_{\nu \in X} \int \!f\,d\nu\, ,
\sup_{\nu \in X} \int \!f\,d\nu\, ] \tag*{\qEd}\]
\end{cor}   

To obtain a corresponding result for predicate transformers we change
the above framework slightly: instead of supposing that $K$ is a
full \KS, $C$ is a d-cone, and $K \xrightarrow{e_1} C$ is
a universal \KS\ embedding, we suppose that $K$ is a full \KS\
semilattice, $C$ is a d-cone semilattice, and $K \xrightarrow{e_1} C$
is a universal \KS\ semilattice embedding.

Then, for Scott-continuous homogeneous functions  $f\type K
\rightarrow L$, $g\type C \rightarrow D$ where $e_2 \circ f = g\circ
e_1 $, we additionally have that $f$ is $\subseteq$-monotone iff $g$
is. This is proved by the universal embedding induction principle, as
usual, but noting that $\subseteq$-monotonicity can be expressed
equationally: for example $g$ is $\subseteq$-monotone if, and only if,
for all $x,y \in C$, $g(x) \cup g(x \cup y) = g(x\cup y)$.

We now write $\mathcal{L}^{\leq 1}_{\rm{mon},\subseteq,\rm{med}}(C,D)$ for the set of Scott-continuous, $\subseteq$-monotone, $\subseteq$-sublinear, medial, nonexpansive functions from $D$ to $E$; equipped with the pointwise structure it forms a sub-\KS\ semilattice of $\mathcal{L}^{\leq 1}_\mathrm{hom}(C,D)$. We further write 
$\mathcal{L}_{\rm{mon},\subseteq,\rm{med}}(K,L)$ for the set of Scott-continuous, $\subseteq$-monotone, $\subseteq$-sublinear, medial functions from $K$ to $L$; equipped with the pointwise structure it forms a sub-\KS\ semilattice of $\mathcal{L}_{\mathrm{hom}}(K,L)$.

As we have seen that $\subseteq$-monotonicity, $\subseteq$-sublinearity, and mediality  transfer along  the \KS\ semilattice  isomorphism of Proposition~\ref{fund-unit}, we now see that that, under our several suppositions, this isomorphism restricts to a \KS\ semilattice isomorphism\display
\[\mathcal{L}^{\leq 1}_{\rm{mon},\subseteq,\rm{med}}(C,D) \cong \mathcal{L}_{\rm{mon},\subseteq,\rm{med}}(K,L)\]
%


To apply this result,  we note that, for any dcpo $P$, the inclusion $\cP \I^P \subseteq \cP \oRp^P$ is a universal \KS\ semilattice embedding (for a proof again use Proposition~\ref{uni-Keg}, now following the same lines as the proof that the inclusion $\cL_{\leq 1}P \subseteq \cL P$ is universal) and that the Minkowski seminorm on $\cP \oRp^P$ is the same as the norm defined  before, i.e., that $\norm{f}_{\cP \oRp^P} = \sup_{x \in P}\ov{f(x)}$. 

Now consider the particular case where $K \xrightarrow{e_1} C$ is the inclusion $\cP \I^Q \subseteq \cP \oRp^Q$, for some dcpo $Q$, $L \xrightarrow{e_2} D$ is the inclusion $\cP \I^P \subseteq \cP \oRp^P$, for some dcpo $P$,  $M \xrightarrow{e_3} E$ is the inclusion $\I^P \subseteq  \oRp^P$, $\myd_{ \cP \oRp^P}(f) \eqdef  \un{f}$, $\myu_{ \cP \oRp^P}(f) \eqdef   \ov{f}$, and $\myd_{\cP \I^P}$ and $\myu_{\cP \I^P}$ are defined similarly. Then   all the above assumptions hold, and so we have  a \KS\ semilattice isomorphism\display
\[\mathcal{L}^{\leq 1}_{\rm{mon}, \subseteq,\rm{med}}(\cP \oRp^Q,\cP \oRp^P) \cong \mathcal{L}_{\rm{mon}, \subseteq,\rm{med}}(\cP \I^Q, \cP \I^P)\]
and then as an immediate corollary of Corollary~\ref{th:convexdomainPred} we obtain\display

\begin{cor} \label{th:convexdomainPred-unit}
Let $P$ be a dcpo and let $Q$ be  a coherent domain. To every state
transformer $s\colon P \to \cP \cV_{\leq 1}Q$ we can assign  a predicate transformer $\PT_{P,Q}(s)\colon
\cP\I ^Q \to \cP\I ^P$  by\display
\[\PT_{P,Q}(s)(g)(x) \eqdef [\inf_{\nu \in s(x)}\int\! \un{g} \, d\nu,
\sup_{\nu \in s(x)}\int\! \ov{g} \, d\nu\ ]        \quad (g\in \cP\I ^Q, x \in P)\]
The predicate transformer $\PT_{P,Q}(s)$ is  
$\subseteq$-monotone, $\subseteq$-sublinear, and medial.
The assignment $\PT_{P,Q}$ is a \KS\ semilattice isomorphism 
\[(\cP\cV_{\leq 1} Q)^P \cong
  \cL_{\mathrm{mon},\subseteq,\rm{med}}(\cP\I^Q,\cP\I^P) \tag*{\qEd}\]
\end{cor}







\section*{Acknowledgements}

We are very grateful to Jean  Goubault-Larrecq for his many helpful
comments and suggestions. 
\bibliographystyle{alpha}
\vspace{-0.5cm}
\bibliography{paper}

\appendix

\section{The other distributive law}\label{appbad}

We consider two equational theories: $\mathrm{B}$, for barycentric
algebras, and $\mathrm{S}$, for semilattices. The first has binary
operation symbols $+_r$ ($r \in [0,1]$) and axioms the equations for
barycentric algebras (\ref{B1}), (\ref{B2}), (\ref{SC}), and
(\ref{SA}) given in Section~\ref{sec:ba}; the second has a single
binary operation symbol $\nonor$ and 
associativity, commutativity, and idempotency axioms. For any $r \in [0,1]$ we
write $r'$ for $1-r$. 

 We write 
$t_1 = u_1, \ldots, t_1 = u_1 \vdash_T t = u $
for a given equational theory $T$ to mean that $t = u$ follows by
equational reasoning from the theory $T$ and  the equations  $t_1 =
u_1, \ldots, t_n = u_n$. 
We need a proof-theoretic version of another lemma of Neumann:
\cite[Lemma 3]{Neu} (with corrected bounds).
 We first show a lemma that can also be derived from general results
 due to Sokolova and Woracek \cite[Theorem 4.4 and Example
 4.6]{SW}\display  
 
\begin{lem} \label{bcong}
Let $\sim$ be a congruence on the barycentric algebra $[0,1]$. Then if
two elements of the open interval $]0,1[$ are congruent, so are any
other two. 
\end{lem}
\begin{proof}
Let $\sim$ be  such a congruence and suppose that we have $r \sim s$
with $0 < r < s < 1$. Set $\alpha = r/s < 1$. 
Then $\alpha r \sim \alpha s = r \sim s$ hence $\alpha r \sim s$.
Repeating the argument yields $\alpha^n r \sim s$, for any $n \geq 0$,
and so we get arbitrarily close to $0$ with elements congruent to
$s$. 

In the other direction, we can define a `symmetric' congruence
relation $\sim'$, setting $ p \sim' q$ to hold if, and only if,  $p' \sim q'$, as the map $p
\mapsto p'$ is an (involutive) automorphism of $[0,1]$.  We then have
$s' \sim' r'$ and $0 < s' < r' <1$. So, by the above argument, we can
get arbitrarily close to $0$ with elements in the symmetric congruence
relation with $r'$, and so arbitrarily close to $1$ with elements
congruent to $r$.  

As congruence classes are convex, we then see that any two elements of
the open interval $]0,1[$ are congruent. 
\end{proof}

\begin{lem} \label{NP} Let $T$ be an equational theory extending
  $\mathrm{B}$. Then for any terms $t$ and $u$,  and any $0 < r  < s <
  1$ and $0 < p< q < 1$ we have\display 
\[t +_r u = t +_s u   \vdash_T \; t +_{p} u = t +_{q} u  \]
\end{lem}
\begin{proof} One fixes $r$ and $s$ with $0 < r  < s < 1$ and then
  applies Lemma~\ref{bcong} to the congruence $\sim$ where 
$p \sim q$ iff $t +_r u = t +_s u   \vdash_T \; t +_{p} u = t +_{q} u$.
\end{proof}

The equational theory $\rm{BSD'}$ 
has axioms those of  $\mathrm{B}$ and $\mathrm{S}$ together with
equations 
 \[x \nonor (y +_r z) = (x \nonor y) +_r (x \nonor z)  \qquad (r \in
 ]0,1[) \tag*{$\rm(D')$}\] 
stating that $\nonor$ distributes over each of the $+_r$.
Recall that a \emph{join-distributive bi-semilattice~\cite{Rom80}} is
an algebra with two semilattice operations $\nonand$ and $\nonor$,
with $\nonor$ distributing over $\nonand$. 
 
\begin{thm} The equational theory $\mathrm{BSD'}$ 
is equivalent to that of   join-distributive bi-semilattices.
\end{thm}

\begin{proof} In the following we just write $t = u$ rather than
  $\vdash_{\mathrm{BSD'}} t = u$. 
Substituting $(y+_r z)$ for $x$ in $\rm{D'}$ (and using $\mathrm{S}$ and
re-using $\rm{D'}$) we get\display 
\begin{equation} \label{E1}
   ( y +_r  z) =   (( y +_r  z) \nonor  y) +_r ((y +_r z) \nonor  z) \\
               =   ( y+_r (y \nonor z))+_r ((y\nonor z)+_{r} z)
\end{equation}
for all $r \in ]0,1[$. Now, substituting $y \nonor  z$ for $z$ (and using
$\mathrm{S}$ and then $\mathrm{B}$) we get\display 
\[ y +_r (y \nonor  z) = ( y+_r (y \nonor z))+_r ((y\nonor z)+_{r} (y
\nonor  z)) = y +_{r^2} (y \nonor  z)\] 
for all $r\in ]0,1[$. So by Lemma~\ref{NP} we obtain\display
\begin{equation} \label{E2}
 y+_r (y\nonor z)  =   y+_s (y\nonor z)
\end{equation}
for all $r, s\in ]0,1[$. But now, for all $r\in \, ]0,1[$, we have\display
\begin{align*}
( y+_r z) =&~( y+_r (y \nonor z))+_r ((y\nonor z) +_{r} z)    \tag{by Equation~(\ref{E1})}\\
                             =&~(y+_{r'} (y \nonor z))+_r ( (y\nonor z) +_{r'}  z)   \tag{by Equation~(\ref{E2}), and using  $\mathrm{B}$}\\
                             =&~(y+_{r} (y \nonor z))+_{r'} ((y\nonor z) +_{r} z)            \tag{by the entropic identity}\\
                             =&~(y+_{r'} (y \nonor z))+_{r'} ((y\nonor z) +_{r'} z)       \tag{by Equation~(\ref{E2}), and using  $\mathrm{B}$}\\
                             =&~(y +_{r'} z)  \tag{by Equation~(\ref{E1}), substituting $r'$ for $r$}
\end{align*}
and so we can apply Lemma~\ref{NP} again, and obtain\display
\[ y +_r  z =  y +_s z\]
for all $r, s \in\, ]0,1[$.

But then we have an equivalence of our theory with that of  join-distributive bi-semilattices, where join and meet are $\nonor$  and $\nonand$. 
To translate from the theory of join-distributive bi-semilattices into our theory one translates $\nonor$  as $x_0 \nonor x_1$ and $\nonand$ as (e.g.) $x_0 +_{\half} x_1$. In the other direction, one translates $\nonor$  as $x_0 \nonor x_1$, $+_r$ as $x_0 \nonand x_1$, for all $r \in ]0,1[$, and $+_0$ and $+_1$ as $x_0$ and $x_1$, respectively. 
\end{proof}

We next consider a weaker theory: we drop the idempotence of
$\cup$. Let $C$ be the theory of commutative semigroups, that is of
algebras with an associative and commutative multiplication operation
(having dropped idempotence, the  change to a multiplicative notation
is natural). The equational theory 
$\CCSA$ of \emph{convex commutative semigroup algebras}
has as axioms those of
$\mathrm{B}$ and $\mathrm{C}$ together with the following distributive
laws, stating that multiplication distributes over the $+_r$\display 
\[x(y +_r z) = xy  +_r  xz \qquad (r \in ]0,1[)\]
The real interval $[0,1]$ provides an example 
convex commutative semigroup algebra
with the usual barycentric operations and multiplication. 

Let $\cDw$ be the finite probability distributions monad. We regard
$\cDw X$ as consisting of convex
combinations $\sum_{i = 1,m} p_ix_i$
of elements of $X$; with the evident barycentric operations it is the
free barycentric algebra over $X$. Then the free 
convex commutative semigroup algebra 
over a given commutative semigroup $M$ is provided
by the barycentric algebra  $\cDw M$, with the following
multiplication\display 
\[(\sum_{i \,=\, 1,\dots,m} p_i x_i)(  \sum_{i \,=\, 1,\dots,n} q_j y_j) =
  \sum_{\mbox{$\substack{\scriptsize i \,=\, 1,\dots,m\\ j \,=\, 1,\dots,n}$}}\!\!(p_iq_j) x_i y_j \]
and with unit the inclusion (this construction is a variant of the
standard group algebra construction). In particular, writing $\cMwp$
for the finite non-empty multisets monad, we see that  $\cDw \cMwp X $
is the free 
convex commutative semigroup algebra 
over $X$, as $\cMwp$ is the free
commutative semigroup monad.  

We regard $\cDw \cMwp X $ as consisting of all polynomials with no
constants, with variables in $X$, and with coefficients in $[0,1]$
adding up to $1$. In other words, it consists of all convex
combinations of non-trivial polynomials with variables in $X$. The
barycentric operations are the evident convex 
combinations of such
polynomials, and multiplication is the usual polynomial
multiplication. The unit $\eta\colon x \to \cDw \cMwp X$ is the
inclusion, and the extension  $\ov{f}\colon \cDw \cMwp X \to A $ to a
$\CCSA$-homomorphism  
of a map $f$ from $X$  to a 
convex commutative semigroup algebra 
 $A$  assigns to any polynomial $p(x_1,\ldots,x_n)$ in $\cDw
\cMwp X$ its value $p(f(x_1),\ldots,f(x_n)) \in A$ as obtained using
the $\CCSA$ operations of $A$. 

A convex commutative semigroup algebra 
$A$ is  \emph{complete} if,  for all
$\CCSA$-terms 
$t$ and $u$ we have\display 
\[A \models t = u \quad \Rightarrow \quad  \vdash_{\CCSA} t = u\]
This holds if, and only if, distinct polynomials in $\cDw \cMwp X $
can be separated by elements of $A$, that is, if for any such
$p(x_1,\ldots,x_n) \neq q(x_1,\ldots,x_n)$, with variables in
$x_1,\ldots,x_n$, there are $a_1,\ldots a_n \in A$ such that
$p(a_1,\ldots,a_n) \neq q(a_1,\ldots, a_n)$. For example, $[0,1]$ is
complete in this sense.

We now focus on the convex semigroup algebra 
$\cDw\cPwp X$. We need two lemmas.
\begin{lem} \label{BCDC} Let $X$ be a set with at least two
  elements. Then  $\cDw\cPwp X$ is complete. 
\end{lem}
\begin{proof} Let $p(x_1,\ldots,x_n)$, $q(x_1,\ldots,x_n)$ be distinct
  polynomials in  $\cDw \cMwp X $, and choose $r_1,\ldots,r_n$ in
  $[0,1]$ separating them. Choose two distinct elements $y$, $z$ in
  $X$. We can define a semigroup homomorphism $h\colon \cPwp X \to
  [0,1]$ by\display   
\[h(u) = \left \{ \begin{array}{lcl}
                          1 & (u = \{y\})\\
                          0  & (\mbox{otherwise})
                       \end{array} \right .\] 
Then $h$ has an extension to a 
$\CCSA$-homomorphism $\ov{h}\colon
\cDw \cPwp X \to [0,1]$, and, taking $a_i \in  \cDw \cPwp X $ to be 
$\{y\} +_{r_i} \{z\}$, for $i = 1,n$, and noting that $\ov{h}(a_i) =
h(\{y\}) +_{r_i} h(\{z\}) = r_i$, we see that\display
\[\ov{h}(p(a_1, \ldots, a_n)) = p(\ov{h}(a_1), \ldots, \ov{h}(a_n)) = p(r_1,\ldots,r_n) \neq q(r_1,\ldots,r_n) = \ov{h}(q(a_1, \ldots, a_n))\]
and so that $a_1,\ldots,a_n$ are elements of $ \cDw \cPwp X $
separating $p$ and $q$. 
\end{proof}

\begin{lem} \label{notfree} Let  $X$ be a nonempty set and let $T$ be a subtheory of $\CCSA$. Then
  $\cDw\cPwp X$ is not the free $T$-algebra
  over $X$.
\end{lem}
\begin{proof}
First note that, for any $p \in \cDw\cMwp X$,  we have $p \neq  pp$. 
Then, if  $\cDw\cPwp X$ were the free 
$T$-algebra 
over $X$, there would be a $T$-algebra homomorphism $h\colon \cDw\cPwp X \to \cDw\cMwp X$, as $\cDw\cMwp X$ is a $\CCSA$-algebra and so a $T$-algebra. Choosing any $y \in \cPwp X$, we  would then find  
$h(y) = h(yy) = h(y)h(y)$, a contradiction.
\end{proof} 

So, in particular, for non-empty $X$,    $\cDw\cPwp X$ is, as may be
expected, not the free $\CCSA$-algebra. 
We can now prove\display

\begin{thm}
Let $X$ be a set with at least two elements.
 Then $\cDw \cPwp X$ is
not the free $T$-algebra over $X$ for any equational theory $T$ with the same 
signature as  that of
$\CCSA$.  
\end{thm}
\begin{proof}
Suppose, for the sake of contradiction that $ \cDw\cPwp X$ is the free $T$-algebra over $X$  for an equational theory $T$ with the same signature as $\CCSA$. Then, in particular, it is a $T$-algebra, and so all $T$-equations hold in it. But, by Lemma~\ref{BCDC} all equations holding in it are in $\CCSA$. So $ \cDw\cPwp X$ is the free algebra for a subtheory of $\CCSA$. However, by Lemma~\ref{notfree}, that cannot be the case.
\end{proof}

So we have shown that there is no algebraic (i.e., equational) account
of the natural random set algebras $ \cDw\cPwp X$. It may even be that   
$\cDw\!\circ\!\cPwp$ admits no monadic structure.

\section{A counterexample}\label{subsec:app}

\newcommand{\thing}{f}
\newcommand{\alphag}{g}
\newcommand{\red}[1]{#1}

\newcommand{\one }{\mathrm{k}_1}

Referring to Lemma \ref{lem:clever} and the preceding
discussion,  we give an example of a continuous \KS\ $E$ whose scalar
multiplication does 
not preserve the way-below relation $\ll_E$. That is, there are
elements $a\ll_E b$ in $E$ such that $ra \not\ll_E rb$ for some
$r<1$. 

We consider the half-open unit interval $]0,1]$
with its upper (= Scott) topology; the open
subsets are the half-open intervals $]r,1]$, $0\leq r\leq 1$. 
Let $\cL$ be the continuous d-cone of continuous functions $\thing\colon
]0,1]\to\oRp$ (in classical analysis one would have said that the
$\thing\in\cL$ are the monotone increasing lower semicontinuous
functions). Such a function $\thing$ has a greatest value, namely
$\thing(1)$. 
We write $\cL_{\leq 1}$ and $\cL_{\leq 2}$ for 
the Scott-closed, hence continuous, \KSs\ of functions $\thing\in\cL$
such that $\thing(1)\leq 1$ and $\thing(1)\leq 2$, respectively. 
We write $\one$  for the constant function 1 on $]0,1]$.

Now let $E$ be the collection of those $\thing\in\cL_{\leq 2}$ which can be
represented as a  
convex combination $f = q\cdot 2\one + (1-q)\cdot h\ (0\leq q \leq 1)$ of the constant function $2\one$  with some
$h \in\cL_{\leq 1}$. This provides our counterexample. 

We first claim that $E$ forms a sub-\KS\ of $\cL_{\leq 2}$. It is clearly
convex and contains the least element of $\cL_{\leq 2}$. To show $E$ is a
sub-dcpo of $\cL_{\leq 2}$ with the inherited ordering, we set 
\[
\red{E' \eqdef \{\thing \in \cL_{\leq 2}\mid 2\thing(1)\one    \leq   \thing + 2\one\}}
\] 
and, \red{noting that} 
$E'$ forms a sub-dcpo, show 
that $E$ and $E'$
coincide. 
%

It is clear that $E \subseteq E'$. Conversely,
suppose $\thing \in E'$ \red{so that $2\thing(1)\one    \leq   \thing + 2\one$}.
If $\thing(1) = 2$, then
$\thing =2\one$  and so $\thing\in E$. If $\thing(1)\leq 1$, then trivially
$\thing\in\cL_{\leq 1}\subseteq E$. Thus we can suppose that $1<\thing(1)<2$. 
We can rewrite the condition for membership in $E'$ as $2(\thing(1)-1)\one  \leq\thing$
and we let
$\widehat\thing = 
\thing-2(\thing(1)-1)\one$. Since $\widehat\thing(1) = \thing(1) - 2(\thing(1)-1) = 
2-\thing(1) \red{\,(\neq 0)\,}$,
we have $\frac{\widehat\thing(1)}{2-\thing(1)} = 1$, so
that  $\frac{\widehat\thing}{2-\thing(1)}\in\cL_{\leq 1}$. Letting $q = \thing(1)-1$
we have $1-q =2-\thing(1)$ and $0<q<1$\makeblue{, since $1<\thing(1)<2$,}   and $\thing$ becomes a convex
combination of the constant function $\red{2}\one$  and the function
$\frac{\widehat\thing}{2-\thing(1)}\in\cL_{\leq 1}$, namely 
$\thing =  \red{2(\thing(1)-1)\one  +\widehat \thing = q\cdot 2\one  + (1-q)\cdot\frac{\widehat\thing}{1-q}}$, 
so that $\thing\in E$.

 
 
 \red{We remark that $E$ is not a full \KS. For example, writing
 $\chi_s$ for the characteristic function of the interval
$]s,1]$  ($0 < s < 1$), we have $\chi_s \leq \half\cdot (2\one)$. However,
  we cannot have $\chi_s = \half\cdot\thing$ for any $\thing\in E$, as we
would then have $f(1) = 2$ and so, by the defining property of $E'$, $f = 2\one$.
But, as  $\chi_s = \half\cdot\thing$, we have $f(r) = 0$ for any  $0 < r < s$. 
}

We next claim that $E$ is continuous. For this, it is enough to show
that each element $\thing$ is the lub of a directed set of elements
$\ll_E$ below it. 
In case $\thing(1) \leq 1$,  any element in $\cL_{\leq 2}$
below it is also in $E$, and we use the continuity of $\cL_{\leq 2}$ to get a
directed set of elements $\ll_{\cL_{\leq 2}}$ below it, and so $\ll_E$ 
below it, as $E$ is a sub-dcpo of $\cL_{\leq 2}$.

\red{Otherwise}  
$\thing(1)>1$, \red{and we can again set}
$\widehat \thing = \thing-2(\thing(1)-1)\one$.
Since $\widehat
\thing(1)= 
 2-\thing(1) <1$, we have $\widehat\thing\in\cL_{\leq 1}$. Let 
\[\rho_{p,\alphag}= 2p\one  +\alphag\]
For 
$\red{0} <p<\thing(1)-1$ and $\alphag \ll_{\cL_{\leq 1}} \widehat\thing$, we obtain a 
family of functions which clearly is directed and has $2(\thing(1)
-1)\one  + \widehat \thing = \thing$ as its least upper bound. Each of these
functions belongs to $E$, since it can be written as a convex
combination $p\cdot 2\one  + (1-p)\cdot\frac{\alphag}{1-p}$ and since  
$\frac{\alphag}{1-p}\in\cL_{\leq 1}$ \red{(the latter as 
$\alphag(1) \leq \widehat\thing(1) =  2 - \thing(1) < 1 - p$)}.



It remains to show that $\rho_{p,\alphag}\ll_E \thing$.
For this, let $(\thing_i)_i$ be a directed family in $E$ such that
$\thing\leq\dsup_{i}\thing_i$. Then 
$\thing(1)\leq \dsup_{i}\thing_i(1)$ so that there is an $i_0$ such that 
$p< \thing_{i_0}(1)-1$, \red{since $p < f(1) -1 $}.
 We now restrict our attention to
the indices \makeblue{$i$ 
such that $\thing_i \geq \thing_{i_0}$.}  
Since these $\thing_{i}$ belong to $E=E'$,  
they
satisfy $2(\thing_{i}(1) -1)\one  \leq \thing_{i}$. 
It follows that \red{$2p\one \leq 2(\thing_{i}(1) -1)\one  \leq  \thing_{i}$,}
\makeblue{
 that is $0\leq f_i -2p\one\in\cL$. As  
 $\thing-2p\one  \leq(\dsup_{i}\thing_i) - 2p\one  = \dsup_{i}(\thing_i -2p\one)$, we then have $\thing-2p\one  \leq_{\cL} \dsup_{i}(\thing_i -2p\one)$. Since $\alphag\ll_{\cL_{\leq 1}} \widehat\thing = f-2(f(1)-1)\one\leq_{\cL} f-2p\one$, we also have $\alphag\ll_\cL f-2p\one$, and so 
$\alphag\leq \thing_{i_1} -2p\one$  for some $i_1\geq i_0$. Thus 
$\rho_{p,\alphag} \leq 2p\one+\alphag\leq \thing_{i_1}$.}

Now that we have shown $E$ to be a continuous Kegelspitze, it only remains to
see that, as claimed, scalar multiplication does not
preserve $\ll_E$. One the one hand, from the above discussion we have
$\one  \ll_E 2\one$
\red{(for, setting $f = 2\one$  and $p = 1/2$, we see that $\widehat f = \perp$ and $p < f(1) - 1$, and then that $\one  = \rho_{1/2, \perp}$)}.
However, on the other hand, we have  $\half\cdot\one\not\ll_E\one$   \red{(for $\one = \dsup_{n > 0} \chi_{2^{-n}}$, but we have  $\half\cdot\one  \leq \chi_{2^{-n}}$ for no $n > 0$)}.

\end{document}